\documentclass[sn-mathphys-num]{sn-jnl}
\usepackage{amsmath,amsxtra,amssymb,amsthm,mathtools}
\usepackage{lipsum}
\usepackage{xcolor,color}
\usepackage{tcolorbox}
\usepackage{quantikz}
\usepackage{graphicx,epstopdf}
\usepackage{caption}
\usepackage{subcaption}
\usepackage{tikz}
\usepackage{circuitikz}
\usepackage{physics}
\usepackage{cancel}
\numberwithin{equation}{section}
\hypersetup{
    colorlinks,
    citecolor=black,
    filecolor=black,
    linkcolor=black,
    urlcolor=black
}

\newcommand{\kla}{\left ( }
\newcommand{\mer}{\right ) }

\newcommand{\for}{\begin{eqnarray*}}
\renewcommand{\mel}{\end{eqnarray*}}
\def\fr{\begin{align*}}

\newcommand{\kl}{\pl \le \pl}
\newcommand{\gl}{\pl \ge \pl}

\newcommand{\lel}{\pl = \pl}

\newcommand{\wt}{\widetilde}

\newcommand{\ten}{\otimes}

\newcommand{\pl}{\hspace{.1cm}}

\newcommand{\ran}{\rangle}
\newcommand{\lan}{\langle}

\newcommand{\al}{\alpha}

\newcommand{\la}{\lambda}

\newcommand{\adm}{{\mathrm{ad}}}

\newcommand{\E}{{\mathcal E}}

\newcommand{\Bb}{{\mathbb{B}}}

\newcommand{\D}{{\mathcal D}}

\newcommand{\Ha}{{\mathcal H}}

\newcommand{\U}{{\mathcal U}}
\newcommand{\zd}[1]{{\color{blue}[ZD: #1]}}
\newcommand{\mj}[1]{{\color{magenta}[MJ: #1]}}
\newcommand{\pw}[1]{{\color{red}[PW: #1]}}
\newcommand{\ps}[1]{{\color{purple}[PS: #1]}}

\newcommand{\mc}{\mathcal}
\newcommand{\mb}{\mathbb}

\newcommand{\ez}{{\mathbb E}\vspace{0.1cm}}

\newtheorem{lemma}{Lemma}[section]
\newtheorem{prop}[lemma]{Proposition}
\newtheorem{theorem}[lemma]{Theorem}
\newtheorem{cor}[lemma]{Corollary}
\newtheorem{rmk}[lemma]{Remark}

\newtheorem{rem}[lemma]{Remark}

\newcommand{\re}{\begin{rem}\rm}
  \newcommand{\mar}{\end{rem}}
\newtheorem{exam}[lemma]{Example}
\newtheorem{defi}[lemma]{Definition}
\newcommand{\qd}{\end{proof}\vspace{0.5ex}}

\newcommand{\prf}{\begin{proof}[\bf Proof:]}

\newcommand{\xspace}{\hbox{\kern-2.5pt}}

\newcommand{\mm}{\mathrm{M}}

\DeclareMathOperator{\uu}{\mathcal{U}}

\allowdisplaybreaks

\begin{document}
\title{Lower bound for simulation cost of open quantum systems: Lipschitz continuity approach}
\author[,1]{Zhiyan Ding\footnote{zding.m@berkeley.edu, part. supported by Quantum Systems Accelerator}\,\;} 
\author[,2]{Marius Junge\footnote{mjunge@illinois.edu, part. supported by DMS 2247114}\,\;}
\author[,3]{Philipp Schleich\footnote{philipps@cs.toronto.edu}\,\;}
\author[,4,5]{Peixue Wu\footnote{p33wu@uwaterloo.ca, supported by the Canada First Research
Excellence Fund (CFREF)}\,\;}
\affil[1]{\emph{Department of Mathematics, University of California, Berkeley, USA}}
\affil[2]{\emph{Department of Mathematics, University of Illinois at Urbana and Champaign, USA}}
\affil[3]{\emph{Department of Computer Science, University of Toronto, Canada}}
\affil[4]{\emph{Department of Applied Mathematics, University of Waterloo, Canada}}
\affil[5]{\emph{Institute for Quantum Computing, University of Waterloo, Canada}}

\abstract{Simulating quantum dynamics is one of the most promising applications of quantum computers. While the upper bound of the simulation cost has been extensively studied through various quantum algorithms, much less work has focused on establishing the lower bound, particularly for the simulation of open quantum system dynamics. In this work, we present a general framework to calculate the lower bound for simulating a broad class of quantum Markov semigroups. Given a fixed accessible unitary set, we introduce the concept of convexified circuit depth to quantify the quantum simulation cost and analyze the necessary circuit depth to construct a quantum simulation scheme that achieves a specific order. Our framework can be applied to both unital and non-unital quantum dynamics, and the tightness of our lower bound technique is illustrated by showing that the upper and lower bounds coincide in several examples.}

\maketitle
\tableofcontents

\section{Introduction}\label{sec:intro}
Simulating physics was one of the primary motivations for the development of quantum computers, as initially suggested by~\cite{manin_computable_1980,benioff_computer_1980,feynman_simulating_1982}. The feasibility of such computations was first discussed in the context of the physical Church-Turing theorem~\cite{deutsch1985quantum}, which focuses on simulating closed quantum dynamics. In this direction, Hamiltonian simulation has long been considered the most promising application of quantum computers. Significant efforts have been dedicated to uncovering the potential of quantum computers for Hamiltonian evolution and developing efficient quantum algorithms. Several efficient quantum algorithms have been proposed for Hamiltonian evolution~\cite{berry2007efficient,berry2014exponential,berry2015hamiltonian,
Low2019hamiltonian,GilyenSuLowEtAl2019,PhysRevX.11.011020,childs2012hamiltonian}.

Recently, there has been increasing interest in using quantum computers to capture and simulate more general quantum dynamics, specifically open quantum system dynamics. Unlike closed quantum systems, open quantum systems involve interactions with an environment that affect the system during its evolution. Understanding and capturing the effect of the environment not only helps us track the evolution of physical systems but also serves as a new tool for state preparation, such as thermalizing quantum systems~\cite{rall2022thermal,chifang2023quantum,ding2023thermal} and preparing ground states~\cite{ding2023single}. In addition, the concept of the ``environment'' can also represent the noise that occurs in current quantum devices. The ability to simulate noise may also serve as a verifier for existing experiments and as a testing platform for understanding various aspects of quantum devices. Along with increasing interest in open quantum systems, efficient simulation algorithms have been developed in this area. Similarly to the general argument for closed quantum systems,~\cite{kliesch2011dissipative} argue that quantum computers in a gate-based model representing pure states are sufficiently resourceful to efficiently simulate open-system dynamics as well. They refer to this as the dissipative quantum Church-Turing theorem. 
Subsequently, several quantum algorithms have been developed for simulating the Lindblad quantum master equation~\cite{kliesch2011dissipative,cleve_et_al:LIPIcs.ICALP.2017.17,li2022simulating,chen2023quantum,PRXQuantum.5.020332,pocrnic2024quantumsimulationlindbladiandynamics}.

Although the upper bound of quantum cost for simulating quantum dynamics has been well studied through various quantum algorithms under different oracle assumptions, much less work has been done on establishing the lower bound. Specifically, we ask:

\emph{Given a set of accessible unitary gates and a quantum dynamics $T_t$, what is the minimum number of gates (minimal circuit depth) needed to simulate $T_t$?}

\noindent Studying the lower bound of circuit depth can not only help us understand the efficiency of our quantum simulation algorithms but also illustrate the inherent limitations of a predetermined accessible unitary gate set. The latter is crucial for understanding the capabilities of practical quantum computers, given the often limited availability of unitary gates.

In this work, our main contribution is that we present a general framework to provide a lower bound on the simulation cost of a broad class of quantum dynamics given by $T_t = \exp(tL)$. Here, 
\begin{equation}\label{Lindblad generator: general}
    L(\rho)= -i[H,\rho] + \sum_{j=1}^m L_{V_j}(\rho),\quad L_{V}(\rho)=V\rho V^* - \frac{1}{2}(V^*V \rho+\rho V^*V)
\end{equation}
is a typical Lindblad generator \cite{lindblad1976generators}. With $V^*$ we denote the adjoint of $V$, which is also written as $V^{\dagger}$. In our analysis, In our analysis, we provide a general framework for calculating the necessary circuit depth of a quantum simulation scheme using a predetermined accessible unitary gate set that achieves a specific order. Detailed discussion of the main result can be found in the next section. Additionally, our results indicate a lower bound on the simulation cost with a fixed target time and precision. Our lower bound is established via a resource-dependent convexified gate length introduced in \cite{li2022wasserstein, araiza2023resource}, where the resource set is chosen according to the predetermined accessible unitary gate set. Different from the worst-case analysis studied in \cite{berry2015hamiltonian, childs2016efficient}, our framework works for any given Lindbladian generators and accessible unitary set. Thus our framework can be seen as an efficiency test for a given simulation scheme. We illustrate the tightness of our lower bound in some cases by establishing an upper bound for the simulation cost of some simple models.

The remaining part of the paper is organized as follows: The main result of this paper is summarized in Section \ref{sec:main_result}. We discuss related works in Section \ref{sec:related_works}. In Section \ref{sec:sc}, we introduce the definitions and notations used throughout this paper, focusing primarily on unital dynamics. We establish the upper bound of simulation cost for unital dynamics in Section \ref{sec:upper} and its lower bound in Section \ref{sec:lower_bound}. Section \ref{sec:non-unital} extends these definitions from Section \ref{sec:sc} to the non-unital case and establishes corresponding upper and lower bounds. Finally, Section \ref{sec:conclusion} provides a conclusion and summary of this paper.
\begin{figure}[ht]
    \centering\includegraphics[width=.6\textwidth]{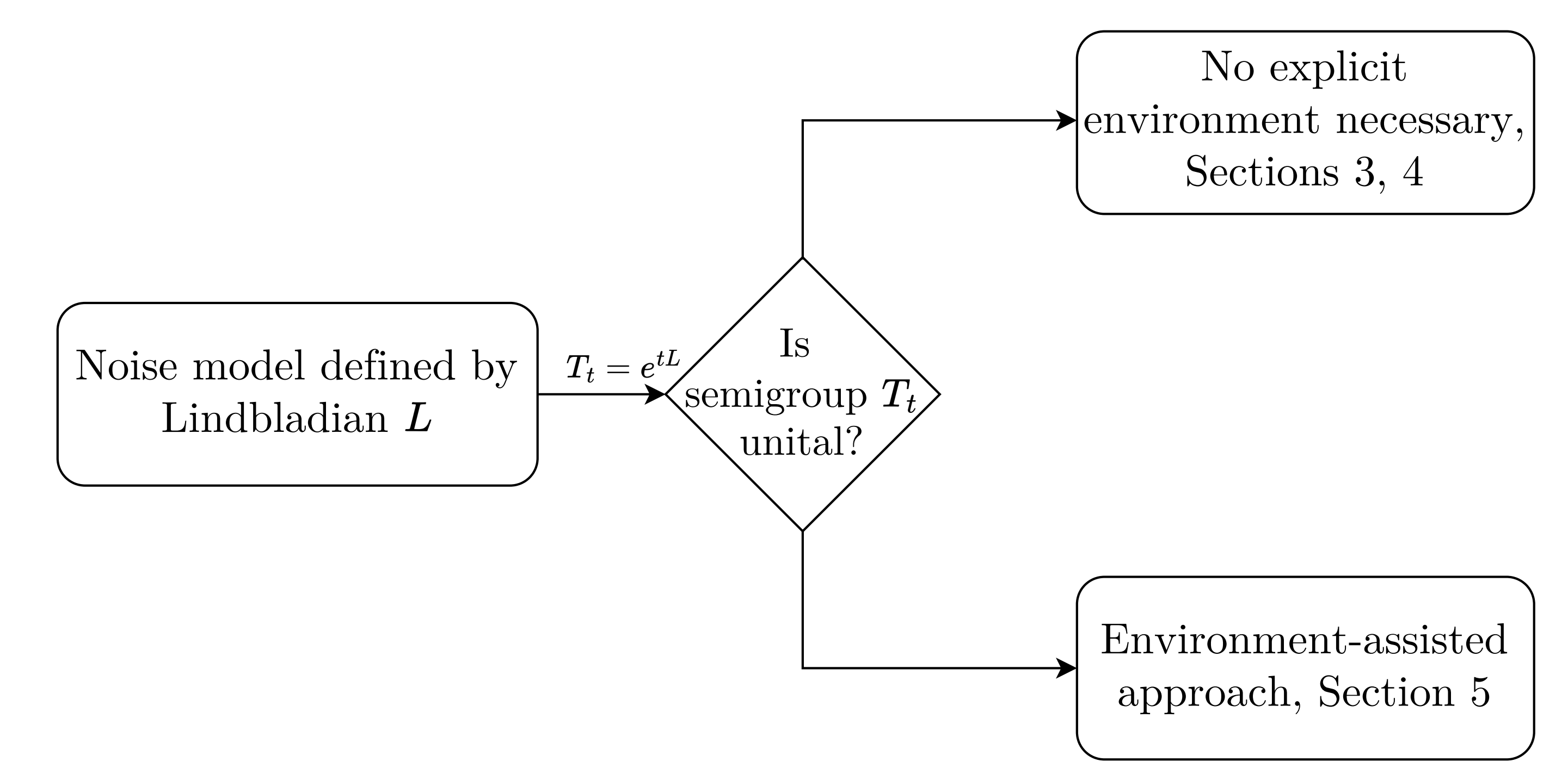}
    \caption{Unital noise models can be treated without allocating an explicit environment as quantum states, which is described in Sections~\ref{sec:upper}~and~\ref{sec:lower_bound}. The case of non-unital noise with environment and correspondingly adjusted Lipschitz complexity is discussed in Section~\ref{sec:non-unital}.}
    \label{fig:unitalono}
\end{figure}

\newpage
\subsection{Main results}\label{sec:main_result}
We begin our discussion with unital quantum dynamics, meaning $L(I)=0$ where $I$ is the identity operator. In this case, we consider the approximation of $T_t$ using mixed unitary channels. Suppose $\adm_u(\rho):=u\rho u^*$. The random channels $\left\{\adm_{u_1},\cdots,\adm_{u_m}\right\}$ are sampled from a probability distribution $\mu^m(u_1,\cdots, u_m)$ on $\mc U^m$ where $\mc U$ is an accessible unitary set, such that $T_t$ is close to $\int_{u_1,\cdots, u_m \in \uu}\adm_{u_1}\cdots \adm_{u_m} d\mu^m(u_1,\cdots, u_m)$. We define $m$ as the circuit depth of the approximation and consider this as the quantum cost. To quantify the capability of the unitary set $\mathcal{U}$ and the simulation difficulty of $T_t$, we recall that most quantum simulation algorithms are based on a high-order local approximation scheme. Specifically, for any positive and small enough $\tau$, we find a distribution $\mu^{m}_\tau\in\mathcal{P}(\mathcal{U}^m)$ such that 
\begin{equation}\label{eqn:high_order}
\left\|T_\tau-\Phi_{\mu^{m}_\tau}\right\|_{\diamond}=\mathcal{O}(\tau^\beta),\quad \Phi_{\mu^{m}_\tau}:=\int_{u_1,\cdots, u_m \in \uu}\adm_{u_1}\cdots \adm_{u_m} d\mu^{m}_\tau(u_1,\cdots, u_m).
\end{equation}
with some integer $\beta>1$ and $\mathcal{O}(\tau^\beta)$ means there exists a universal constant $c>0$ such that $\mathcal{O}(\tau^\beta) \le c \tau^{\beta}$. We note that $\beta \leq 1$ is trivial by using the identity channel as the approximation scheme.
Using the contraction property of semigroups, $\left\|T_t-(\Phi_{\mu^{m}_\tau})^{t/\tau}\right\|_{\diamond}=\mathcal{O}(t\tau^{\beta-1})$. This gives a $\beta-1$-th order simulation scheme for $T_t$. In accordance with this framework, we introduce a definition called convexified circuit depth in terms of $\alpha>0,\beta >1$~(see general definition in Equation \eqref{eqn:convex_depth}):
\[
l^\U_{\alpha \tau^\beta}(T_\tau): = \inf\left\{m \ge 1 \pl | \pl \exists \mu^m_{\tau} \in \mc P(\uu^m) : \left\|T_\tau-\Phi_{\mu^m_{\tau}}\right\|_{\diamond}<\alpha \tau^\beta \right\},
\]
where $\mathcal{U}$ is a set of accessible unitary gates, $\mathcal{P}(\uu^m)$ is the space of probability measure on $\mathcal{U}^m$, and $\alpha \tau^\beta$ is the target precision.
We use $l^\U_{\alpha \tau^\beta}(T_\tau)$ to quantify the minimal circuit depth of the high-order local approximation scheme for $T_\tau$. Given an accessible set of unitaries $\uu$, the convexified circuit depth quantifies the number of unitaries in a single coherent realization of a randomly mixed unitary channel to achieve a specified target error. This new definition serves as a generalization of deterministic circuit depth defined in \eqref{eqn:depth}, encompassing randomized Hamiltonian simulation techniques as discussed in \cite{campbell2019random} and below, which provides a lower bound for the depth of the deterministic circuit. Higher-order numerical schemes need to satisfy \eqref{eqn:high_order} for any small $\tau$, \footnote{Note that for $\tau$ greater than $\left(\frac{2}{\alpha}\right)^{1/\beta}$, we always have $\|T_{\tau} - id\|_{\diamond} \le 2 \le \alpha \tau^{\beta}$.} thus in order to characterize the quantum cost of higher-order schemes, we consider the supremum of $l^\U_{\alpha \tau^\beta}(T_\tau)$ and introduce the uniform maximal simulation cost (see also Definition \ref{def:M}):
\begin{align*}
    \mm_{\alpha,\beta}=\sup_{\tau\in (0,t_0]} l^{\mathcal{U}}_{\alpha \tau^{\beta}}(T_{\tau})\,,\quad t_0 = \left(\frac{2}{\alpha}\right)^{1/\beta}
\end{align*}
which indicates the necessary circuit depth (cost) of a quantum simulation scheme that achieves a specific order. Our first result, with full details given in Section \ref{sec:lower_bound} and the proof given in Theorem \ref{main: lower bound mixing}, is given as follows:
\begin{theorem}\label{Intro: main theorem unital}
    Suppose the accessible unitary gate set $\mc U$ is given. Under the ergodic assumption 
    \begin{align*}
        \lim_{t\to \infty} \frac{1}{t}\int_0^{t} T_s ds = E_{\infty},
    \end{align*}
    and that we can choose a compatible resource set $S$~(see Definition \ref{compatible: U and S} and the discussion after that), the lower bound of $\mm_{\alpha,\beta}$ is given as
    \begin{equation}
        \mm_{\alpha,\beta} \ge \frac{C_S^{cb}(E_{\infty})}{F(\alpha,\beta,t_{\rm mix})},
    \end{equation}
    where $C_S^{cb}(E_{\infty})$ is a convex cost function of the quantum channel $E_{\infty}$, defined in Definition \ref{eqn:C_cb_phi}, $F$ is a function of $\alpha>0,\beta>1$ and the mixing time of $T_t$ introduced in Definition \ref{def:mixing_time}. 
\end{theorem}
The above theorem can be proved via a new framework for the lower bound $\mm_{\alpha,\beta}$ in Section \ref{sec:lower_bound}. We sketch the rough idea here, which contains two steps: 
\begin{itemize}
\item We use convex cost functions $ C^{cb}_S $ (see Definition \ref{eqn:C_cb_phi}) defined on the set of quantum channels. Intuitively, the cost function quantifies a certain distance between the quantum channel and the identity map. These functions are also referred to as Lipschitz complexity in \cite{araiza2023resource}. We demonstrate Lipschitz continuity: 
\[
\left|C^{cb}_S(\Phi)-C^{cb}_S(\Psi)\right| \leq \kappa^{cb}(S)\left\|\Phi-\Psi\right\|_{\diamond},
\]
where the Lipschitz constants $\kappa^{cb}(S)$ depends on the resource set $S$. See Lemma \ref{dd} for a rigorous statement. We can use the subadditivity and convexity properties of $ C^{cb}_S $ to derive that, for any $\delta$-approximation channel with circuit depth $ m $ to $ T_t $,
\[ 
m \geq \Omega\left(C^{cb}_S(T_t) - \delta \kappa^{cb}(S)\right),
\]
where $\Omega\left(C^{cb}_S(T_t) - \delta \kappa^{cb}(S)\right)$ means there exists a universal constant $c>0$ such that $\Omega\left(C^{cb}_S(T_t) - \delta \kappa^{cb}(S)\right) \ge c \left(C^{cb}_S(T_t) - \delta \kappa^{cb}(S)\right)$. The rigorous statement is stated in Lemma \ref{dd}.

\item According to the inequality above, it suffices to establish a lower bound for $C^{cb}_S(T_t) $. In this paper, this lower bound is derived by establishing connections with the stationary (or ergodic average) $ E_{\infty} = \lim_{t \to \infty} \frac{1}{t} \int_0^t T_s \, ds $. Intuitively, for sufficiently large $t$, $T_t$ approaches $ E_{\infty} $, hence $C^{cb}_S(T_t)$ can be lower bounded by $C^{cb}_S(E_{\infty})$ which is the exact Lipschitz constant $\kappa^{cb}(S)$. This argument can be made rigorously by using the semigroup property and introducing the $S$-induced mixing time of $T_t$. Note that the mixing time defined here can be smaller than the mixing time via diamond, which provides a better lower bound. In conclusion, we can prove
\[
\mm_{\alpha,\beta}\geq C^{cb}_S(E_{\infty})/F(\alpha,\beta,t_{\rm mix})\,,
\]
where the formula of $F$ is calculated in Theorem \ref{main: lower bound mixing}. The inequality above suggests that finding a lower bound for $ \mm_{\alpha,\beta} $ can be simplified into two subproblems: 1. Calculating $ C^{cb}_S(E_{\infty}) $; 2. Determining the mixing time of $ T_t $.
\end{itemize}
In real applications, there are many cases where $ E_{\infty} $ is known and the mixing time can be calculated. As a simple application, we can derive a lower bound for Pauli noise on $n$-qubit system as follows:
\begin{prop}
    Suppose the Lindblad generator $L(\rho):= \sum_{j=1}^n (L_{X_j} + L_{Y_j})(\rho)$, where $X_j,Y_j$ are Pauli $X,Y$ operators acting on the $j$-th system. Then the lower bound of $\mm_{\alpha,\beta}$ given accessible unitary set $\mc U = \{\exp(itX_j),\, \exp(itY_j): |t|\le D, 1\leq j \le n\}$ is as follows:
    \begin{equation}
\mm_{\alpha,\beta} \ge C_{\alpha,\beta} \frac{n}{D},\quad C_{\alpha,\beta}:= \frac{1}{2}\min \left\{ \alpha^{-\frac{1}{\beta -1}} \left( \frac{1}{6\ln 2} \right)^{\frac{\beta}{\beta - 1}} , \frac{1}{8} \right\}.        
    \end{equation}
\end{prop}
\noindent Note that the above estimate is tight in terms of the underlying dimension of the system if we construct an approximation channel which achieves $\mc O(n)$ as the upper bound. In fact, this approximation channel can be constructed using what we call the Poisson approach in Section \ref{sec:nonsymm_case} and \ref{sec: example:Poisson}. See also \cite{berry2014simulating} for similar techniques in Hamiltonian simulation.

Although Theorem \ref{Intro: main theorem unital} only considers the local cost of a simulation scheme that achieves a specific order, our framework can also provide a lower bound for a fixed time $\tau > 0$. The informal theorem is presented as follows:
\begin{theorem}\label{Intro: fixed time lower bound}
Under the assumption of Theorem \ref{Intro: main theorem unital}, for any $\alpha>0,\beta>1$, we have \begin{equation}
        l^{\mc U}_{\alpha \tau^{\beta}}(T_{\tau}) \ge \Omega\left(\tau \kappa^{cb}(S)/t_{\rm mix}\right),\quad \forall \tau\leq \mc O\left(\left(\frac{1}{\alpha t_{\rm mix}}\right)^{1/(\beta-1)}\right) \,,
    \end{equation}
   where $t_{\rm mix}$ is defined as $t_{\rm mix}(\frac{1}{2})$ in Definition \ref{def:mixing_time}.
\end{theorem}
We put the rigorous version in Theorem \ref{main result: lower bound fixed time} in Section \ref{subsec:lower bound fixed}. Let us state a special case assuming $t_{\rm mix}>1$. Given $0<\delta<\frac{1}{t_{\rm mix}}$ small, we set $\tau=\Omega(1)$, $\alpha=\delta/\tau^\beta$. Applying the above theorem, we obtain
    \begin{equation}\label{ff}
       l^{\mc U}_{\delta}(T_\tau) \ge \Omega\left(\tau \kappa^{cb}(S)/t_{\rm mix}\right)
    \end{equation}
This shows a lower bound for simulation cost at a specific time with specific error threshold, see Corollary \ref{main result: lower bound fixed time and precision} for a rigorous statement. For an explicit example of calculating a nontrivial ratio $\frac{\kappa^{cb}(S)}{t_{\rm mix}}$, refer to Remark \ref{iterate}. Note that the flexibility of the choice of $\alpha,\beta,\tau$ can also provide us a local result. That is, a lower bound when the target time $\tau$ is arbitrarily small.

Now we transfer our focus to non-unital dynamics. In this case, considering simulation of non-unital quantum dynamics solely using unitary gates is not sufficient, as shown by Proposition \ref{nogo: non-unital}. It makes the simulation much more intricate compared to the simulation of unital dynamics. In general, the simulation of these channels induced by Lindblad generators require non-classical ancilla qubits to implement an encoding and decoding process~\cite{cleve_et_al:LIPIcs.ICALP.2017.17,PRXQuantum.5.020332,pocrnic2024quantumsimulationlindbladiandynamics}. Specifically, in the simulation algorithm, a few ancilla qubits are reused to embed (encode) the original system into a larger Hilbert space at the beginning of each iteration. Proper unitary gates are then applied to this larger Hilbert space to approximate the Lindblad evolution. After applying these gates, the ancilla qubits are traced out (decoding) to obtain the updated system qubits, concluding one iteration. This simulation framework closely mirrors the behavior of open physical quantum systems, where the use of ancilla qubits is analogous to the ``environment''. 
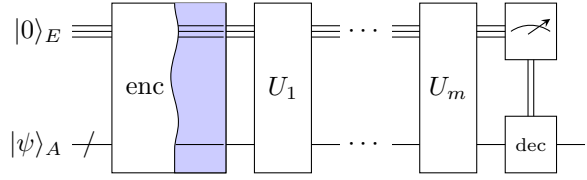
\begin{figure}[bthp]
\color{black}
\begin{center}
\begin{minipage}{.7\textwidth}
\begin{center}
\begin{tikzpicture}[y=-1cm,scale=1.5]
    \draw (4,0) -- (0,0) node[left] {$\lvert0\rangle_E$};
    \draw (4,-0.05) -- (0,-0.05) ;
    \draw (4,0.05) -- (0,0.05) ;
    \draw (4.5,1) -- (0,1) node[left] {$\lvert\psi\rangle_A$};
\draw[fill=white] (0.35,-0.25) rectangle  (1.35, 1.25);
\node at (.64, .5)   {$\mathrm{enc}$};
    \node at (1.05, .45)   {\tiny{e.g.}};
    \draw[fill=blue!20] (.9, -.25)   decorate [decoration={snake, amplitude=0.5mm, segment length=12mm}] { --(.9, 1.25)} -- (1.35, 1.25) -- (1.35, -.25) ;
\draw (.93,-0.05) -- (1.33,-0.05) ;
    \draw (.935,0) -- (1.33,0);
    \draw (.93,0.05) -- (1.33,0.05) ;
    \draw (.91,1) -- (1.33,1) ;
    \draw[fill=white] (1.6,-0.25) rectangle node {$U_1$} (2.1, 1.25);
    \draw[fill=white] (3.05,-0.25) rectangle node {$U_m$} (3.55, 1.25);
    \draw[draw=none,fill=white] (2.35, -.1)  rectangle node {$\cdots$}  (2.8, .1);
    \draw[draw=none,fill=white] (2.35, 0.9)  rectangle node {$\cdots$}  (2.8, 1.1);
    \draw[fill=white] (3.8, -.25) rectangle  (4.25, .25);
        \draw (3.85,0) arc[start angle=-125, end angle=-55, radius=.3];
        \draw[>=stealth,->] (4.0, .0255) -- (  4.195 , -.1108  );
    \draw[fill=white] (3.8, .75) rectangle node {{$\scriptstyle{\mathrm{dec}}$}} (4.25, 1.25);
    \draw (4.0,.25) -- (4.0, .75);
    \draw (4.05,.25) -- (4.05, .75);
\draw (.15, 1) node {$/$};

\end{tikzpicture}
\end{center}
\end{minipage}
\end{center}
\normalcolor
\caption{One iteration of simulating non-unital quantum dynamics with support of an environment $E$. Around the actual simulation, there is an encoding and a decoding map. This is discussed in more detail in Section~\ref{sec:non-unital}. }
\label{fig:a1}
\end{figure}

In our paper, the simulation cost is defined within the fixed admissible unitary set on expanded Hilbert space. It can be formulated as follows: given an accessible unitary set $\mc U_{A\otimes E}$ acting on the joint system $A\otimes E$ and encoding and decoding schemes denoted by $\mc E$ and $\mc D$ respectively, 
\begin{equation}
            \widehat{l}^{\mc U_{A\otimes E}}_{\alpha \tau^{\beta}}(T_{\tau}) := \inf\left\{\sum_{j=1}^M m_j\middle | \exists M\ge 1 ,\exists \mu_j \in \mc P(\mc U_{A\otimes E}^{m_j}), \left\|T_{\tau} - \prod_{j=1}^M \mc D \circ \Phi_{\mu_j} \circ \mc E\right\|_{\diamond} < \alpha \tau^{\beta} \right\}.
\end{equation}
We establish results analogous to those in the unital case, despite dealing with a more intricate scenario involving Lindblad operators with non-symmetric and non-unitary jump operators. 
Due to the additional ancilla qubits, constructing the framework for analyzing lower bounds in this case becomes significantly more complex. Our contribution is to adjust the lower bound framework for unital case and apply it to the non-unital case. As a result, we can provide a lower bound for $\widehat{l}^{\mc U_{A\otimes E}}_{\alpha \tau^{\beta}}(T_{\tau})$ and the supremum version 
\begin{align*}
    \widehat{\mm}_{\alpha,\beta}:= \sup_{\tau\in (0,t_0]}  \widehat{l}^{\mc U_{A\otimes E}}_{\alpha \tau^{\beta}}(T_{\tau}), \quad t_0 = \left(\frac{2}{\alpha}\right)^{1/\beta},
\end{align*}
which captures the same behavior as shown in Theorem \ref{Intro: main theorem unital}. 
In the expanded Hilbert space, the original Lipschitz complexity function cannot be naively applied. Instead, we develop a new Lipschitz complexity function tailored to accommodate the encoding and decoding processes of the algorithm. We also introduce a new class of resource sets in the larger Hilbert space and incorporate additional mild assumptions. For a detailed discussion, please see Section \ref{sec:non-sym-non-uni}. Furthermore, in Section \ref{sec:nqadn}, we address the $n$-qubit amplitude damping noise model and propose a perfect simulation scheme using simple unitaries. This scheme allows us to demonstrate a matching upper and lower bound for the uniform maximal simulation cost $\widehat{\mm}_{\alpha,\beta}$ in this case, showcasing the tightness of our lower bound framework.

\newpage

\subsection{Related works and open questions}\label{sec:related_works}
There is a great amount of different quantum algorithms have been proposed for Hamiltonian evolution~\cite{berry2007efficient,berry2014exponential,berry2015hamiltonian,Low2019hamiltonian,GilyenSuLowEtAl2019,PhysRevX.11.011020,childs2012hamiltonian}, each assuming different oracle access to the Hamiltonian $H$. For example,~\cite{berry2007efficient} assumes oracle access to the matrix entries of of the sparse Hamiltonian $H$ and the quantum cost is defined by the number of queries to the matrix entries. Based on the higher order Trotter splitting and sparse matrix decomposition, they present an efficient algorithm to simulation $\exp(-itH)$ with number of queries almost linearly in $t$ and logarithmic in the number of qubit $n$. A notable result of nearly optimal algorithm is given in \cite{berry2015hamiltonian} and the authors provide a linear in time lower bound for sparse Hamiltonian. 

More recently, as the applications of open quantum systems continue to expand, more attention has been focused on their simulation, especially for the Lindblad quantum master equation~\cite{kliesch2011dissipative,cleve_et_al:LIPIcs.ICALP.2017.17,li2022simulating,chen2023quantum,PRXQuantum.5.020332, borras2024quantum,chen2024randomized,childs2016efficient}. Because open quantum dynamics are often non-unital and irreversible, their simulation becomes much more difficult than that of the Hamiltonian simulation. Although the quantum oracle assumptions and simulation methods vary across different works, the simulation of the Lindblad master equation often requires a similar encoding and decoding process by using additional ancilla qubits, such as in~\cite{cleve_et_al:LIPIcs.ICALP.2017.17,chen2023quantum,PRXQuantum.5.020332}. The first-order encoding process was initially proposed by Cleve et al.~\cite{cleve_et_al:LIPIcs.ICALP.2017.17}, who also introduced a novel compression scheme and oblivious amplitude amplification to improve the scheme's scaling to nearly optimal levels. Similar concepts are employed in~\cite{chen2023quantum}, where the authors assume slightly different oracle access to manage a large number of jump operators. In~\cite{PRXQuantum.5.020332}, the authors present an alternative encoding process that achieves arbitrarily high-order accuracy but involves a more complex Hamiltonian simulation. 

The lower bound analysis is only derived for some special cases. In \cite{berry2015hamiltonian,bu2024complexity}, the authors provide some lower bounds for cost of Hamiltonian simulation. In \cite[Section 7]{childs2016efficient}, linear in time lower bound for certain single Lindblad generators are constructed -- Theorem \ref{Intro: fixed time lower bound} in this paper captures that behavior as well. It is straightforward to extend our framework to more complex encoding and decoding processes, provided one has access to more sophisticated unitary gates and knows how to calculate the Lipschitz complexity of the corresponding quantum channels with the new resource set. We leave more concrete analysis when the coherent part exists for future research. 

\section{Preliminaries on simulation cost of quantum channels}\label{sec:sc}
Let $\mc H$ be a Hilbert space and $\mb B(\mc H)$ be the space of linear bounded operators on $\mc H$. For any $X \in \mb B(\mc H)$, we denote the Schatten $p$-norm $\|X\|_p:= [\Tr(|X|^p)]^{1/p}$ for $p\in [1,\infty]$, where $\Tr$ is the trace and $\|X\|_{\infty}$ is the operator norm given by the largest singular value of $X$. In this paper, $\|X\|_{\infty}$ is denoted as $\|X\|$ for simplicity. For any linear map $\Phi: \mb B(\mc H) \to \mb B(\mc H)$, denote
\begin{align*}
    \|\Phi\|_{p \to q}:= \sup_{\|X\|_p \le 1} \|\Phi(X)\|_q, \quad \|\Phi\|^{cb}_{p \to q}:= \sup_{X^{(n)} \in \mb B(\mc H \otimes \mb C^n): \|X^{(n)}\|_p \le 1} \|(\Phi\otimes id_n)(X^{(n)})\|_q,
\end{align*}
where $id_n: \mb B(\mb C^n)\to \mb B(\mb C^n)$ is the identity map. In particular, the diamond norm of a linear map $\Phi: \mb B(\mc H) \to \mb B(\mc H)$ is given by
\begin{equation}\label{def:diamond norm}
    \|\Phi\|_{\diamond} = \|\Phi\|^{cb}_{1 \to 1}.
\end{equation}
A quantum channel $\Phi: \mb B(\mc H) \to \mb B(\mc H)$ is a linear map such that it is trace preserving and completely positive. For any unitary operator $u\in \mb B(\mc H)$, we denote $\adm_u(\cdot) = u\cdot u^*$ as the quantum channel induced by $u$. The notion of simulation cost of quantum channels with respect to a given unitary set $\mc U \subseteq \mb B(\mc H)$ is defined as 
 \begin{equation}\label{eqn:convex_depth}
    l^\U_{\delta}(\Phi): = \inf\left\{m \ge 1 \pl | \pl \exists \mu^m \in \mc P(\uu^m) : \left\|\Phi-\int_{u_1,\cdots, u_m \in \uu}\adm_{u_1}\cdots \adm_{u_m} d\mu^m(u_1,\cdots, u_m) \right\|_{\diamond}<\delta \right\},
 \end{equation}
where $\mc P(\uu^m)$ is the set of all probability distributions on $\uu^m = \uu \times \cdots  \times \uu$.

In this work, \eqref{eqn:convex_depth} is referred to as the convexified notion of circuit depth. In contrast, the traditional (approximate) circuit depth of a unitary $U$ is defined as \cite{haferkamp2022linear, haferkamp2023moments}
\begin{equation}\label{eqn:depth}
    \text{depth}_{\delta}^{\uu}(U):= \inf\left\{m \ge 1 \pl | \exists u_1, \cdots, u_m \in \uu : \left\|U - u_1\cdots u_m\right\|<\delta \right\}.
\end{equation}
Different from traditional circuit depth~\eqref{eqn:depth}, the convexified notion~\eqref{eqn:convex_depth} allows any kind of convex combination of $\adm_{u_{i}}$.

Although previous work focuses mainly on the depth of traditional circuits~\eqref{eqn:depth}, similar convexified ideas have appeared in various quantum simulation algorithms before. For example, the qDrift protocol for Hamiltonian simulation, $U = e^{itH}$, where $H = \sum_{j=1}^L h_j H_j$, was introduced by the authors in \cite{campbell2019random}. Unlike first-order trotter splitting $U\approx \left(\Pi^L_{j=1}\exp(itH_j/N)\right)^N+O(t^2/N)$, qDrift approximates $U$ using random Hamiltonian evolution:
\begin{equation}\label{eqn:Qdrift}
U\approx \underbrace{\mathbb{E}\left(\Pi^N_{n=1}\left(\exp(itH_{r_n}/N)\right)\right)}_{\text{convexified simulation}}=\Pi^N_{n=1}\mathbb{E}\left(\exp(itH_{r_n}/N)\right)\approx \Pi^N_{n=1}\exp(itH_{r_n}/N),\quad N\gg 1\,,
\end{equation}
where $\{r_n\}^N_{n=1}$ are independently sampled from the uniform distribution over the set $[1,\dots,L]$.
Given a fixed $N$, the circuit depth for the first-order Trotter is $NL$, while the circuit depth for qDrift is $N$. A detailed analysis of the complexity differences between qDrift and Trotter splitting must take into account the impact of randomness, and the reader is directed to~\cite{campbell2019random} for further details. 
Beyond approximating unitary processes, a more significant application of convexified simulation involves simulating non-unitary CPTP maps, such as Lindblad dynamics~\cite{cleve_et_al:LIPIcs.ICALP.2017.17,PRXQuantum.5.020332}. Due to the non-unitary nature of Lindblad dynamics, simulations often require an encoding and decoding process. The decoding process is often implemented by measuring and resetting (tracing out) of ancilla qubits. This tracing out mechanism ultimately relates to a convexified simulation approach as described in \eqref{eqn:convex_depth}. In this paper, we further explore the upper and lower bounds of the convexified circuit depths for this type of simulation algorithms, with more detailed explanation provided in Sections \ref{sec:upper} and \ref{sec:lower_bound}. It is also worth noting that the convex combination of unitary gates is different from linear combinations of unitary gates (LCU)~\cite{childs2012hamiltonian}. When implementing LCU on a quantum computer, we need to construct a preparation gate for the coefficients with some ancilla qubits. Additionally, in each step, we must measure the ancilla qubits and obtain the correct outcome to ensure the success of the algorithm. In contrast, the convex combination of unitary gates is often implemented in the expectation sense. Specifically, we implement random gates according to the distribution $\mu^m$. The tracing out also renders the outcome of the ancilla qubit measurement irrelevant.

In this paper, our main interest lies in analysing the simulation costs of quantum dynamics that can be described as a quantum Markov semigroup:
\begin{defi}
    Suppose $\{T_t\}_{t\ge 0}$ is a family of $\sigma$-weakly continuous quantum channels acting on  $\mb B(\mc H)$. We say that it forms a quantum Markov semigroup if 
    \begin{itemize}
        \item $T_0 = id$.
        \item $T_s T_t = T_t T_s = T_{t+s}$, $\forall s,t\ge 0$.
        \item $T_t$ converges to $T_0$ $\sigma$-weakly as $t$ goes to 0.
    \end{itemize}
\end{defi}
It is well-known \cite{lindblad1976generators} that if $\lim_{t \to 0}\|T_t - T_0\|_{\infty \to \infty} = 0$, then $T_t$ can be written as $T_t = \exp(tL)$, where $L:\mb B(\mc H)\to \mb B(\mc H)$ can be written as: 
\begin{align*}
    L(\rho)= -i[H,\rho] + \sum_{j=1}^m L_{V_j}(\rho),\quad L_{V}(\rho)=V\rho V^* - \frac{1}{2}(V^*V \rho+\rho V^*V).
\end{align*}
To quantify the simulation cost of a quantum Markov semigroup, we establish the definition of uniform maximal simulation cost as follows.
\begin{defi}\label{def:M}
    Suppose $\{T_t\}_{t\ge 0}$ is a quantum Markov semigroup of quantum channels. Suppose $\alpha>0,\beta >1$ and $I \subseteq [0,\infty)$. Denote the uniform maximal simulation cost of $\{T_t\}_{t\ge 0}$ within approximation error at most $\alpha t^{\beta}$ by \begin{equation}
        \mm_{\alpha,\beta}^I\left(\{T_t\}_{t\ge 0}\right)= \sup_{t\in I}l_{\alpha t^{\beta}}^{\mc U}\left(T_t\right).
    \end{equation}
\end{defi}
\noindent Note that when $u,u^{-1}\in\mathcal{U}$, we at least have a linear growth simulation bound for $T_t$:
\begin{equation}\label{eqn:T_t_i_d_difference}
    \|T_t - uu^{-1}\|_{\diamond}=\|T_t - id\|_{\diamond} \le \min\left\{t\|L\|_{\diamond}\exp(t\|L\|_{\diamond}),2\right\}\leq \exp(2)\|L\|_{\diamond}t\,,
\end{equation}
where we use $\|T_t - id\|_{\diamond}\leq \|T_t\|_{\diamond}+\|id\|_{\diamond}$ in the first inequality. This implies that simulating $T_t$ within an error scale linearly in $t$ is always possible. Furthermore, linear behavior in $t$ is not very useful in practical applications, as it ultimately results in a zeroth-order accuracy algorithm. Therefore, our work focuses mainly on scenarios in which the simulation error behaves as $t^{\beta}$ with $\beta>1$.

In the following part of this section, we introduce several important properties that can be directly inferred from the definition and will be useful for our forthcoming analysis. We start with a simple one using the semigroup property:
\begin{lemma}\label{lem:separation_property}
Suppose $\{T_t\}_{t\ge 0}$ is a quantum Markov semigroup of quantum channels. Then for any $\delta >0, t\ge 0$ and integer $m \ge 1$, we have
\begin{equation}\label{eqn:local_to_global}
    l_{m\delta}^{\mc U}(T_t) \le m l_{\delta}^{\mc U}(T_{t/m}).
\end{equation}
\end{lemma}
The proof of the above Lemma is standard in numerical analysis, relying on the contraction property of semigroups. For completeness, we still include a short proof here.
\begin{proof}
    We only need to consider the case where $l_{\delta}^{\mc U}(T_{t/m}) < \infty $ as the statement is trivial otherwise. Suppose that there exists a random unitary channel $\Psi_{t/m}$ with length $l = l_{\delta}^{\mc U}(T_{t/m}) < \infty$ such that 
    \begin{equation}
        \|T_{t/m} - \Psi_{t/m}\|_{\diamond} < \delta.
    \end{equation}
    By the semigroup contraction property and the telescoping argument, 
    \begin{align}
      \|T_{t/m}^m - \Psi_{t/m}^m\|_{\diamond} & = \|\sum_{i=0}^{m-1}(T_{t/m})^{m-i} (\Psi_{t/m})^i - (T_{t/m})^{m-i-1} (\Psi_{t/m})^{i+1}\|_{\diamond}  \\
      & \le \sum_{i=0}^{m-1} \|(T_{t/m})^{m-i-1} (T_{t/m} - \Psi_{t/m})(\Psi_{t/m})^i\|_{\diamond} \\
      & \le m \|T_{t/m} - \Psi_{t/m}\|_{\diamond} < m\delta. 
    \end{align}
    This implies that $\Psi_{t/m}^m$ is a random unitary channel of length $m l_{\delta}^{\mc U}(T_{t/m})$ such that the approximation error to $T_t$ is within $m\delta$. Therefore, by definition we have 
    \begin{equation}
        l_{m\delta}^{\mc U}(T_t) \le \text{depth}\left(\Psi_{t/m}^m\right)=m l_{\delta}^{\mc U}(T_{t/m}).
    \end{equation}
\end{proof}

By employing a similar semigroup contraction argument as used in the proof above, we can also demonstrate the following relation between the maximum uniform length $\mm_{\alpha,\beta}$ and the upper bound on the convexified length $l^{\U}_{\delta}$. In particular, if we know that the maximum length to simulate up to a certain time is finite, it follows that for any finite time at least as large as this reference time, the \textit{upper} bound on the simulation length will also remain finite.
\begin{lemma}\label{init} Suppose there exists $t_0>0$ such that 
 \[ 
 \mm_{\al,\beta}^{[0,t_0]} < \infty \pl\,,
 \] 
then for any $\delta>0$ and $t>0$
, we have
\begin{equation}\label{upper estimate: any time}
     l_{\delta}^\U(T_t) \kl \left(\max\left\{ \left(\frac{\al}{\delta}\right)^{\frac{1}{\beta-1}} t^{\frac{\beta}{\beta-1}}, \frac{t}{t_0}\right\} +1\right)  \pl  \mm_{\al,\beta}^{[0,t_0]}.
\end{equation}
\end{lemma}

\begin{proof} Let $t>0$. Choose $n$ such that we have $t/n\le t_0$ and hence there exists a channel $\Psi_{t/n}$ of length $l\leq \mm_{\al,\beta}^{[0,t_0]}<\infty$ such that 
 \[ \|T_{t/n}-\Psi_{t/n}\|_{\diamond}<\al (t/n)^{\beta} \] 
Using the same contraction and telescoping argument as in the previous lemma, we have
\[ 
\|T_t -\Psi_{t/n}^n\|_{\diamond} < n^{1-\beta} \al t^{\beta}. 
\] 
If $n\gl (\frac{\al}{\delta})^{\frac{1}{\beta-1}} t^{\frac{\beta}{\beta-1}}$, we have 
\begin{equation}
    \|T_t -\Psi_{t/n}^n\|_{\diamond} < n^{1-\beta} \al t^{\beta} \le \delta.
\end{equation}
Therefore, we find a random unitary channel given by $\Psi_{t/n}^n$ of length $ln$, such that $\|T_t -\Psi_{t/n}^n\|_{\diamond} < \delta$. Choosing
\begin{equation}
    n = \left\lfloor \max\left\{ \left(\frac{\al}{\delta}\right)^{\frac{1}{\beta-1}} t^{\frac{\beta}{\beta-1}}, \frac{t}{t_0}\right\}\right\rfloor + 1,
\end{equation}
we have 
\begin{equation}
    l_{\delta}^\U(T_t) \leq \text{depth}\left(\Psi_{t/n}^n\right) \leq ln \le \left(\max\left\{ \left(\frac{\al}{\delta}\right)^{\frac{1}{\beta-1}} t^{\frac{\beta}{\beta-1}}, \frac{t}{t_0}\right\} +1\right)  \pl  \mm_{\al,\beta}^{[0,t_0]}.
\end{equation}
\end{proof}
Although the definition of $\mm^{[0,\infty)}_{\alpha,\beta}$ considers the supremum over all $t \geq 0$, it suffices to focus on the interval where $t$ is below a certain threshold. In particular, for $t >\left(\frac{2}{\alpha}\right)^{1/\beta}$ and any $u \in \mc U$,
\begin{equation}
\|T_t - ad_u\|_{\diamond} \leq \|T_t\|_{\diamond} + \|ad_u\|_{\diamond} \leq 2 < \alpha t^{\beta}\,.
\end{equation}
By definition, this implies
\begin{align*}
    l_{\alpha t^{\beta}}^{\mc U}(T_t) \le 1,\quad \forall t > \left(\frac{2}{\alpha}\right)^{1/\beta}.
\end{align*}
Thus, to find an upper bound of $\mm^{[0,\infty)}_{\alpha,\beta}(\{T_t\})$, we only focus on the a small region where $t\in\left[0, \left(\frac{2}{\alpha}\right)^{1/\beta}\right]$. Without loss of generality, when the underlying accessible unitary and quantum Markov semigroup are clear, we define
\begin{equation}\label{def: maximal length}
    \mm_{\alpha,\beta}:= \sup_{0\leq t \leq (\frac{2}{\alpha})^{1/\beta}} l_{\alpha t^{\beta}}^{\mc U}(T_t).
\end{equation}
In our later analysis, we can often derive the upper and lower bounds of $\mm_{\alpha,\beta}$ for a large $\alpha$. Thefollwoing lemma can help rescale these bounds to arbitrary $\alpha$ up to a constant factor.
\begin{lemma}\label{different parameter}
For fixed $\beta>1$, and any $0< \alpha' \leq \alpha$, we have 
\begin{equation}\label{different alpha}
    \mm_{\alpha,\beta} \leq \mm_{\alpha',\beta} \leq 2 \left(\frac{\alpha}{\alpha'}\right)^{\frac{1}{\beta - 1}}\mm_{\alpha,\beta}.
\end{equation}
\end{lemma}

\begin{proof}
    Without loss of generality, we assume $\mm_{\al,\beta}< +\infty$. Then the assumption in Lemma \ref{init} holds true for $t_0 = (\frac{2}{\alpha})^{1/\beta}$.
    
    \noindent For any $t\in \left(0,(\frac{2}{\alpha'})^{1/\beta}\right]$, we choose $\delta = \alpha' t^{\beta}$. Then using Lemma \ref{init} we have
\begin{align*}
    l_{\alpha' t^{\beta}}^\U(T_t) & \leq \left(\max\left\{ \left(\frac{\al}{\alpha' t^{\beta}}\right)^{\frac{1}{\beta-1}} t^{\frac{\beta}{\beta-1}}, \frac{t}{t_0}\right\} +1\right)  \pl  \mm_{\al,\beta} \\
    & \leq 2 \max\left\{\left(\frac{\alpha}{\alpha'}\right)^{\frac{1}{\beta - 1}}, \left(\frac{\alpha}{\alpha'}\right)^{\frac{1}{\beta}}\right\}\mm_{\alpha,\beta} = 2\left(\frac{\alpha}{\alpha'}\right)^{\frac{1}{\beta - 1}}\mm_{\alpha,\beta}.
\end{align*}
Taking the supremum over $t\in \left(0,(\frac{2}{\alpha'})^{1/\beta}\right]$, we get the desired result.
\end{proof}

\section{Upper estimates for uniform maximal simulation cost of unital quantum dynamics.}\label{sec:upper}
In this section, we demonstrate that for a specific category of unital quantum Markov semigroups and a suitably selected unitary set $\mc U$ acting on the (sub)system $\mc H$ only, efficient simulation of $T_t$ is achievable. Specifically, we present two instances where an upper bound for the uniform maximal simulation cost can be established in Sections \ref{sec:symm_case} and \ref{sec:nonsymm_case}.

\subsection{Symmetric noise generators: Trotter approximation}\label{sec:symm_case}

The first specific instance we consider are quantum dynamics with symmetric jump operators. In this case, we can simplify the quantum dynamics simulation to a unitary simulation, using the Trotter formula to establish an upper bound for the depth of the circuit~\cite{PhysRevX.11.011020}.


We look at the dynamical semigroup that is generated by the following symmetric generator:
\begin{equation}\label{generator: symmetric}
    L(x) = \sum_{j=1}^m a_j x a_j - \frac{1}{2}\left(xa_j^2 + a_j^2 x\right),\quad a_j = a_j^*\,,
\end{equation}
where $a^*_j$ is the adjoint of $a_j$. Note that $a_j$ can be non-local when the underlying system is $n$-qubit system. We assume the accessible unitary set:
\begin{equation}\label{unitary: symmetric}
    \mc U = \{\exp(ita_j): |t| \le \tau, 1\le j \le m\}.
\end{equation}

The upper bound of $\mm_{\alpha,\beta}$ is summarized in the following theorem:
\begin{theorem}\label{main: symmetric}
    Suppose $T_t = e^{tL}$ with $L$ given by \eqref{generator: symmetric} and $t>0$. Then, for any $\alpha >0$ and $1< \beta \leq 3$, we have $\mm_{\alpha,\beta} < +\infty$. Moreover, there exists a universal constant $C>0$ depending only on $\max_i \|a_i\|<+\infty$ such that
    \begin{equation}
        \mm_{\alpha,\beta} \le C\frac{\max\left\{m^{\frac{5}{2}},m^2/\tau\right\}}{\sqrt{\alpha}}.
    \end{equation}
\end{theorem}
In the proof of the above theorem, we construct a second-order simulation scheme for $\exp\left(t\sum_j L_{a_j}\right)$. The second-order approximation primarily arises from two sources: (1) the second-order simulation scheme for $\exp(tL_{a_j})$, and (2) the second-order Trotter splitting of $\exp\left(t\sum_j L_{a_j}\right)$. Although it is possible to improve the former to arbitrarily high order, the latter splitting error remains the bottleneck for improving the $5/2$ power in $m$. Since we do not assume further locality assumptions on $a_j$, it is not a trivial task to split $\exp\left(t\sum_j L_{a_j}\right)$ to achieve a higher-order scheme.

To prove the above theorem, we first establish the following Lemma for Lindbladian dynamics with a single jump operator.

\begin{lemma}\label{symmetric: approximation}
    For a self-adjoint matrix $a = a^*$, denote 
    \begin{equation}
        L_a(x):= 2 axa - (a^2x + xa^2).
    \end{equation}
    Then there exists a universal constant $C>0$ such that
    \begin{equation}\label{eqn:expectation_approximation}
    \left\|\exp(tL_a) - \mb E\left[\adm_{\exp(i \sqrt{2t}Za)}\right]\right\|_{\diamond}\leq C \|a\|^6t^3\,.
    \end{equation}
    where $Z$ is a random variable such that
    \begin{equation}\label{eqn:distribution}
    \mathbb{P}(Z=2)=\mathbb{P}(Z=-2)=1/12,\quad \mathbb{P}(Z=1)=\mathbb{P}(Z=-1)=1/6,\quad \mathbb{P}(Z=0)=1/2\,.
    \end{equation}
\end{lemma}

\begin{proof}
To prove \eqref{eqn:expectation_approximation}, we first let $Z$ be an arbitrary random variable with real value. We diagonalize $a=\sum_j \la_j e_j$ where $e_j$ are orthogonal projections. Then,
 \begin{equation}\label{eqn:exp_form}
 e^{-i\sqrt{2t}Za}\rho e^{i\sqrt{2t}Za}
 \lel \sum_{jl} e^{i\sqrt{2t}Z(\la_l-\la_j)} e_j\rho e_l
 \lel \sum_{jl} \sum_k \frac{i^k(2t)^{k/2}Z^k}{k!} (\la_l-\la_j)^k  e_j\rho e_l \pl\,. 
 \end{equation}
 Because $L^k_a(\rho)=\sum_{jl}(-1)^k(\lambda_l-\lambda_j)^{2k}e_j\rho e_l$, we have
\[ 
\ez \left(e^{-i\sqrt{2t}Za}\rho e^{i\sqrt{2t}Za}\right)
\lel \sum_k \frac{(2t)^k}{(2k)!} \ez \left(Z^{2k}\right) L_a^k (\rho)=\exp(tL_a)(\rho)
\,. 
\]
Note that if $Z\sim \mathcal{N}(0,1)$, standard Gaussian random variable, we have $\mathbb{E}(Z^{2k})=(2k)!/(k!2^k)$ thus 
$
\ez \left(e^{-i\sqrt{2t}Za}\rho e^{i\sqrt{2t}Za}\right) = \exp(tL_a)(\rho)
$.

Now, to find an approximation of order $m$, we construct a symmetric bounded random variable $Z$ such that $|Z|\leq b$ and $\mathbb{E}(Z^{2k})=(2k)!/(k!2^k)$ for $k=1,...,m-1$. Then 
\begin{align*} 
\left\| \ez\left(\adm_{e^{i\sqrt{t}Za}}\right)-e^{-tL_a}\right\|_{\diamond}
\le \sum_{k\ge m} \frac{(2t)^{k}}{(2k)!} b^{2k} \|L_a\|_{\diamond}^k + \sum_{k\ge m} \frac{t^{k}}{k!}\|L_a\|_{\diamond}^k \kl \frac{t^m}{m!} \left(e^{\sqrt{2t}b\|L_a\|_{\diamond}} + e^{t\|L_a\|_{\diamond}}\right)\leq \mc O( \|a\| t^{m}) 
\end{align*}
holds for $|t|\le 1$. Noticing that \eqref{eqn:distribution} satisfies the above condition with $m=3$ and $b=2$, we conclude the proof.
\end{proof}

This lemma demonstrates how to simulate $\exp(tL_a)$ using few unitary gates and is the key tool that we will use to obtain the upper bound of $\mm_{\alpha,\beta}$.  As seen above, the proof of this lemma relies on Taylor expansion and moment matching for Gaussian random variables. For more than one term $L=\sum_j L_{a_j}$, we will apply the usual Suzuki-Trotter approximation. In particular, for $\exp\left(t\left(\sum^m_{i=1} L_i\right)\right)$, we consider second-order Trotter splitting: 
\begin{align}
    S_2(t) :=\exp(tL_1/2) \cdots \exp(tL_m/2)\exp(tL_m/2) \cdots \exp(tL_1/2)
    \label{eqn:second_order_trotter}\,.
\end{align}
According to \cite{PhysRevX.11.011020,suzuki1976generalized}, we have
\begin{equation}\label{eqn:trotter_error}
\left\|S_2(t)-\exp\left(t\sum^m_{j=1} L_j\right)\right\|_{\diamond}\leq O\left(\left(\sum^m_{j=1} \|L_j\|_{\diamond}\right)^3t^3\right)\,, \;\forall t>0\,.
\end{equation}

\noindent Now we are ready to provide the proof of Theorem \ref{main: symmetric}:


\begin{proof}
Because $l^{\mathcal{U}}_{\alpha t^\beta}$ is an increasing function in $\beta$ when $0<t<1$, we have
\[
\mm^{[0,1]}_{\alpha,\beta}\leq \mm^{[0,1]}_{\alpha,3}
\]
if $1<\beta\leq 3$. According to Lemma \ref{init}, to prove $\mm_{\alpha,\beta}<\infty$ for $1<\beta\leq 3$, it suffices to show $\mm_{\alpha,3}<\infty$.

First, we assume that $\left\{Z_j \right\}^m_{j=1}$ is a sequence of independent probability distributions that is identically distributed as Lemma \ref{symmetric: approximation} \eqref{eqn:distribution}. According to Lemma \ref{symmetric: approximation}, 
\begin{equation}\label{eqn:integral_approx}
\begin{aligned}
    & \left\|\exp(tL_{a_j}) - \mb E\left[ad_{\exp(i \sqrt{2t}Z_ja)}\right] \right\|_{\diamond} \leq C \|a_j\|^6t^3\,.          
\end{aligned}
\end{equation}

Using the second-order Trotter approximation, we have 
\begin{equation}\label{eqn:trotter_approx}
    \left\|\exp(tL) - S_2(tL_{a_1}, \cdots, tL_{a_m})\right\|_{\diamond}\leq O\left(\left(\sum^m_{j=1} \|L_{a_j}\|_{\diamond}\right)^3t^3\right)\leq C_2\left(\sum^m_{j=1} \|a_j\|^2 \right)^3t^3,
\end{equation}
where $S_2(tL_{a_1}, \cdots, tL_{a_m})$ is defined in \eqref{eqn:second_order_trotter}. 

Combining \eqref{eqn:integral_approx} and \eqref{eqn:trotter_approx}, we have 
\begin{align*}
    & \left\|\exp(tL) -  \prod_{j=1}^m \mb E\left[ad_{\exp(i \sqrt{t}Za)}\right] \prod_{j=m}^1 \mb E\left[ad_{\exp(i \sqrt{t}Za)}\right]\right \|_{\diamond} \\
    & \le  \left\|\exp(tL) -  \prod_{j=1}^m \exp(tL_{a_j}/2)\prod_{j=m}^1\exp(tL_{a_j}/2) \right\|_{\diamond} \\
    & + \left\| \prod_{j=1}^m \exp(tL_{a_j}/2)\prod_{j=m}^1\exp(tL_{a_j}/2) -\prod_{j=1}^m \mb E\left[ad_{\exp(i \sqrt{t}Za)}\right] \prod_{j=m}^1 \mb E\left[ad_{\exp(i \sqrt{t}Za)}\right] \right\|_{\diamond} \\
    & \le C_2\left(\sum^m_{j=1} \|a_j\|^2 \right)^3t^3 + 2C \sum_{j=1}^m \|a_j\|^6 t^3 \leq C_1 m^3 t^3.
\end{align*}
For the second term in the second last inequality, we use a telescoping argument and the contraction property of the semigroup.
According to the above inequality, for $\alpha_0 = C_1 m^3$, we get an upper bound for $l^{\mc U}_{\alpha_0 t^{3}}(T_t)$. This is the length of \begin{align*}
    \prod_{j=1}^m \mb E\left[ad_{\exp(i \sqrt{t}Za)}\right] \prod_{j=m}^1 \mb E\left[ad_{\exp(i \sqrt{t}Za)}\right]
\end{align*}
Recall that our accessible unitary set is given by 
\begin{align*}
    \mc U = \{\exp(it a_j): |t|\le \tau, 1\le j\le m\},
\end{align*}
thus the length is given by
\begin{equation}
    \begin{cases}
         2m |Z| ,\quad \sqrt{t} \le \tau, \\
        2m |Z| \left\lceil  \frac{\sqrt{t}}{\tau}\right\rceil,\quad \sqrt{t} \ge \tau\,,
    \end{cases}
\end{equation}
where $\left\lceil \frac{\sqrt{t}}{\tau} \right\rceil$ is the smallest integer that is greater than $\frac{\sqrt{t}}{\tau}$. Finally, because $|Z| \le 2$,
\begin{equation} \label{local estimate: upper bound}
    l^{\mc U}_{\alpha_0 t^{3}}(T_t) \le 4m \max\left\{1, \left\lceil \frac{\sqrt{t}}{\tau} \right\rceil\right\}. 
\end{equation}
Taking the supremum over $0\le t \le \left(\frac{2}{\alpha_0}\right)^{1/3}$ where $\alpha_0 = C_1 m^3$, we get
\[
\mm_{\alpha_0, 3} \le \max\left\{4m, \frac{2^{\frac{13}{6}} \sqrt{m}}{C_1^{\frac{1}{6}} \tau}\right\}\,.
\]

Finally, applying Lemma \ref{different parameter}, the above inequality implies 
\begin{align*}
    \mm_{\alpha,3} \le 2 \left(\frac{\alpha_0}{\alpha}\right)^{1/2} \mm_{\alpha_0,3} \le 2 \sqrt{\frac{C_1}{\alpha}} \max\left\{4m^{\frac{5}{2}}, \frac{2^{\frac{13}{6}} m^2}{C_1^{\frac{1}{6}} \tau}\right\} \sim m^{\frac{5}{2}}/\sqrt{\alpha}\,,
\end{align*}
which concludes the proof.
\end{proof}

\subsection{Unitary noise generators: Poisson model}\label{sec:nonsymm_case}

Suppose $T_t = \exp(tL)$ with
\begin{equation} \label{generator: Poisson}
    L(x) = \sum_j \lambda_j L_{u_j}(x)=\sum_{j=1}^m \lambda_j (u_j x u_j^* - x)\,,
\end{equation}
where the jump operators $\{u_j\}^m_{j=1}$ are unitary and $\lambda_j>0$. Without loss of generality, we assume normalization in the sense that 
\begin{align*}
    \sum_j \lambda_j = 1,
\end{align*}
otherwise we replace $\lambda_j$ by $\frac{\lambda_j}{\sum_j \lambda_j}$.  
We assume the accessible unitary set $$\mc U = \{u_1,\ldots, u_m\}.$$

The following theorem shows that the uniform maximal simulation cost for the above $T_t$ is always finite for any approximation model with $\alpha>0,\beta>1$:
\begin{theorem}\label{main: Poisson}
Suppose $T_t = \exp(tL)$ with $L$ given by \eqref{generator: Poisson}. For any $\alpha>0,\beta>1$, we have
\begin{equation}
    \mm_{\alpha,\beta}:= \sup_{0\leq t\leq\left(\frac{2}{a}\right)^{1/\beta}} l_{\alpha t^{\beta}}^{\uu}(T_t) < N\,,
\end{equation}
where $N$ is a smallest positive integer that satisfies the following two inequalities:
\begin{align}\label{existence of N}
N>\left(\frac{2}{\alpha}\right)^{1/\beta},\quad 
    \frac{1}{\beta}\ln\left(\frac{2}{\alpha}\right) < \frac{(\frac{2}{\alpha})^{1/\beta} + \ln (\alpha/2) + (N+1)\ln(N+1) - N - 1}{N+1-\beta}.
\end{align}
\end{theorem}
To prove Theorem \ref{main: Poisson}, we first observe that by Taylor expansion, and the commuting property of $\Phi = \sum_j \lambda_j\, ad_{u_j}$ and the identity operator, $T_t$ is given by
\[
T_t = \exp(tL)= \exp(t(\Phi - I)) = e^{-t}e^{t\Phi} = \sum_{k\ge 0} \frac{t^k e^{-t}}{k!} \Phi^k, \quad \Phi = \sum_j \lambda_j\, \adm_{u_j}\,.
\]
Then, given our resource set $\U = \{u_1,\ldots,u_m\}$, we can implement a channel truncated in order $N$ with
\[
T_t^{[N]} = \frac{1}{\sum_{k = 0}^N \frac{t^k e^{-t}}{k!}}\sum_{k = 0}^N \frac{t^k e^{-t}}{k!} \Phi^k.
\]
It is worth noting that $T^{[N]}_t$ should not be implemented using LCU. In practice, $T^{[N]}_t$ should be implemented in the expectation sense. Specifically, we implement $\adm_{u_1}\adm_{u_2}\cdots \adm_{u_k}$ with $\{u_i\}$ being i.i.d. samples of the distribution $\mathbb{P}(u=i)=\lambda_i$, and $k$ randomly chosen according to the Poisson distribution $\mathbb{P}(k=l)=\frac{t^l e^{-t}}{l!}\left/\sum_{l = 0}^N \frac{t^l e^{-t}}{l!}\right.$.

\begin{lemma}\label{key lemma: Poisson estimate} For any $N\ge 1$, we have
\begin{equation}
    \left\|T_t-T_t^{[N]}\right\|_{\diamond}\kl 2 e^{N+1 - t - (N+1)\ln \frac{N+1}{t}}, \quad \forall t \in (0,N+1).
\end{equation}
\end{lemma}
\begin{proof}
The key ingredient of the proof is the following estimate:
Define $A(n,t): = e^{-t} \sum_{k=n}^{\infty} \frac{t^k}{k!}$, then
\[
A(n,t) \le e^{n-t - n\ln \frac{n}{t}}
\]
for any $n\ge 1, t>0$ with $n> t$. The estimate is straightforward and we provide a simple proof here. Suppose $N_t$ is a Poisson random variable with parameter $t>0$, then for any $\lambda>0$
\begin{align*}
    A(n,t) =\mb P(N_t \ge n) \le \frac{\mb E \exp(e^{\lambda N_t})}{e^{\lambda n}} = \exp((e^{\lambda} - 1)t - \lambda n).
\end{align*}
Now we choose $\lambda = \ln \frac{n}{t}$ and we get the desired estimate. Using the above estimate, we have
\begin{align*}
& \left\|T_t - T_t^{[N]}\right\|_{\diamond} = \left\|\sum_{k\ge 0} \frac{t^k e^{-t}}{k!} \Phi^k - \frac{1}{\sum_{m = 0}^N \frac{t^m e^{-t}}{m!}}\sum_{k = 0}^N \frac{t^k e^{-t}}{k!} \Phi^k\right\|_{\diamond} \\
& \le \left\| \sum_{k=0}^N \left(e^{-t} - \frac{e^{-t}}{\sum_{m=0}^N \frac{e^{-t} t^m}{m!}}\right)\frac{t^k \Phi^k}{k!}\right\|_{\diamond} + \left\|\sum_{k=N+1}^{\infty} \frac{e^{-t}t^k}{k!} \Phi^k\right\|_{\diamond} \\
& =  \left\| \sum_{k=0}^N - \frac{e^{-t}\sum_{m=N+1}^{\infty} \frac{t^m}{m!}}{\sum_{m=0}^N \frac{t^m}{m!}} \frac{t^k \Phi^k}{k!}\right\|_{\diamond} + \left\|\sum_{k=N+1}^{\infty} \frac{e^{-t}t^k}{k!} \Phi^k\right\|_{\diamond}\\
& \le \sum_{m=N+1}^{\infty} \frac{e^{-t} t^m}{m!} \left\| \sum_{k=0}^N \frac{1}{\sum_{m=0}^N \frac{t^m}{m!}} \frac{t^k \Phi^k}{k!}\right\|_{\diamond} + \left\|\sum_{k=N+1}^{\infty} \frac{e^{-t}t^k}{k!} \Phi^k\right\|_{\diamond} \\
& \le 2 \sum_{k=N+1}^{\infty} \frac{e^{-t}t^k}{k!} \le 2e^{N+1 - t - (N+1)\ln \frac{N+1}{t}}.
\end{align*}
\end{proof}
\noindent Now, we are ready to prove Theorem \ref{main: Poisson}:
\begin{proof}
According to the argument in the beginning of this section, it suffices to consider $t \in [0,(\frac{2}{\alpha})^{1/\beta}]$. We notice that $N$ satisfies 
\begin{align*}
    N+1 - (N+1)\ln(N+1) + (N+1-\beta)\ln t < t + \ln \alpha/2 \le \left(\frac{2}{\alpha}\right)^{1/\beta} + \ln \alpha/2\,.
\end{align*}
Using the above lemma, when $t<N+1$, we have
\begin{align*}
    \left\|T_t - T_t^{[N]}\right\|_{\diamond} \le  2 e^{N+1 - t - (N+1)\ln \frac{N+1}{t}} < \alpha t^{\beta}\,.
\end{align*}
Because $N>(\frac{2}{\alpha})^{1/\beta}$, this  implies that
\[
\sup_{t\in [0,(\frac{2}{\alpha})^{1/\beta}]}l^{\mathcal{U}}_{\alpha t^\beta}\left(T_t\right)\leq N\,.
\]

\end{proof}

\section{Lower estimates for uniform maximal simulation cost of unital quantum dynamics}\label{sec:lower_bound}
In this section, we provide a lower bound estimation for the uniform maximal simulation cost $\mm_{\alpha,\beta}$ defined in Definition \ref{def:M}. Although the uniform maximal simulation cost considers the supremum of the circuit depth, according to the linear growth property of $ l^{\mathcal{U}} $ stated in Lemma \ref{lem:separation_property}, the value of $ \mm_{\alpha,\beta} $ is mainly determined by $ l^\mathcal{U}_{\alpha t^\beta} $ when $ t $ is sufficiently small. The lower bound of $ \mm_{\alpha,\beta} $ essentially informs us about the necessary circuit depth of constructing a $\beta-1$-th order scheme for the simulation of $ T_t $. 

\subsection{Generic lower bound estimates via Lipschitz complexity}


To estimate the lower bound of $\mm_{\alpha,\beta}$ given an accessible unitary set $\mc U$, our main tool is a resource-dependent convex cost function obtained from Lipschitz norms induced by a proper resource set $S$, also known as Lipschitz complexity, introduced in \cite{araiza2023resource}. By comparing the Lipschitz complexity of the quantum semigroup with that of the simulation channel, we derive a lower bound for $l^{\mc U}_{\delta}$.


\subsubsection{Preliminary on Lipschitz complexity}
For a given resource set $S\subseteq \Bb(\mc H)$, we define the induced Lipschitz (semi-)norm $\||\cdot|\|_S$ on $\mb B(\mc H)$ as 
	\begin{equation}
	\||x|\|_S:= \sup_{s \in S}\|[s,x]\|_{\infty},\ x\in \mc \Bb(\mc H),
	\end{equation}
where $[s,x]:= sx-xs$ is the commutator and $\|x\|_{\infty}$ is the Schatten-$\infty$ norm  of $x$. For a quantum channel $\Phi$, the Lipschitz complexity is then defined as
\begin{equation}\label{eqn:C_S_phi}
C_S(\Phi):= \sup_{x\in \mb B(\mc H), \||x|\|_S\le 1} \|\Phi^*(x) - x\|_{\infty}\,,
\end{equation}
where $\Phi^*$ is the dual map of $\Phi$ defined by 
\begin{equation}
    \Tr(\Phi(x) y) = \Tr(x \Phi^*(y)), \forall x,y \in \mb B(\mc H).
\end{equation}
Note that $C_S(\Phi)<+\infty$ only if $\Phi^*$ restricted on $S'$ is an identity map. Here, $S'$ is the centralizer of $S$, defined by 
\begin{equation}\label{centralizer}
    S' = \{x\in \mb B(\mc H): xs = sx, \forall s \in S\}.
\end{equation}

\begin{rmk}
   We illustrate how the commutator $[s,x]$ helps us build the lower bound. Suppose $U = u_m \cdots u_1$. Denote $\Phi(\rho)= \adm_U(\rho) = U\rho U^{*}$, then for any $x \in \mb B(\mc H)$
\begin{align*}
    \|\Phi^*(x) - x\|_{\infty} & = \|u_1^* \cdots u_m^* x u_m \cdots u_1 - x\|_{\infty} = \|x u_m \cdots u_1 - u_m \cdots u_1 x \|_{\infty} \\
    & = \|x u_m \cdots u_1 - u_m x u_{m-1} \cdots u_1 + u_m x u_{m-1} \cdots u_1 - u_m \cdots u_1 x \|_{\infty} \\
    & \le \|[u_m,x] u_{m-1}\cdots u_1\|_{\infty} + \|u_m [u_{m-1}\cdots u_1,x]\|_{\infty} \\
    & = \|[u_m,x]\|_{\infty} + \|[u_{m-1}\cdots u_1,x]\|_{\infty} \\
    & \le \sum_{j=1}^m \|[u_j,x]\|_{\infty} \le m \sup_{i} \|[u_i,x]\|_{\infty}
\end{align*}
where we used the property that $\|x\|_{\infty}$ remains unchanged if we multiply $x$ by a unitary. The second-last inequality follows from induction. Therefore, if we define $S = \{u_1,\cdots, u_m\}$, the minimal length of a unitary using elementary gate set $S$ is always lower bounded by $C_S(\adm_U)$. 
\end{rmk}

The Lipschitz complexity can also be written compactly as $$C_S(\Phi)=\|\Phi^* - id : (\Bb(\mc H), \||\cdot|\|_S) \to \Bb(\mc H)\|_{\rm op}.$$
Recall that for a linear map $T_1: (\mc X, \|\cdot\|_{\mc X}) \to (\mc Y, \|\cdot\|_{\mc Y})$ where  $(\mc X, \|\cdot\|_{\mc X}), (\mc Y, \|\cdot\|_{\mc Y})$ are normed spaces, the induced operator norm of the linear map $T_1$ is defined by 
\begin{equation}
    \|T_1: (\mc X, \|\cdot\|_{\mc X}) \to (\mc Y, \|\cdot\|_{\mc Y}) \|_{\rm op} = \sup_{x \in \mc X: \|x\|_{\mc X} \le 1} \|T_1(x)\|_{\mc Y}
\end{equation}
A useful property of the induced operator norm of linear maps is submultiplicativity:
If $T_1: (\mc X, \|\cdot\|_{\mc X}) \to (\mc Y, \|\cdot\|_{\mc Y}), T_2: (\mc Y, \|\cdot\|_{\mc Y})\to (\mc Z, \|\cdot\|_{\mc Z})$ are two linear maps definded on the corresponding normed spaces, then 
\begin{align*}
    & \|T_1 \circ T_2: (\mc X, \|\cdot\|_{\mc X}) \to (\mc Z, \|\cdot\|_{\mc Z})\|_{\rm op} \\
    & \le \|T_1: (\mc X, \|\cdot\|_{\mc X}) \to (\mc Y, \|\cdot\|_{\mc Y}) \|_{\rm op} \cdot \|T_2: (\mc Y, \|\cdot\|_{\mc Y})\to (\mc Z, \|\cdot\|_{\mc Z})\|_{\rm op}.
\end{align*}

To be consistent with the diamond norm of quantum channel, we extend the definition in \eqref{eqn:C_S_phi} to the complete version of the Lipschitz complexity:
\begin{equation}\label{eqn:C_cb_phi}
C_S^{cb}(\Phi):=\sup_{n\ge 1}\left\|id_n\otimes (\Phi^* - id):(\mb M_n(\mb B(\mc H)), \||\cdot|\|^n_S) \to \mb M_n(\mb B(\mc H))\right\|_{\rm op},
\end{equation}
where $\mb M_n(\mb B(\mc H))\cong \mb M_n \otimes \mb B(\mc H)$ and the Lipschitz norm of the $n$-th amplification is defined by
\begin{equation}
    \||f|\|^n_S = \sup_{s\in S}\|[I_n \otimes s, f]\|_{\infty},\quad f\in \mb M_n(\mb B(\mc H)).
\end{equation}
It is straightforward to see that $C_S(\Phi)\leq C^{cb}_S(\Phi)$. We also  compactly rewrite $C_S^{cb}(\Phi)$ as 
\begin{align*}
   C_S^{cb}(\Phi)=\|\Phi^* - id : (\Bb(\mc H), \||\cdot|\|_S) \to \Bb(\mc H)\|_{cb}.
\end{align*}

Before we present more properties of $C_S$, we calculate two simple examples to illustrate how to estimate $C_S$ in practice. 

\begin{exam}[Replacer channel, single qubit] \label{example: replacer}
    Suppose $\mc H = \mb C^2$ and the resource set is given by \begin{equation}
        S = \left\{X= \begin{pmatrix}
            0 & 1 \\
            1 & 0
        \end{pmatrix},\ Y= \begin{pmatrix}
            0 & i \\
            -i & 0
        \end{pmatrix}\right\}
    \end{equation}
    We consider the replacer channel $\Phi_1$ defined by 
    \begin{equation}
        \Phi_1(\rho) = \Tr(\rho)\ketbra{0}.
    \end{equation}
    Then,
    \begin{equation}
        1 \le C_S(\Phi_1)\le C_S^{cb}(\Phi_1) \le 2.
    \end{equation}
\end{exam}
\begin{proof}

For the lower bound, we can choose $Z = \begin{pmatrix}
    1 & 0 \\
    0 & -1
\end{pmatrix}$ as the test operator. Notice the Kraus representation
\begin{equation}
    \Phi_1^*(\rho) = K_0^* \rho K_0 + K_1^* \rho K_1,\quad K_0 = \ketbra{0},\ K_1 = \ketbra{0}{1}\,,
\end{equation}
and
\begin{equation}
        \|\Phi_1^*(Z) - Z\|_{\infty} = 2,\quad \||Z|\|_S = \sup\{\|[X,Z]\|_{\infty}, \|[Y,Z]\|_{\infty}\} = 2. 
\end{equation}
Thus, the lower bound is given by 
\begin{equation}
     C_S(\Phi_1) \ge \frac{\|\Phi_1^*(Z) - Z\|_{\infty}}{\||Z|\|_S} = 1.
\end{equation}
For the upper bound, we apply the Kraus representation of $\Phi_1^*$ again and note that 
\begin{equation}
    K_0 = \frac{iXY + I}{2},\quad K_1 = \frac{X-iY}{2}\,.
\end{equation}
For any $x \in \mb B(\mc H)$,
\begin{align*}
    \|\Phi_1^*(x) - x\|_{\infty} &= \|K_0^* x K_0 + K_1^* x K_1 - (K_0^*K_0 + K_1^* K_1)x\|_{\infty} \\
    & = \|K_1^*[K_1,x] + K_0^*[K_0,x]\|_{\infty} \\
    & \le \frac{1}{2} \left(\|[X,x]\|_{\infty} + \|[Y,x]\|_{\infty} + \|[XY,x]\|_{\infty}\right).
\end{align*}
Finally, using the commutator identity \begin{align*}
    [XY,x] = X[Y,x] + [X,x]Y, 
\end{align*}
we have  
\begin{align*}
    \|\Phi_1^*(x) - x\|_{\infty} \le 2 \sup\{\|[X,x]\|_{\infty}, \|[Y,x]\|_{\infty}\}\le 2\||x|\|_S\,.
\end{align*}
This implies that $C_S(\Phi_1)\le 2$. Note that the above argument remains the same if we replace $X,Y,K_0,K_1$ by $X\otimes I_n, Y\otimes I_n, K_0\otimes I_n, K_1\otimes I_n$ for any $n\ge 1$, thus $C^{cb}_S(\Phi_1)\le 2$.

\end{proof}

Although simple, the above example includes the necessary techniques to obtain the upper and lower bound of $C_S$. For the lower bound, we need to find a ``certificate'' operator so that $C_S$ is large. For the upper bound, we need to reconstruct the quantum channel using linear combinations and products of the operators in the resource set. 
In what follows, we provide another example exhibiting the above ideas. This example will be used later to provide a sharp lower bound for the simulation cost of the $n$-qubit Pauli noise model.

\begin{exam}[Replacer channel, $n$ qubits]\label{example: ergodic}
    Suppose $\mc H = (\mb C^2)^{\otimes n}$, denote $X_j$ as the Pauli-$X$ operator $X = \begin{pmatrix}
    0 & 1 \\
    1 & 0
\end{pmatrix}$ acting on the $j$-th system, i.e., 
\begin{align*}
    X_j = I_2 \otimes \cdots I_2 \otimes \underbrace{X}_\text{j-th system} \otimes I_2 \cdots \otimes I_2.
\end{align*}
We define $Y_j$ as the Pauli-$Y$ operator $Y = \begin{pmatrix}
    0 & i \\
    -i & 0
\end{pmatrix}$ acting on the $j$-th system in a similar manner. The resource set is given by 
\begin{equation}
    S = \{X_j,Y_j: 1\le j \le n\}.
\end{equation}
We consider the maximally mixed replacer channel defined by 
\begin{equation}
    \Phi_2(\rho_n) = \Tr(\rho_n) \frac{I_{2^n}}{2^n},\quad \rho_n \in \mb B(\mc H).
\end{equation}
Then we can show that 
\begin{equation}
    \frac{n}{2} \le C_S(\Phi_2) \le C^{cb}_S(\Phi_2) \le n.
\end{equation}
\end{exam}
\begin{proof}
   Following a similar approach as in the proof of Example~\ref{example: replacer}, we first consider the ``certificate" operator $\sum_{j=1}^n X_j$ and note that 
\begin{align*}
\left\|\left| \sum_{j=1}^n X_j\right|\right\|_S = \sup_k \left\|\left[Y_k, \sum_{j=1}^n X_j\right]\right\|_\infty = \left\|[Y,X]\right\|_\infty = 2,\quad \left\|\sum_{j=1}^n X_j\right\|_{\infty} = n,\quad \Phi_2\left(\sum_{j=1}^n X_j\right)=0\,.
\end{align*}
This provides a lower bound of $C_S(\Phi_2)$:
\begin{align*}
    C_S(\Phi_2) \ge \frac{ \left\|\sum_{j=1}^n X_j\right\|_\infty}{\left\|\left| \sum_{j=1}^n X_j\right|\right\|_S} = \frac{n}{2}.
\end{align*}

For the upper bound, we can use the fact that the $C_S$ is comparable to the Wasserstein complexity introduced in \cite{li2022wasserstein} (see \cite[Corollary 3.15]{araiza2023resource} for a proof). Here, we include a simple proof for completeness. For any operator $\rho$ on $\mb C^2$, 
\begin{align*}
    E_{\tau}(\rho):= \Tr(\rho) \frac{I_2}{2} = \frac{1}{4}(\sigma_0 \rho \sigma_0 + \sigma_1 \rho  \sigma_1 + \sigma_2 \rho  \sigma_2 +\sigma_3 \rho \sigma_3),
\end{align*}
where we take $\sigma_0 = I_2, \sigma_1 = X,\sigma_2 = Y,\sigma_3 = Z$. In fact, the above holds for any orthonormal basis, see \cite[Section 4]{junge2022stability} for the general result. 

Then, we have 
\begin{align*}
    \Phi_2 = E_{\tau}^{\otimes n}.
\end{align*}
Using sub-additivity under concatenation, we have $C_S(\Phi_2) \le n C_{S_0}(E_{\tau}) \le n$ where $S_0= \{X,Y\}$. The calculation $C^{cb}_{S_0}(E_{\tau}) \le 1$ follows from the same argument as in the above example.
\end{proof}
\noindent \textit{Remark}. Note that using submultiplicativity of the norm, one can show that for any $n$-qubit channel $\Phi$, we have $C^{cb}_S(\Phi) \le C^{cb}_S(\Phi_2) \|\Phi- \Phi_2\|_{\diamond} \le 2n$, which is a simpler proof of an upper bound without using the property of Wasserstein distance of order 1, introduced in \cite{de2021quantum}.



Finally, we summarize several useful properties of (complete version) Lipschitz complexities that we will use later:
\begin{lemma} Given any resource set $S\subset\Bb(\Ha)$, $C_S$ and $C^{cb}_S$ as defined in \eqref{eqn:C_S_phi},\eqref{eqn:C_cb_phi} satisfy:
\begin{enumerate}
	\item $C^{cb}_S(id)=C_S(id)=0$. 
	\item Subadditivity under concatenation: 
 \[
 C^{cb}_S(\Phi \Psi) \le C^{cb}_S(\Phi)+C^{cb}_S(\Psi),\quad C_S(\Phi \Psi) \le C_S(\Phi)+C_S(\Psi)
 \]
 for any quantum channels $\Phi,\Psi$.
	\item Convexity: For any probability distribution $\{p_i\}_{i\in I}$ and any channels $\{\Phi_i\}_{i \in I}$, we have 
\begin{equation}
C^{cb}_S\left(\sum_{i\in I} p_i\Phi_i\right) \le \sum_{i\in I} p_iC^{cb}_S\left(\Phi_i\right),\quad C_S\left(\sum_{i\in I} p_i\Phi_i\right) \le \sum_{i\in I} p_iC_S\left(\Phi_i\right)
\end{equation} 
\end{enumerate}
\end{lemma}
These results can be found in~\cite[Lemma 2.2]{araiza2023resource} and the proof idea is to apply subadditivity and submultiplicativity properties of the induced norm of the quantum channels. 

\subsubsection{Generic lower bound approach: Lipschitz continuity.}
The Lipschitz complexities defined in \eqref{eqn:C_S_phi} and \eqref{eqn:C_cb_phi} are crucial quantities that we will use to derive the lower bound of $ l_{\delta}(T_t) $. We first present the sketch of the roadmap here. Let $\Phi = T_t$ be the exact quantum Markov semigroup and
\[
\Psi = \int_{\mc U^m} \operatorname{ad}_{u_1} \cdots \operatorname{ad}_{u_m} \, d\mu^m(u_1, \ldots, u_m)
\]
be a simulation channel for $\Phi$ with $\delta$-diamond error, meaning $\|\Phi - \Psi\|_{\diamond} \leq \delta$. If we choose a suitable resource set $S$ and $C^{cb}_S$ is a Lipschitz function in diamond norm with the Lipschitz constant $\kappa^{cb}(S)$, we obtain that
\begin{align}\label{eqn:closeness}
|&C^{cb}_S(\Phi) - C^{cb}_S(\Psi)| \le \kappa^{cb}(S) \delta \quad \\ \nonumber \Rightarrow \quad &C^{cb}_S(\Phi) \leq C^{cb}_S(\Psi) + \kappa^{cb}(S) \delta \leq m \sup_{u\in\mathcal{U}} C^{cb}_S(\operatorname{ad}_{u}) + \kappa^{cb}(S) \delta\,,
\end{align}
where we use the subadditivity and convexity of $C_S$ in the last inequality. If we further assume a maximum length per unitary $\sup_{u\in\mathcal{U}} C^{cb}_S(\operatorname{ad}_{u}) \leq D$, where $D>0$ is a uniform constant, \eqref{eqn:closeness} implies
\[
m \geq \frac{C^{cb}_S(\Phi) - \kappa^{cb}(S) \delta}{D}\,.
\]
This provides a lower bound for the circuit depth if we can further provide an estimate of $C^{cb}_S(\Phi)= C^{cb}_S(T_t)$. The above argument and lower bound rely on two key aspects: 
\begin{enumerate}
\item Demonstrating the ``Lipschitz" property of $C^{cb}_S$ and calculating the corresponding Lipschitz constant; 
\item  Estimating $C^{cb}_S(\Phi)= C^{cb}_S(T_t)$. 
\end{enumerate}
Both aspects require a proper choice of the resource set $S$ and additional assumptions to ensure the compatibility of the resource set $S$ with the admissible unitary set $\mathcal{U}$. Specifically, we will show that: 
\begin{enumerate}
\item The Lipschitz constant of $C^{cb}_S$ is determined by the \emph{conditional expectation} onto $S'$. 
\item The estimation of $C^{cb}_S(\Phi)= C^{cb}_S(T_t)$ is determined by the $S$-induced mixing time, which is defined in Definition \ref{def:mixing_time}, and the Lipschitz constant of $C^{cb}_S$.
\end{enumerate}
\begin{figure}[h!]
    \centering
    \includegraphics[width=\textwidth]{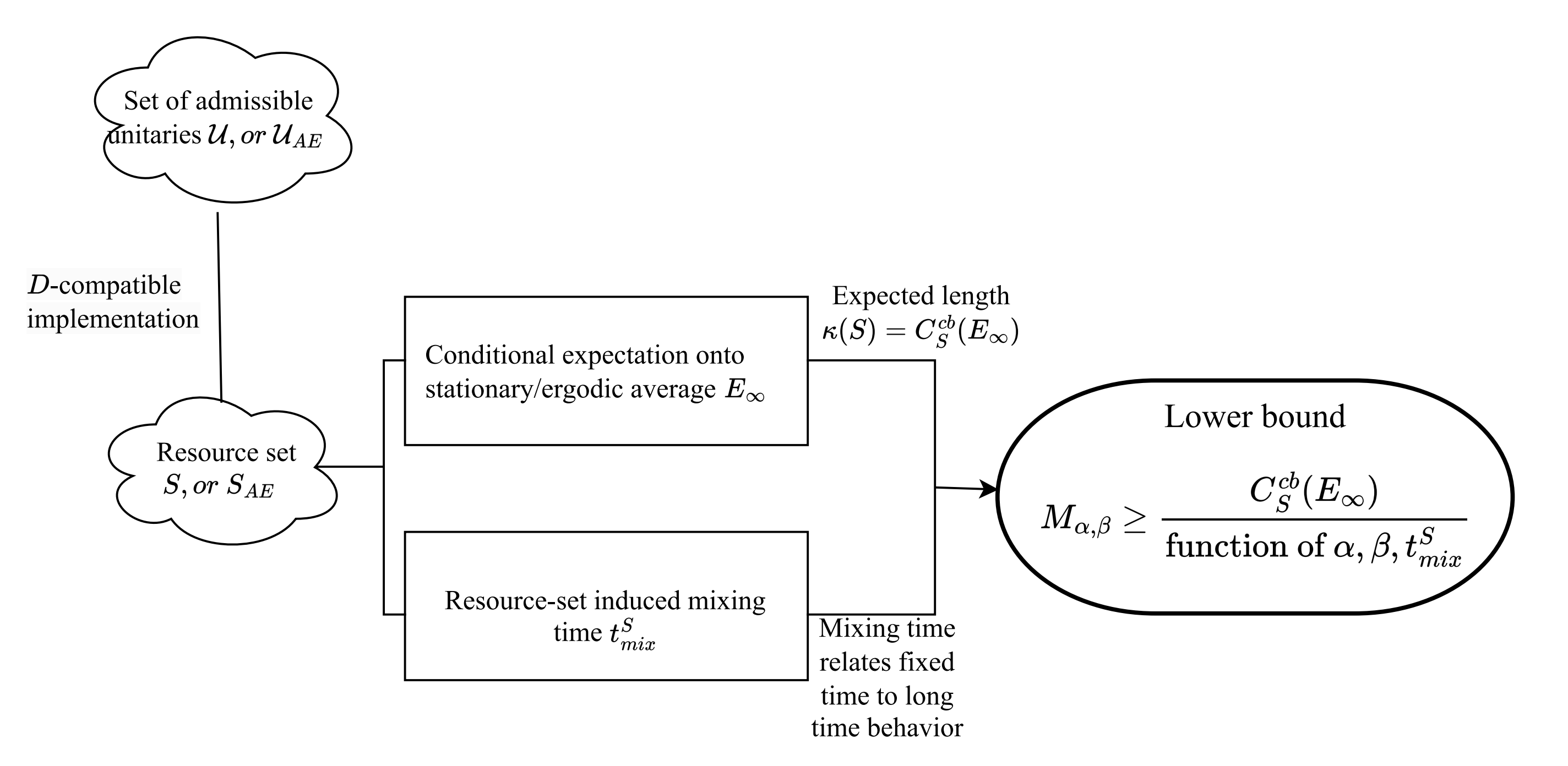}
    \caption{Overview of determining lower bounds of noise models using Lipschitz complexity. We are given simulation parameters $\alpha,\beta>0$, a resource set of bounded operators and a set of admissible unitary gates $\mathcal{U}$ so that every gate in the set has at most complexity $D$ under the Lipschitz complexity induced by the resource set ($D$-compatible). Then, the lower bound can be determined as a function of the resource-set induced mixing time and the expected length $\kappa(S)$, which both rely on knowledge about the conditional expectation on fixed-point densities $E_\infty$ of the open system channel.   }
    \label{fig:lower-bound}
\end{figure}

\begin{rmk} In the roadmap above, as well as in the subsequent calculations, we mainly use $ C^{cb}_S $ to derive the lower bound. However, the same argument also applies to $ C_S $, meaning that $ C^{cb}_S $ can always be replaced with $ C_S $ in these contexts. The only difference is that $ C^{cb}_S $ is always larger than $ C_S $. This means that the lower bound derived using $ C^{cb}_S $ could potentially be sharper than the one derived using $ C_S $.
\end{rmk}

We first briefly review the notion of (non-commutative) conditional expectation. 
Recall that for $C^*$-algebra $\mc N \subset \mb B(\mc H)$, the conditional expectation $E_{\mc N}$ of $\mb B(\mc H)$ onto $\mc N$ is a positive linear map from $\mb B(\mc H)$ to $\mc N$, such that $E_{\mc N}(I_{\mc H}) = I_{\mc H}$ where $I_{\mc H}$ is the identity operator on $\mc H$, and for all $X \in \mb B(\mc H),$ any two elements $S_1,S_2\in \mc N$, we have 
\begin{equation}\label{eqn:conditional-expectation-sandwich}
E_{\mc N}(S_1XS_2) = S_1E_{\mc N}(X)S_2.
\end{equation}
According to \cite[Proposition 1]{kadison2004non}, any conditional expectation is contractive and completely positive. Thus $E_{\mc N}$ is a unital completely positive map and its dual $E_{\mc N}^*$ is a quantum channel. 

Note that a conditional expectation $E_{\mc N}$ might not be unique. Recall that a finite-dimensional $C^*$-algebra $\mc N$ can always be expressed (up to isomorphism) as a direct sum of matrix algebras with multiplicity, i.e.,
\begin{equation}
    \mc N = \bigoplus_i \mb B(\mc H_i) \otimes 1_{\mc K_i},\quad \mc H = \bigoplus_i \mc H_i \otimes \mc K_i
\end{equation}
Denote $P_i$ as the projection onto $\mc H_i\otimes \mc K_i$. The conditional expecation $E_{\mc N}$ can be expressed as
\begin{equation}
    E_{\mc N}(X) = \bigoplus_i \Tr_{\mc K_i}(P_iXP_i (1_{\mc H_i}\otimes \tau_i)) \otimes 1_{\mc K_i}\,,
\end{equation}
where the density operators $\{\tau_i\}$ on the system $\mc K_i$ can be chosen arbitrarily. The dual map of $E_{\mc N}$ can be calculated as
\begin{equation}
    E_{\mc N}^*(\rho) = \bigoplus_i \Tr_{\mc K_i}(P_i \rho P_i ) \otimes \tau_i.
\end{equation}
Different choices of states $\{\tau_i\}$ acting on $\mc K_i$ could possibly give us different conditional expectation. However, if there exists a state $\rho$ acting on $\mc H$ such that 
\begin{equation}\label{state preserving}
    \Tr(\rho E_{\mc N}(s)) = \Tr(\rho s), \forall s\in \mc N
\end{equation}
then $E_{\mc N}$ is uniquely determined by $\rho$. For instance, any state $\rho$ with the form $\bigoplus_i p_i \sigma_i \otimes \tau_i$ satisfies \eqref{state preserving}. From this perspective, if $E_{\mc N}$ satisfies \eqref{state preserving}, then we denote it as $E_{\mc N, \rho}$ to emphasize the dependence on the state. Without the dependence on a certain state, we implicitly assume that it is trace-preserving, i.e., the state $\rho$ is given by maximally mixed state and we denote
\begin{equation}
    E_{\mc N}:= E_{\mc N, \Tr}
\end{equation}
In this case $E_{\mc N}$ is self-adjoint, i.e., $E_{\mc N} = E_{\mc N}^*$.
\begin{rmk}
    For readers familiar with classical probability theory, recall that on a probability space $(\Omega,\mc F, \mb P)$, for a $\sigma$-subalgebra $\mc F' \subseteq \mc F$, to define the conditional expectation $\mb E [X \big| \mc F']$ for $X$ that is measurable w.r.t. $\mc F$, one has to specify the subalgebra (corresponding to the $C^*$-subalgebra) and the probability measure (corresponding to the invariant state).
\end{rmk}
    
The techniques we use in this paper also require some additional technical assumptions on the resource set $S$ and the associated set of admissible unitaries $\mc U$. To this end, we introduce the notion of compatibility of the accessible unitaries $\mc U$ and the resource set $S$.
\begin{defi}\label{compatible: U and S}
    We say that the accessible unitaries $\mc U$ and the resource set $S$ are $D$-compatible with $D>0$ if for any $u\in \mc U$,
    \begin{equation}
        C_S^{cb}(\adm_u) \le D.
    \end{equation}
\end{defi}
\begin{tcolorbox}
\textbf{Assumption (A)}: $S$ is *-closed, i.e., if $s\in S$, then $s^* \in S$. \\
\textbf{Assumption (B)}: There exists a conditional expectation $E_{\rm fix}$ such that for any $u\in \mc U$, we have $E_{\rm fix} \circ \adm_u = E_{\rm fix}$.\\ 
\textbf{Assumption (C)}: The accessible unitaries $\mc U$ and the resource set $S$ are $D$-compatible according to Definition~\ref{compatible: U and S} for some $D>0$.
\end{tcolorbox}
Let us remark that \textbf{Assumption (A)} ensures that $S'$ is a *-algebra. In a finite-dimensional matrix algebra setting, $S'$ is automatically a $C^*$-algebra, since the matrices satisfy $C^*$-identity. In this case, if $E_{S'}\circ \adm_u = E_{\rm fix}$, we often set $E_{\rm fix}=E_{S'}$  to satisfy \textbf{Assumption (B)}. \textbf{Assumption (C)} ensures that one can use the intuitive approach sketched in \eqref{eqn:closeness}. An illustration of the above assumptions is given as follows:
\begin{prop} \label{Assumption: Trotter example}
    Suppose $S$ is given by a set of Hamiltonians: $S = \{H_1,\cdots,H_m \,:\, H_j = H_j^*\}$ and we take the set of admissible unitaries to be $$\mc U = \{\exp(itH_j):|t|\le \tau, j=1,\cdots, m\}.$$
    Then $S$ and $\mc U$ satisfy \textbf{Assumptions (A), (B),} and \textbf{(C)}.
\end{prop}
\begin{proof}
    \textbf{Assumption (A)} is obvious since $S$ consists of self-adjoint operators. For \textbf{Assumption (B)}, we can choose $E_{\rm fix} = E_{S'}$. Then for any $u = \exp(itH_j) \in \mc U$ and $x\in \mb B(\mc H)$, one has 
    \begin{align*}
       E_{S'}(\adm_u(x)) = E_{S'}(\exp(itH_j) x \exp(-itH_j)) = \exp(itH_j) E_{S'}( x )\exp(-itH_j) = E_{S'}(x).
    \end{align*}
    The first equality follows from the definition. The second equality follows from the fact that $\exp(itH_j)$ commutes with any element in $S'$, and we apply \eqref{eqn:conditional-expectation-sandwich}. For the last equality, we use again the fact that $\exp(itH_j)$ commutes with any element in $S'$, thus we have \[\exp(itH_j) E_{S'}( x )\exp(-itH_j) = E_{S'}(x).\] 
    For \textbf{Assumption (C)}, we use the fact that \begin{align*}
        \frac{d}{dt}\big(\exp(-itH_j)x\exp(itH_j)\big) = -i \exp(-itH_j)[H_j,x]\exp(itH_j),
    \end{align*}
    so we have
    \begin{align*}
        & \|\adm_u^*(x) - x\| = \|\exp(-itH_j)x\exp(itH_j)) - x\|_{\infty} \\
        & = \left\|-i\int_0^t \exp(-isH_j)[H_j,x]\exp(isH_j))ds \right\|_{\infty} \le |t| \|[H_j,x]\|_{\infty} \le \tau \sup_j \|[H_j,x]\|_{\infty}.
    \end{align*}
    Now we can use the definition of the Lipschitz complexity $C_S(\cdot)$ from \eqref{eqn:C_S_phi}, and obtain $C_S(\adm_{\exp(itH_j)}) \le \tau=:D$. Using the same argument and replacing the infinity norm by the completely bounded norm, we get the same result  $C_S^{cb}(\adm_{\exp(itH_j)}) \le D$.
\end{proof}
We are now ready to establish the connection between $C^{cb}_S(\cdot)$ and $l_{\delta}^{\mc U}$, and state the rigorous version of \eqref{eqn:closeness}. This is summarized in the following lemma:
\begin{lemma}\label{dd} Let $\Phi$ be a channel such that $\Phi^* \circ E_{\rm fix} = E_{\rm fix}$ for some conditional expectation $E_{\rm fix}$. Under \textbf{Assumptions (B),(C)}, we have that for any $\delta>0$,
\[ C_S^{cb}(\Phi)\kl D l^{\mc U}_{\delta}(\Phi)+ C_S^{cb}(E_{\rm fix}^*) \delta \pl.\]
\end{lemma} 
\begin{rem}\label{rem:continuity} Although it is not explicitly stated in Lemma~\ref{dd} above, the proof of it makes it straightforward to see the ``Lipschitz property" of $C_S^{cb}$:
\[
\left|C^{cb}_S(\Phi)-C^{cb}_S(\Psi)\right|\leq C_S^{cb}(E_{\rm fix}^*) \left\|\Phi-\Psi\right\|_{\diamond}\,.
\]
\end{rem}
\begin{proof}[Proof of Lemma~\ref{dd}]
We only need to prove the case when $l_{\delta}^{\mc U}(\Phi)= m<\infty$. Let $$\Psi=\int_{\mc U^m} \adm_{u_1}\cdots \adm_{u_m} d\mu^m(u_1\cdots u_m)$$ be an approximating random unitary channel such that $\|\Phi-\Psi\|_{\diamond}<\delta$. Then, we claim that 
\begin{equation}  
 C_S^{cb}(\Phi)\kl C_S^{cb}(\Psi)+ C_S^{cb}(E_{\rm fix}^*) \delta \pl .
\end{equation} 
In fact, under the assumptions (\textbf{B}), we have $\Phi^* \circ E_{\rm fix} = \Psi^*\circ E_{\rm fix}= E_{\rm fix}$, and
\begin{align*}
    \Phi^* - id = \Phi^* - \Psi^* + \Psi^* - id = (\Phi^* - \Psi^*)(id - E_{\rm fix}) + \Psi^* - id.
\end{align*}
Then by subadditivity of $C_S^{cb}$ and submultiplicativity of the norm, we have
\begin{align*}
     C_S^{cb}(\Phi) &=\|\Phi^* - id : (\Bb(\mc H), \||\cdot|\|_S) \to \Bb(\mc H)\|_{\rm cb} \\
    & = \|(\Phi^* - \Psi^*)(id - E_{\rm fix}) + \Psi^* - id : (\Bb(\mc H), \||\cdot|\|_S) \to \Bb(\mc H)\|_{\rm cb} \\
    & \le \|\Psi^* - id : (\Bb(\mc H), \||\cdot|\|_S) \to \Bb(\mc H)\|_{\rm cb} 
 + \|(\Phi^* - \Psi^*)(id - E_{\rm fix}): (\Bb(\mc H), \||\cdot|\|_S) \to \Bb(\mc H)\|_{\rm cb} \\
 & \le C_S^{cb}(\Psi) + \|(id - E_{\rm fix}): (\Bb(\mc H), \||\cdot|\|_S) \to \Bb(\mc H)\|_{\rm cb}\cdot \|(\Phi^* - \Psi^*):\Bb(\mc H)\to \Bb(\mc H)\|_{\rm cb} \\
 & = C_S^{cb}(\Psi) + C_S^{cb}(E_{\rm fix}^*)\|\Phi^* - \Psi^*\|_{\rm cb} = C_S^{cb}(\Psi) + C_S^{cb}(E_{\rm fix}^*)\|\Phi^* - \Psi^*\|_{\diamond}.
\end{align*}
Finally, using the convexity and subadditivity under concatenation, it follows that
\begin{align*}
    C_S^{cb}(\Psi)& = C_S^{cb}\left(\int_{\mc U^m} \adm_{u_1}\cdots \adm_{u_m} d\mu^m(u_1\cdots u_m)\right) \le \int_{\mc U^m} C_S^{cb}( \adm_{u_1}\cdots \adm_{u_m} )d\mu^m(u_1\cdots u_m) \\
    & \le \int_{\mc U^m} \sum_{j=1}^m C_S^{cb}(\adm_{u_j})d\mu^m(u_1\cdots u_m) \le \int_{\mc U^m} Dm d\mu^m(u_1\cdots u_m) = D l_{\delta}^{\mc U}(\Phi).
\end{align*}
Therefore, we get $C_S^{cb}(\Phi) \le C_S^{cb}(\Psi) + C_S^{cb}(E_{\rm fix}^*)\|\Phi - \Psi\|_{\diamond} \le D l_{\delta}^{\mc U}(\Phi) + C_S^{cb}(E_{\rm fix}^*) \delta $.
\end{proof}
We demonstrate how to relate the Lipschitz complexity of $T_t$ to the uniform maximum length, the cost-per-unitary $D$ and the Lipschitz constant. This will help us later on to find more specific expressions for the Lipschitz constant. 
To this end, recall that we use the following definition for the uniform maximum length 
\begin{equation}\label{M_alphabeta}
    \mm_{\al,\beta}:=\mm_{\al,\beta}^{[0,(2/\al)^{1/\beta}]}.
\end{equation}
Then, using Lemma \ref{dd}, we can show the following proposition:
\begin{prop} \label{ppp} Let $\al>0$, $\beta>1$. Under \textbf{Assumptions (A),(C)} and assuming that $T_t$ is a unital semigroup satisfying $ T^*_t \circ E_{S'} = E_{S'}$ for all $t\ge 0$, we have
\begin{equation}
    C_S^{cb}(T_t) \kl 
\begin{cases}
    3(\alpha \kappa^{cb}(S))^{1/\beta}(D \mm_{\alpha,\beta})^{1-\frac{1}{\beta}} t,\ \text{if}\ t \ge (\frac{D\mm_{\alpha,\beta}}{\alpha \kappa^{cb}(S)})^{1/\beta}, \\
    2D\mm_{\alpha,\beta},\ \text{if}\ t \le (\frac{D\mm_{\alpha,\beta}}{\alpha \kappa^{cb}(S)})^{1/\beta},
\end{cases}
\end{equation}
where we define
\begin{equation}
    \kappa^{cb}(S) = C_S^{cb}(E_{S'}).
\end{equation}
\end{prop} 
\begin{proof} First, using subadditivity $C^{cb}_S(\cdot)$ under concatenation, the semigroup property and Lemma \ref{dd} with $\delta = \alpha (\frac{t}{n})^{\beta}, n\ge 1$, we have
\begin{equation}\label{Lower bound: arbitrary time}
 \begin{aligned}
 C_S^{cb}(T_t) \lel C_S^{cb}(T_{t/n}^n) \le n C_S^{cb}(T_{t/n}) & \le   n D l_{\alpha (t/n)^{\beta}}^{\mc U}(T_{t/n})+ n^{1-\beta} \kappa^{cb}(S)\al t^{\beta}\\
 & \le n D \mm_{\al,\beta}+ n^{1-\beta} \kappa^{cb}(S)\al t^{\beta}.
 \end{aligned}
\end{equation}
We try to minimize the right hand side by optimizing over $n$. Thus, we define 
 \[ b = \left(\frac{\al\kappa^{cb}(S)}{D\mm_{\al,\beta}}\right)^{1/\beta} t \pl .\]
If $b\gl 1$, we choose $n= \lfloor b \rfloor$ the largest integer number smaller than $b$. In this case, we have $n-1 \le b \le n$ and $b\le n \le 2b$. Therefore, 
\begin{align*}
  C_S^{cb}(T_t)&\le  2b D\mm_{\alpha,\beta} + b^{1-\beta} \alpha \kappa^{cb}(S) t^{\beta}  = 2D\mm_{\alpha,\beta} \left(\frac{\al\kappa^{cb}(S)}{D\mm_{\al,\beta}}\right)^{1/\beta} t + \left(\frac{\al\kappa^{cb}(S)}{D\mm_{\al,\beta}}\right)^{1/\beta - 1}  \alpha \kappa^{cb}(S) t \\
& = 3 (\alpha \kappa^{cb}(S))^{1/\beta}D\mm_{\al,\beta}^{1-1/\beta}t.
 \end{align*}
If $b< 1$, we define $n=1$. Then, 
\begin{align*}
    C_S^{cb}(T_t)&\le D\mm_{\alpha,\beta} + \alpha \kappa^{cb}(S) t^{\beta}.
\end{align*}
Note that for $b <1$, which is equivalent to $\alpha \kappa^{cb}(S)t^{\beta} < D \mm_{\alpha,\beta}$, it then holds that
\begin{align*}
    C_S^{cb}(T_t)&\le 2 D\mm_{\alpha,\beta}.
\end{align*}
\end{proof}
\noindent Based on Proposition \ref{ppp}, we can go one step further by replacing $C^{cb}_S(T_t)$ on the left-hand side with $\kappa^{cb}(S)$ and the mixing time. First, we note that under the assumptions in Proposition \ref{ppp}, via Proposition \ref{characterization: fixed point} we have
\begin{equation}\label{eqn:mixing}
    \lim_{t\to \infty} \frac{1}{t}\int_0^t T^*_s ds = E_{S'}.
\end{equation}
To replace the Lipschitz complexity $C^{cb}_S(T_t)$ with the Lipschitz constant $\kappa^{cb}(S)$, we define the \textit{complexity-induced mixing time} as follows: 
\begin{defi}\label{def:mixing_time}
For any $\varepsilon \in (0,1)$, the complexity-induced mixing time of $T_t$ is defined by
    \begin{equation}
    t_{\rm mix}^S(\varepsilon) = \inf\{t\ge 0: C^{cb}_S(T_t)\ge (1-\varepsilon)\kappa^{cb}(S)\}.
\end{equation}
\end{defi}
\begin{rmk}
    The mixing time $t_{\rm mix}$ defined via the diamond norm is
\begin{equation}
     t_{\rm mix}(\varepsilon) := \inf\{t\ge 0: \|T_t - E_{S'}\|_{\diamond} \le \varepsilon\}.
\end{equation}
Note that the complexity-induced mixing time $t^S_{\rm mix}$ is usually strictly smaller than the diamond norm mixing time $t_{\rm mix}$. We prove a more general case in Lemma \ref{upper bound: mixing time}.  
\end{rmk}
\noindent If the complexity-induced mixing time is finite, $t_{\rm mix}^S(\varepsilon)<\infty$ for some small $\varepsilon>0$, we can now relate the complexity $C^{cb}_S$ to the Lipschitz constant $\kappa^{cb}(S)$ and combine this with Proposition~\ref{ppp},
\begin{equation*}
    (1-\varepsilon)\kappa^{cb}(S) \le C_S^{cb}(T_{t_{\rm mix}^S(\varepsilon)}) \le \begin{cases}
    3(\alpha \kappa^{cb}(S))^{1/\beta}(D \mm_{\alpha,\beta})^{1-\frac{1}{\beta}} t_{\rm mix}^S(\varepsilon),\ \text{if}\ t_{\rm mix}^S(\varepsilon)\ge (\frac{D\mm_{\alpha,\beta}}{\alpha \kappa^{cb}(S)})^{1/\beta}, \\
    2D\mm_{\alpha,\beta},\ \text{if}\ t_{\rm mix}^S(\varepsilon) \le (\frac{D\mm_{\alpha,\beta}}{\alpha \kappa^{cb}(S)})^{1/\beta}.
    \end{cases}
\end{equation*}
Simplifying the above formula, we get the following theorem:
\begin{theorem} \label{main: lower bound mixing}
    Under the assumptions of Proposition \ref{ppp} and $t^S_{\rm mix}(\varepsilon)<\infty$ for some $\varepsilon>0$, we have 
    \begin{equation}
        \mm_{\alpha,\beta} \ge C_{\alpha,\beta} \frac{\kappa^{cb}(S)}{D},\quad C_{\alpha,\beta}:= \min \left\{ \alpha^{-\frac{1}{\beta -1}} \left( \frac{1-\varepsilon}{3t_{\rm mix}^S(\varepsilon)} \right)^{\frac{\beta}{\beta - 1}} , \frac{1-\varepsilon}{2} \right\}
    \end{equation}
\end{theorem}
\subsubsection{Lower bound with fixed target time and precision}\label{subsec:lower bound fixed}

Now, we relate our framework to the classical lower bound problems introduced in \cite{kliesch2011dissipative}. The problem can be stated as follows:

\emph{Given an accessible unitary gate set $\mathcal{U}$ and a quantum channel $T_t$, what is the minimal number of gates we need to simulate $T_t$ within error $\delta$ in diamond norm for fixed $t>0$ and $\delta>0$?}

In general, the estimates of the number of gates defined in \eqref{eqn:convex_depth} should be expressed in terms of $t,\delta$, the underlying dimension of the system, and the number of local operators we need to express the jump operators and the Hamiltonian. Following the roadmap appearing in the previous section, we are able to prove the following theorem answering the above question:

\begin{theorem}\label{main result: lower bound fixed time}
    For a $D$-compatible resource set $S$ given accessible unitary $\mc U$ such that the assumption of Proposition \ref{ppp} holds, we have 
    \begin{equation}
        l^{\mc U}_{\alpha \tau^{\beta}}(T_{\tau}) \ge \tau \frac{\kappa^{cb}(S)}{8D t_{\rm{mix}}^S}
    \end{equation}
    for any $\alpha>0,\beta>1, \tau>0$ such that $t_{\rm mix}^S \alpha \tau^{\beta - 1} < \frac{1}{8}$, where $t_{\rm mix}^S$ is defined as $t_{\rm mix}^S(\frac{1}{2})$ in Definition \ref{def:mixing_time}.
\end{theorem}
\begin{proof}
    Recall \eqref{Lower bound: arbitrary time} in the proof of Proposition \ref{ppp}: if we choose a $D$-compatible resourse set $S$ given accessible unitary $\mc U$, for any $t>0, m\ge 1$, 
\begin{align*}
    C_S^{cb}(T_t) \le m D l^{\mc U}_{\alpha \left(\frac{t}{m}\right)^{\beta}}\left(T_{\frac{t}{m}}\right) + m^{1-\beta} \kappa^{cb}(S) \alpha t^{\beta}.
\end{align*}
If we take $\tau = \frac{t}{m}$, then
\begin{align*}
    C_S^{cb}(T_{m \tau}) \le m D l^{\mc U}_{\alpha \tau^{\beta}}(T_{\tau}) + m \kappa^{cb}(S) \alpha \tau^{\beta},
\end{align*}
from which we get a lower bound by optimizing $m$:
\begin{align} 
    l^{\mc U}_{\alpha \tau^{\beta}}(T_{\tau}) \ge \sup_{m \ge 1} \frac{C_S^{cb}(T_{m \tau}) - m \kappa^{cb}(S) \alpha \tau^{\beta}}{Dm}.
\end{align}
Recall that $t_{\rm mix}^S: = \inf\{t>0: C_S^{cb}(T_t) \ge \frac{1}{2}\kappa^{cb}(S)\}$. Choosing $m$ such that  \begin{equation}
   t_{\rm mix}^S \le m \tau \le 2 t_{\rm mix}^S\,,
\end{equation}
we have 
\begin{equation}\label{lower bound: fixed time}
    l^{\mc U}_{\alpha \tau^{\beta}}(T_{\tau}) \ge \tau \left(\frac{1}{2} - 2 t_{\rm mix}^S \alpha \tau^{\beta-1}\right)\frac{\kappa^{cb}(S)}{2D t_{\rm{mix}}^S}.
\end{equation}
Finally, if $t_{\rm mix}^S \alpha \tau^{\beta - 1} < \frac{1}{8}$, the desired lower bound is derived via \eqref{lower bound: fixed time}.
\end{proof}
Note that the flexibility of the choice of $\alpha,\beta,\tau$ can provide us lower bounds when the target time $\tau$ is small or large. If we fix $\alpha,\beta$ then the above estimate is a lower bound when $\tau$ is small, i.e., $\tau^{\beta - 1} < \frac{1}{8t_{\rm mix}^S \alpha }$. If we fix the target time $\tau>0$, then we can choose $\alpha$ small enough so that our lower bound is valid. As a corollary, we can provide a lower bound answering the general question proposed at the beginning of this subsection:
\begin{cor}\label{main result: lower bound fixed time and precision}
    A lower bound on the minimal number of gates we need to simulate $T_t$ for any $t>8\delta t_{\rm mix}$ within constant error $\delta>0$ is given by
    \begin{equation}
       l^{\mc U}_{\delta}(T_t) \ge t \frac{\kappa^{cb}(S)}{8Dt_{\rm{mix}}^S}
    \end{equation}
    where the resource set $S$ and the given accessible unitary $\mc U$ are $D$-compatible such that the assumption of Proposition \ref{ppp} holds.
\end{cor}

\subsection{An illustrating example: $n$-qubit Pauli noise} \label{subsec:example-section-4}
We provide a simple example illustrating Theorem \ref{main: lower bound mixing}. The model is defined on $n$-qubit system $\mc H_2^{\otimes n} = (\mb C^2)^{\otimes n}$. Denote $X_j$ as the Pauli-$X$ operator $X = \begin{pmatrix}
    0 & 1 \\
    1 & 0
\end{pmatrix}$ acting on the $j$-th system, i.e., 
\begin{align*}
    X_j = I_2 \otimes \cdots I_2 \otimes \underbrace{X}_\text{j-th system} \otimes I_2 \cdots \otimes I_2.
\end{align*}
We define $Y_j$ as the Pauli-$Y$ operator $Y = \begin{pmatrix}
    0 & i \\
    -i & 0
\end{pmatrix}$ acting on the $j$-th system in a similar manner. 
We consider the following 1-local model given by the Lindbladian generator
\begin{equation}\label{Pauli example}
    L(x):= \sum_{j=1}^n (L_{X_j} + L_{Y_j})(x),\ x\in \mb B(\mc H_2^{\otimes n}).
\end{equation}
Here we denote $L_{X_j}(x) = X_j x X_j - \frac{1}{2}(X_j^2 x + x X_j^2) = X_jxX_j - x$. Note that this model is a symmetric model and a Poisson model at the same time. We show later that in terms of the number of qubits, the approach for Poisson model provides a better simulation protocol, meaning less simulation cost. 
By the commuting property of $L_{X_j} + L_{Y_j}$ for different qubits $j$, the semigroup $T_t = \exp(tL)$ is given by independent copies of a single-qubit semigroup. In fact, suppose $\mc L(x): = (L_X + L_Y)(x), \ x\in \mb B(\mc H_2)$ and $\mc T_t = \exp(t\mc L)$, then 
\begin{equation}\label{eqn:separation}
    T_t = \mc T_t ^{\otimes n}.
\end{equation}
A natural choice of accessible unitary set $\mc U$ is given by individual Hamiltonian simulation of the generators,
\begin{equation}    
\mc U = \{\exp(itX_j),\, \exp(itY_j): |t|\le \tau, 1\leq j \le n\}.
\end{equation}
Applying Theorem \ref{main: lower bound mixing}, we demonstrate the lower bound of $\mm_{\alpha,\beta}$ given accessible unitary $\mc U$ in the following theorem:
\begin{theorem}\label{Lower bound: Pauli example}
    For any $\alpha > 0, \beta > 1$ with $\mm_{\alpha,\beta}< +\infty$. We have
    \begin{equation}
        \mm_{\alpha,\beta} \ge C_{\alpha,\beta} \frac{n}{\tau},\quad C_{\alpha,\beta}:= \frac{1}{2}\min \left\{ \alpha^{-\frac{1}{\beta -1}} \left( \frac{1}{6\ln 2} \right)^{\frac{\beta}{\beta - 1}} , \frac{1}{8} \right\}.
    \end{equation}
\end{theorem}
\begin{proof}
    First, we choose the resource set $S = \{X_j,Y_j: 1\le j \le n\}$. Then via Proposition \ref{Assumption: Trotter example}, $S$ and $\mc U$ satisfy \textbf{Assumptions (A), (B), (C)} with the conditional expectation given by
    \begin{equation}\label{conditional expectation: trace}
        E_{S'}(x): = \Tr(x) \frac{I_{2^n}}{2^n}.
    \end{equation}
    Then via Example \ref{example: ergodic}, we have 
    \begin{equation}
        \frac{n}{2} \le C_S(E_{S'}) \le C_S^{cb}(E_{S'}) = \kappa^{cb}(S) \le n.
    \end{equation}
    To show that $t_{\rm mix}^S(\varepsilon)$ is upper bounded by a universal constant, we consider the certificate operator $\sum_{i=1}^n X_i$. We have
$\sum_j \left(L_{X_j} + L_{Y_j}\right) \left(\sum_{i=1}^n X_i\right) = \sum_i L_{Y_i}(X_i) = -2 \sum_{i=1}^n X_i$, which implies
\begin{align*}
T_t\left(\sum_{i=1}^n X_i\right) = e^{-2t} \sum_{i=1}^n X_i\,.
\end{align*}
Recalling the definition of Lipschitz complexity,
\begin{equation}\label{eqn:sec-4-pauli-example-lower-bound-on-complexity}
\begin{aligned}
    & C_S^{cb}(T_t)\ge C_S(T_t)  = \sup \left\{\frac{\|T_t(x) - x\|_{\infty}}{\left\|\left|x\right|\right\|_S}:x \in \mb B(\mc H_2^{\otimes n}), \left\|\left|x\right|\right\|_S \neq 0\right\} \\
    & \ge (1-e^{-2t})\frac{ \|\sum_{j=1}^n X_j\|}{\left\|\left| \sum_{j=1}^n X_j\right|\right\|_S} = \left(1-e^{-2t}\right)\frac{n}{2} \ge \frac{1-e^{-2t}}{2} \kappa^{cb}(S). 
\end{aligned}
\end{equation}
Here $\left\|\left| \sum_{j=1}^n X_j\right|\right\|_S = 2, \left\| \sum_{j=1}^n X_j\right\| = n$ follows from direct calculation. Using the definition of $t_{\rm mix}^S$ in Definition \ref{def:mixing_time}, we have 
\begin{align*}
    C_S^{cb}(T_t)\ge \frac{\kappa^{cb}(S)}{4},\quad t\ge \frac{\ln 2}{2} \implies t_{\rm mix}^S\left(\frac{3}{4}\right) \le \frac{\ln 2}{2}.
\end{align*}
Finally, we apply Theorem \ref{main: lower bound mixing} with $\varepsilon = \frac{3}{4}$ there, we have 
\begin{equation}
        \mm_{\alpha,\beta} \ge C_{\alpha,\beta} \frac{n}{\tau},\quad C_{\alpha,\beta}:= \frac{1}{2}\min \left\{ \alpha^{-\frac{1}{\beta -1}} \left( \frac{1}{6\ln 2} \right)^{\frac{\beta}{\beta - 1}} , \frac{1}{8} \right\}
    \end{equation}
\end{proof}
Similar to what we have done in Section~\ref{sec:upper}, we now discuss upper bounds of $\mm_{\alpha,\beta}$ via Trotter and Poisson approaches.
\subsubsection{Trotter approach}
Motivated by the Trotter estimate, 
For any $t>0$, $\alpha >0$ and $\beta =3$, we have shown before in Theorem \ref{main: symmetric} that there exists an approximate random unitary channel $\Phi_t$ such that
\begin{equation}
    \|\mc T_t - \Phi_t\|_{\diamond} < \alpha t^3.
\end{equation}
The length of $\Phi_t$ is bounded above by $\frac{C_{\alpha}}{\tau}$ where $C_{\alpha}$ is a constant depends on $\alpha.$ Applying a telescoping argument,
\begin{equation}
    \|\mc T_t^{\otimes n} - \Phi_t^{\otimes n}\|_{\diamond} < n \alpha t^3
\end{equation}
The length of $\Phi_t^{\otimes n}$ is bounded above by $\frac{nC_{\alpha}}{\tau}$. 
This gives us an upper bound on $\mm_{n\alpha,3}$, the uniform maximal simulation cost for $T^{\otimes n}_t$, namely 
\begin{equation}
    \mm_{n\alpha,3} \le \frac{nC_{\alpha}}{\tau}.
\end{equation}
Applying Lemma \ref{different parameter}, for any $\alpha>0$ and $\beta = 3$, we have 
\begin{equation}\label{upper bound: commuting Pauli}
    \mm_{\alpha,3} \le  \frac{n^{3/2}\wt C_{\alpha}}{\tau}.
\end{equation}
The above upper bound improves upon the result in Theorem \ref{main: symmetric}, where $\mm_{\alpha,3} \le  \frac{n^{5/2}C_{\alpha}}{\tau}$. This key improvement stems from the equality in \eqref{eqn:separation} due to the local properties of $X_j$ and $Y_j$. In contrast, the proof of Theorem \ref{main: symmetric} introduces additional error when decomposing $\exp(tL)$ into products of $\exp\left(tL_{a_j}\right)$, which in general do not exhibit the same commuting property.

\subsubsection{Poisson approach} \label{sec: example:Poisson}
Now we exhibit another upper bound of $\mm_{\alpha,\beta}$ using the Poisson approach. Note that $X_j,Y_j$ are also unitaries, our model $T_t = \mc T_t^{\otimes n}$ fits the form of a Poisson model from Section \ref{sec:nonsymm_case}. Using the Poisson approach, we can provide a scheme with upper bound behaving like $\frac{n}{\tau}$. This means our lower bound captures the correct behavior of $n,\tau$.

In fact, we note that \begin{align*}
    \exp(itX) = \begin{pmatrix}
        \cos t & i\sin t \\
        i \sin t & \cos t
    \end{pmatrix}, \quad \exp(itY) = \begin{pmatrix}
        \cos t & -\sin t \\
        \sin t & \cos t
    \end{pmatrix}.
\end{align*}
We have $\exp(i\frac{\pi}{2}X) = iX, \exp(i\frac{\pi}{2}Y) = -iY$. Therefore, there exist $u_0, u_1 \le \tau$ and integers $k_0, k_1 \le \frac{\pi}{\tau}$ such that 
\begin{align}
    \adm^{k_0}_{\exp(iu_0X_j)} = \adm_{X_j}, \quad \adm^{k_1}_{\exp(iu_1X_j)} = \adm_{Y_j},
\end{align}
where $\adm^{k}_U$ is $k$ iterations of $\adm_U$. A candidate for approximation is given as 
\begin{align*}
    \Phi = \frac{1}{2n} \sum_{j=1}^n (\adm_{X_j} + \adm_{Y_j}),
\end{align*}
which has length at most $\frac{\pi}{\tau}$. $N$ uses of $\Phi$ at time $t$ following a certain probability is defined as follow:
\begin{align*}
   \Phi_t^N:= \frac{1}{\sum_{k=0}^N \frac{(2nt)^k}{k!}} \sum_{k=0}^N \frac{(2nt)^k}{k!} \Phi^k.
\end{align*}
Using the Poisson estimate in Lemma \ref{key lemma: Poisson estimate}, we have
\begin{align*}
    \|T_t - \Phi_t^N \|_{\diamond} \le 2e^{-2nt} e^{(N+1)(1-\ln(N+1))} (2nt)^{N+1},\quad 2nt < N+1. 
\end{align*}
Note that $(1-\ln(N+1))$ tend to $-\infty$ as $N$ tends to infinity. If we choose $N = nC_{\alpha,\beta} $ with $C_{\alpha,\beta}$ reasonably large(the explicit form is given in Proposition \ref{main: Poisson}), we have
\begin{equation}
    2e^{-2nt} e^{(N+1)(1-\ln(N+1))} (2nt)^{N+1} \le \alpha t^{\beta},\quad t \le (\frac{2}{\alpha})^{1/\beta}.
\end{equation}
The length of $\Phi_t^N$ is given by $C_{\alpha,\beta}\pi \frac{n}{\tau} $ thus $\mm_{\alpha,\beta} \le C_{\alpha,\beta}\pi \frac{n}{\tau}$ which scales like $\frac{n}{\tau}$. Thus the Poisson approach provides a tight upper bound in terms of $n$.

\begin{rmk}
    Using the result in Corollary \ref{main result: lower bound fixed time and precision}, we have a lower bound 
    \begin{equation}
        l^{\mc U}_{\delta}(T_t) \ge \Omega(nt),
    \end{equation}
    and the upper bound provided by the above Poisson approach is 
    \begin{equation}
        l^{\mc U}_{\delta}(T_t) \le \mc O\left(\log(\frac{1}{\delta})nt\right).
    \end{equation}
    In Appendix \ref{app:hamsim-for-pauli-noise}, using one ancilla qubit, a factor of $\log(\frac{1}{\delta})$ can be removed. Thus lower bound and upper bound are both given by $\Theta(nt)$ using the result in Section \ref{sec:non-unital}.
\end{rmk}

\section{Estimates for the uniform maximal simulation cost of non-unital quantum dynamics.}\label{sec:non-unital}

In this section, we consider the simulation of non-unital quantum dynamics on the system $\mc H_A$. Unlike the unital case, achieving an approximation in this case with an error of $O\left(t^{\beta}\right)$ with $\beta > 1$ is not possible if solely considering a set of accessible unitary gates without an ancillary quantum environment. We summarize the no-go result in the following proposition. 
\begin{prop}\label{nogo: non-unital}
 Let $T_t$ be a completely positive and trace preserving semigroup but non-unital, meaning $T_t(I_A)\neq I_A$, where $I_A$ is the identity operator on $\mc H_A$. Then for any random unitary channel $\Phi:\mb B(\mc H_A) \to \mb B(\mc H_A)$, there exists $t_0>0$ and $C>0$ such that 
   \begin{equation}
       \left\|T_t - \Phi\right\|_{\diamond} \ge Ct,\quad \forall t\in [0,t_0].
   \end{equation}
\end{prop}
The above proposition shows that  if we were to proceed with the same approach as in the unital case, the diamond error of the approximating channel is at least linear in $t$ and thus $\beta\not> 1$.
\begin{proof}
   Because $T_t = \exp(tL)$ is non-unital, $L(I_A)\neq 0$. Thus $\|L(I)\|_1 \neq 0$. Note that for random unitary channel $\Phi$, we always have $\Phi(I_A) = I_A$. This implies
   \begin{align*}
       \|T_t(I_A) - \Phi(I_A)\|_1 & = \left\|\sum_{k\ge 1} \frac{t^kL(I_A)^k}{k!}\right\|_1 \ge t \|L(I_A)\|_1 - \sum_{k\ge 2} \frac{t^k \|L(I_A)\|_1^k}{k!} \\
       & \ge  t \|L(I_A)\|_1\big(1 -  t \|L(I_A)\|_1\exp( t \|L(I_A)\|_1)\big) \ge \frac{t}{2}\|L(I_A)\|_1.
   \end{align*}
   for any $t\leq t_0:=\frac{1}{c(d) \|L(I)\|_1}$ where $c(d)$ is a constant depending on the dimension of $\mc H_A$. Since $\|T_t(I) - \Phi(I)\|_{\diamond}\geq \|T_t(I) - \Phi(I)\|_1$, this concludes the proof.

\end{proof}

\subsection{Environment-assisted simulation cost}\label{sec:environment-assisted simulation cost}
According to Proposition \ref{nogo: non-unital}, it is not possible to approximate non-unital quantum dynamics using only unitaries acting on the quantum system $\mathcal{H}_A$. To approximate non-unital dynamics, we need to include two additional steps beyond the unitary approximation: Encoding and decoding via an assisting environment. For simplicity, the assisting environment system $\mc H_E$ is denoted as $E$ and $\mc H_A$ is denoted as $A$.

Our general simulation protocol for non-unital dynamics is as follows:
\begin{itemize}
    \item An encoding scheme from $A$ to $A\otimes E$, where $E$ is another quantum system which we call environment. 
    
    In this paper, we mainly use the encoding scheme given by the quantum channel $\mathcal{E}(\rho) := \rho \otimes \ketbra{0^k}$, where $k$ is the size of $E$. This encoding scheme is commonly used in previous works on Lindblad simulation~\cite{cleve_et_al:LIPIcs.ICALP.2017.17}. The value of $k$ should be chosen appropriately to achieve a high-order simulation scheme~\cite{PRXQuantum.5.020332}.

    \item A set of accessible unitaries $\mc U $ acting on $A\otimes E$. Analogous to the simulation of unitary dynamics, we assume that the accessible unitaries $\mathcal{U}$ acting on $A\otimes E$ are straightforward and easily implemented, such as those given by the time evolution of Hamiltonians.  
    \item A decoding scheme from  $A\otimes E$ to $A$, modelled as a quantum channel $\mc D: \mb B(A \otimes E) \to \mb B(A)$. In this paper, we mainly use the decoding scheme given by the partial trace $\Tr_E$, which can be implemented by measuring and resetting the ancilla qubits of $E$.
\end{itemize}
The quantum channel that implements the above three steps is summarized in Figure \ref{fig:a2}. Although we perform a measurement in the decoding process, it is a trace-out operation and does not reduce the success probability.

\begin{figure}[bthp]
\color{black}
\begin{center}
\begin{minipage}{.7\textwidth}
\begin{center}
\begin{tikzpicture}[y=-1cm,scale=1.5]
    \draw (4,0) -- (0,0) node[left] {$\lvert0\rangle_E$};
    \draw (4,-0.05) -- (0,-0.05) ;
    \draw (4,0.05) -- (0,0.05) ;
    \draw (4.5,1) -- (0,1) node[left] {$\lvert\psi\rangle_A$};
\draw[fill=white] (0.35,-0.25) rectangle  (1.35, 1.25);
\node at (.64, .5)   {$\mathrm{enc}$};
    \node at (1.05, .45)   {\tiny{e.g.}};
    \draw[fill=blue!20] (.9, -.25)   decorate [decoration={snake, amplitude=0.5mm, segment length=12mm}] { --(.9, 1.25)} -- (1.35, 1.25) -- (1.35, -.25) ;
\draw (.93,-0.05) -- (1.33,-0.05) ;
    \draw (.935,0) -- (1.33,0);
    \draw (.93,0.05) -- (1.33,0.05) ;
    \draw (.91,1) -- (1.33,1) ;
    \draw[fill=white] (1.6,-0.25) rectangle node {$U_1$} (2.1, 1.25);
    \draw[fill=white] (3.05,-0.25) rectangle node {$U_m$} (3.55, 1.25);
    \draw[draw=none,fill=white] (2.35, -.1)  rectangle node {$\cdots$}  (2.8, .1);
    \draw[draw=none,fill=white] (2.35, 0.9)  rectangle node {$\cdots$}  (2.8, 1.1);
    \draw[fill=white] (3.8, -.25) rectangle  (4.25, .25);
        \draw (3.85,0) arc[start angle=-125, end angle=-55, radius=.3];
        \draw[>=stealth,->] (4.0, .0255) -- (  4.195 , -.1108  );
    \draw[fill=white] (3.8, .75) rectangle node {{$\scriptstyle{\mathrm{dec}}$}} (4.25, 1.25);
    \draw (4.0,.25) -- (4.0, .75);
    \draw (4.05,.25) -- (4.05, .75);

\end{tikzpicture}
\end{center}
\end{minipage}
\end{center}
\normalcolor
\caption{Quantum circuit for non-unital quantum dynamics simulation (a single step).}
\label{fig:a2}
\end{figure}
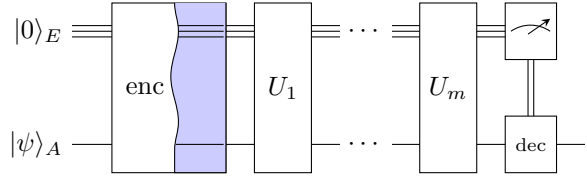

Based on the above simulation protocal, we modify Definition \ref{eqn:convex_depth} and define the simulation cost of non-unital quantum dynamics as follows:
\begin{defi}
    Suppose $\mc E: \mb B(A) \to \mb B(A\otimes E)$ and 
    $\mc D: \mb B(A\otimes E) \to \mb B(A)$ are fixed. Then the simulation cost of non-unital $T_t$, with accessible unitary $\mc U_{A\otimes E}$ acting on $A\otimes E$ is defined as 
    \begin{equation}
        \widehat{l}^{\mc U_{A\otimes E}}_{\delta}(T_t) = \inf\left\{\sum_{j=1}^M m_j\middle | \exists \mu_j \in \mc P(\mc U_{A\otimes E}^{m_j}), \left\|T_t - \prod_{j=1}^M \mc D \circ \Phi_{\mu_j} \circ \mc E\right\|_{\diamond} < \delta \right\}\,,
    \end{equation}
    where $m_j$ is the circuit depth of $\Phi_{\mu_j}$.
\end{defi}
Note that the above definition provides an upper bound for the circuit depth defined in Equation \ref{eqn:convex_depth} when there is no environment. Specifically, when there is no environment, encoding and decoding are trivial $\mathcal{E} = \mathcal{D} = \text{id}$. Then, $\prod_{j=1}^M \mathcal{D} \circ \Phi_{\mu_j} \circ \mathcal{E}$ is equal to $\Phi_{\mu}$ if one takes $\mu = \mu_1 \times \cdots \times \mu_M$ and each $\mu_j$ with length $m_j$.


Similar to the uniform maximum simulation length in the unital case from Definition \ref{def:M}, we aim to approximate $T_t$ with an error of $O(t^{\beta})$ for $\beta > 1$. We define the modified maximal uniform length:
\begin{defi}
    For $\alpha >0,\ \beta > 1$, and non-unital quantum dynamics $T_t$ and a set of accessible unitaries $\mc U_{A\otimes E}$, define the maximum uniform length with environment $\widehat{\mm}_{\alpha,\beta}$ by 
    \begin{equation}
        \widehat{\mm}_{\alpha,\beta}(\{T_t\}_{t\ge 0}):= \sup_{0\le t \le (\frac{2}{\alpha})^{1/\beta}}  \widehat{l}^{\mc U_{A\otimes E}}_{\alpha t^{\beta}}(T_t).
    \end{equation}
\end{defi}
When the context is unambiguous, we omit the dependence of $\{T_t\}_{t\ge 0}$ for simplicity. We refer the reader to \cite{PRXQuantum.5.020332} for a general approach to construct the accessible unitary set given by a single Hamiltonian.

\subsection{Lower bound via environment-assisted resource-dependent complexity}

To establish the lower bound of $\widehat{\mm}_{\alpha,\beta}$ for non-unital quantum dynamics, we introduce a compatible Lipschitz complexity that incorporates the environment. One might initially think to define Lipschitz complexity directly on the extended Hilbert space $A \otimes E$ by tensoring with identity on the environment. However, this direct extension is not suitable in our case, as it does not account for the encoding and decoding processes. Specifically, we select a resource set $S_{A\otimes E}$, a set of operators on the combined system $A\otimes E$, which is compatible with the accessible unitary set $\mc{U}_{AE}$.

\begin{defi}
    Suppose $\mc E: \mb B(A) \to \mb B(A\otimes E)$ and 
    $\mc D: \mb B(A\otimes E) \to \mb B(A)$ are fixed encoding and decoding schemes. For a resource set $S_{A\otimes E} \subseteq \mb B(A\otimes E)$ and a quantum channel $\Phi: \mb B(A) \to \mb B(A)$, we define 
    \begin{equation}
        \widehat{C}_{S_{A\otimes E}}(\Phi): = \|(\mc E (\Phi - id_A) \mc D)^*: (\mb B(A\otimes E), \||\cdot|\|_{S_{A\otimes E}}) \to \mb B(A\otimes E)\|_{\rm op}.
    \end{equation}
    Similarly, we can define the complete version $ \widehat{C}^{cb}_{S_{A\otimes E}}(\Phi)$ by defining the norm of arbitrary amplifications as in \eqref{eqn:C_cb_phi}. 
\end{defi}
Our goal is to show a crucial lemma as in Lemma \ref{dd} that allows us to relate the Lipschitz complexity to the lower bound of the length and the expected length given the resource set expressed through the Lipschitz constant, under resonable assumptions on $\mc E, \mc D$, the resource set $S_{A\otimes E}$, and accessible unitaries $\mc U_{A\otimes E}$. Then, one can establish a similar lower bound as in Section \ref{sec:lower_bound}.

First, we demonstrate that $\widehat{C}_{S_{A\otimes E}}$ and $\widehat{C}^{cb}_{S_{A\otimes E}}$ shares the same fundamental properties as $C_S$ and $C^{cb}_S$. Since the proof relies solely on the fact that the definition can be rewritten in the form of norms, the argument proceeds similarly to the case without environment.
\begin{lemma} Given any resource set $S_{A\otimes E} \subseteq \mb B(A\otimes E)$, $\widehat{C}_{S_{A\otimes E}}$ and $\widehat{C}^{cb}_{S_{A\otimes E}}$ satisfy:
\begin{enumerate}
	\item $\widehat{C}^{cb}_{S_{A\otimes E}}(id_A)=\widehat{C}_{S_{A\otimes E}}(id)=0$. 
	\item Subadditivity under concatenation: 
 \[
 \widehat{C}_{S_{A\otimes E}}(\Phi \Psi) \le \widehat{C}_{S_{A\otimes E}}(\Phi)+
 \widehat{C}_{S_{A\otimes E}}(\Psi),\quad \widehat{C}^{cb}_{S_{A\otimes E}}(\Phi \Psi) \le
 \widehat{C}^{cb}_{S_{A\otimes E}}(\Phi)+\widehat{C}^{cb}_{S_{A\otimes E}}(\Psi)
 \]
 for any quantum channels $\Phi,\Psi : \mb B(A) \to \mb B(A)$.
	\item Convexity: For any probability distribution $\{p_i\}_{i\in I}$ and any quantum channels $\{\Phi_i\}_{i \in I}$ acting on the operators on $A$, we have 
\begin{equation}
\widehat{C}^{cb}_{S_{A\otimes E}}\left(\sum_{i\in I} p_i\Phi_i\right) \le \sum_{i\in I} p_i\widehat{C}^{cb}_{S_{A\otimes E}}\left(\Phi_i\right),\quad \widehat{C}_{S_{A\otimes E}}\left(\sum_{i\in I} p_i\Phi_i\right) \le \sum_{i\in I} p_i\widehat{C}_{S_{A\otimes E}}\left(\Phi_i\right).
\end{equation} 
\end{enumerate}
\end{lemma}

Similar to Remark \ref{rem:continuity}, we show the following continuity argument for $\widehat{C}$.
\begin{lemma}\label{lem:continuity}
    Suppose $E_{\rm fix}^A$ is a conditional expectation on the system $A$ and that $\Phi,\Psi: \mb B(A) \to \mb B(A)$ are two quantum channels with 
    \begin{equation}\label{eqn:condition_A}
       \Phi^* \circ E_{\rm fix}^A = \Psi^* \circ E_{\rm fix}^A = E_{\rm fix}^A.
    \end{equation}
    Then, we have 
    \begin{equation}
        \left|\widehat{C}^{cb}_{S_{A\otimes E}}(\Phi)-\widehat{C}^{cb}_{S_{A\otimes E}}(\Psi)\right|\le \widehat{C}^{cb}_{S_{A\otimes E}}\left(\left(E_{\rm fix}^A\right)^*\right)\|\Phi - \Psi\|_{\diamond}.
    \end{equation}
\end{lemma}
\begin{proof}
We notice that 
\begin{align*}
    \Phi - id_A = \Phi - \Psi + \Psi- id_A = (id_A - \left(E_{\rm fix}^A\right)^*)(\Phi - \Psi) + \Psi- id_A\,.
\end{align*}
By the definition, subadditivity, and submultiplicativity of semi-norms, we have 
\begin{align*}
    \widehat{C}^{cb}_{S_{A\otimes E}}(\Phi)& = \|(\mc E (\Phi - id_A) \mc D)^*: (\mb B(A\otimes E), \||\cdot|\|_{S_{A\otimes E}}) \to \mb B(A\otimes E)\|_{cb} \\
    & = \|(\mc E \big((id_A - \left(E_{\rm fix}^A\right)^*)(\Phi - \Psi) + \Psi- id_A \big) \mc D)^*: (\mb B(A\otimes E), \||\cdot|\|_{S_{A\otimes E}}) \to \mb B(A\otimes E)\|_{cb} \\
    & \le \widehat{C}^{cb}_{S_{A\otimes E}}\left(\left(E_{\rm fix}^A\right)^*\right) \|\mc E(\Phi - \Psi)\mc D\|_{\diamond} + \|(\mc E (\Psi- id_A) \mc D)^*: (\mb B(A\otimes E), \||\cdot|\|_{S_{A\otimes E}}) \to \mb B(A\otimes E)\|_{cb} \\
    & \le \widehat{C}^{cb}_{S_{A\otimes E}}(\Psi) + \widehat{C}^{cb}_{S_{A\otimes E}}\left(\left(E_{\rm fix}^A\right)^*\right)\|\Phi - \Psi\|_{\diamond}.
\end{align*}
Exchanging the order of $\Phi$ and $\Psi$, we can also show
\[
\widehat{C}^{cb}_{S_{A\otimes E}}(\Psi)\leq \widehat{C}^{cb}_{S_{A\otimes E}}(\Phi) + \widehat{C}^{cb}_{S_{A\otimes E}}\left(\left(E_{\rm fix}^A\right)^*\right)\|\Phi - \Psi\|_{\diamond}\,,
\]
which concludes the proof.
\end{proof}

To establish the lower bound like in Section \ref{sec:lower_bound}, we need further assumptions on $S_{A\otimes E}$ and $\mc U_{A\otimes E}$. Similar to Definition \ref{compatible: U and S}, we need the following definition of compatible $\mc U_{A\otimes E}$ and $S_{A\otimes E}$ under the encoding and decoding schemes:
\begin{defi}[$\widehat{D}$-compatibility with environment]\label{compatible: U and S with environment}
    Suppose $\mc E: \mb B(A) \to \mb B(A\otimes E)$ and 
    $\mc D: \mb B(A\otimes E) \to \mb B(A)$ are fixed encoding and decoding schemes. For an accessible unitary set $\mc U_{A\otimes E}$ acting on $A\otimes E$ and a resource set $S_{A\otimes E} \subseteq \mb B(A\otimes E)$, we say that $\mc U_{A\otimes E}$ and $S_{A\otimes E}$ are $\widehat{D}$-compatible with $\widehat{D}>0$ if for any $u\in \mc U_{A\otimes E}$, we have 
    \begin{equation}
        \widehat{C}^{cb}_{S_{A\otimes E}}(\mc D \circ \adm_u \circ \mc E) \le \widehat{D}.
    \end{equation}
\end{defi}
\begin{tcolorbox}
\textbf{Assumption (O)}: $\mc D \circ \mc E = id_A$.

\textbf{Assumption (B')}: There exists a conditional expectation $E_{\rm fix}^A$ acting on system $A$ such that for any $u\in \mc U_{A\otimes E}$, we have $(\mc D \circ \adm_u \circ \mc E)^* \circ E_{\rm fix}^A  = E_{\rm fix}^A$. 

\textbf{Assumption (C')}: $\mc U_{A\otimes E}$ and $S_{A\otimes E}$ are $\widehat{D}$-compatible.
\end{tcolorbox}
We note that Assumptions \textbf{(B'),(C') }is similar to Assumptions \textbf{(B),(C)} in the unital case. The condition $(\mc D \circ \adm_u \circ \mc E)^* \circ E_{\rm fix}^A = E_{\rm fix}^A$ is used to ensure that the approximation channel also satisfies property \eqref{eqn:condition_A}. The $\widehat{D}$-compatibility condition is used to ensure the boundedness property in Lemma \ref{lem:continuity}.


\begin{rmk}
    We provide some sufficient conditions for the assumptions \textbf{(O)}, \textbf{(B'),(C')} to hold. To argue about existence of the conditional expectation $E_{\rm fix}^A$, we can choose it as the conditional expectation onto the fixed point algebra of $T_t^*$ when it is a true von Neumann algebra. A sufficient condition for the fixed point set to be a von Neumann algebra is given by the following.
\end{rmk}
\begin{prop}\label{characterization: fixed point}
        Suppose $T_t$ has a unique fixed state, or there exists a faithful invariant state for $T_t$, then the fixed point set for $T_t^*$ is a von Neumann algebra.
\end{prop}
    \begin{proof}
        If $T_t$ has a unique fixed state, then by \cite[Theorem 2]{burgarth2013ergodic}, the fixed point set of the adjoint map $T_t^*$ is characterized as 
        \begin{equation}
            \{c I_A: c \in \mb C\},
        \end{equation}
        which is certainly a von Neumann algebra. If there exists a faithful invariant state for $T_t$, then by \cite[Proposition 8]{2016structure}, the fixed point set of $T_t^*$ is given by a von Neumann algebra, specifically
        \begin{equation}
            \{V_j,V_j^*,H\}',
        \end{equation}
        where $L(x) = i[H,x] + \sum_j \big(V_j^*xV_j - \frac{1}{2}(V_j^*V_j x + x V_j^*V_j)\big)$.
    \end{proof} 

Based on the above assumption and the continuity statement in Lemma \ref{lem:continuity}, we demonstrate a result on the relationship of the environment-assisted Lipschitz complexity and the lower bound on the simulation length and the expected length, which is similar to Lemma \ref{dd}.
\begin{lemma}
    Under \textbf{Assumptions (O), (B'),(C')}, for any quantum channel $\Phi$ with $\Phi^* \circ E_{\rm fix}^A = E_{\rm fix}^A$ and any $\delta>0$, we have
    \begin{equation}
        \widehat{C}^{cb}_{S_{A\otimes E}}(\Phi) \le \widehat D \widehat{l}^{\mc U_{A\otimes E}}_{\delta}(\Phi)+\widehat{C}^{cb}_{S_{A\otimes E}}\left(\left(E_{\rm fix}^A\right)^*\right)\delta.
    \end{equation}
\end{lemma}
\begin{proof}
The proof is similar as the proof of Lemma \ref{dd}. Suppose there exist $\mu_1\cdots \mu_j \cdots \mu_M$ with $\Phi_{\mu_j} = \mb E_{\mu_j} \adm_{u_1\cdots u_{m_j}} $ such that
\begin{align*}
    \left\|\Phi - \prod_{j=1}^M \mc D\circ \Phi_{\mu_j} \circ \mc E\right\|_{\diamond} < \delta.
\end{align*}
    
Then, similar to the argument in Lemma \ref{dd}, we have 
\begin{align*}
 \widehat{C}^{cb}_{S_{A\otimes E}}(\Phi)& = \left\|(\mc E (\Phi - id_A) \mc D)^*: (\mb B(A\otimes E), \||\cdot|\|_{S_{A\otimes E}}) \to \mb B(A\otimes E)\right\|_{cb} \\
& \le \left\|(\mc E \big(\Phi - \mc D\circ \prod_{j=1}^M \Phi_{\mu_j} \circ \mc E \big) \mc D)^*: (\mb B(A\otimes E), \||\cdot|\|_{S_{A\otimes E}}) \to \mb B(A\otimes E) \right\|_{cb} \\
& + \left\|(\mc E \big( \mc D\circ \prod_{j=1}^M \Phi_{\mu_j} \circ \mc E - id_A) \mc D\big)^*: (\mb B(A\otimes E), \||\cdot|\|_{S_{A\otimes E}}) \to \mb B(A\otimes E)\right\|_{cb} \\
& \le \widehat{C}^{cb}_{S_{A\otimes E}}\left(\left(E_{\rm fix}^A\right)^*\right)\delta + \widehat{C}^{cb}_{S_{A\otimes E}}(\mc D\circ \prod_{j=1}^M \Phi_{\mu_j} \circ \mc E) \\
& \le \widehat{C}^{cb}_{S_{A\otimes E}}\left(\left(E_{\rm fix}^A\right)^*\right) \delta + \widehat{D} \sum_{j=1}^M m_j = \widehat{C}^{cb}_{S_{A\otimes E}}\left(\left(E_{\rm fix}^A\right)^*\right) \delta + \widehat{D} \widehat{l}^{\mc U_{A\otimes E}}_{\delta}(\Phi).
\end{align*}
\end{proof}
\noindent Now, we are ready to state the lower bound for $\widehat{\mm}_{\alpha,\beta}$ which can be derived similarly as Theorem \ref{main: lower bound mixing}. 
\begin{theorem} \label{lower bound: non-unital}
     Under \textbf{Assumptions (O), (B'),(C')}, suppose $\{T_t\}_{t\ge 0}$ is a non-unital semigroup acting on operators on $A$ and $E_{\rm fix}^A$ is a conditional expectation on the system $A$ such that \begin{align*}
        T_t^* \circ E_{\rm fix}^A= E_{\rm fix}^A,\quad \forall t>0.
    \end{align*}
    Define
    \begin{equation}\label{def:induced mixing time non-unital}
        t_{\rm mix}^{S_{A\otimes E}}(\varepsilon):= \inf\left\{t \ge 0: \widehat{C}^{cb}_{S_{A\otimes E}} (T_t) \ge (1-\varepsilon) \widehat{C}^{cb}_{S_{A\otimes E}}\left(\left(E_{\rm fix}^A\right)^*\right)\right\}
    \end{equation}
    If $t_{\rm mix}^{S_{A\otimes E}}(\varepsilon) < +\infty$ is finite, then we have
   \begin{equation}
      \widehat{\mm}_{\alpha,\beta} \ge C_{\alpha,\beta} \frac{\widehat{C}^{cb}_{S_{A\otimes E}}\left(\left(E_{\rm fix}^A\right)^*\right)}{\widehat D},\quad C_{\alpha,\beta}:= \min \left\{ \alpha^{-\frac{1}{\beta -1}} \left( \frac{1-\varepsilon}{3t_{\rm mix}^{S_{A\otimes E}}(\varepsilon)} \right)^{\frac{\beta}{\beta - 1}} , \frac{1-\varepsilon}{2} \right\}.
   \end{equation}
\end{theorem}
With the above theorem, we have a general routine to compute a concrete lower bound of $\widehat{\mm}_{\alpha,\beta}$ for a non-unital semigroup $T_t$ with accessible unitary $\mc U_{A\otimes E}$:
\begin{itemize}
\item We first find the fixed point set of $T_t^*$. Then one can choose $E_{\rm fix}^A $ as the conditional expectation onto the fixed point algebra of $T_t^*$. This ensures
\begin{equation}
    T_t^* \circ E_{\rm fix}^A = E_{\rm fix}^A.
\end{equation}
\item Based on the structure of $\mc U_{A\otimes E}$, we can choose a certain compatible resource set $S_{A\otimes E}$ which satisfies Assumption (B').

\item Finally, applying Theorem \ref{lower bound: non-unital}, we can get a lower bound estimate of $\widehat \mm_{\alpha,\beta}$ if we have an upper bound of $t_{\rm mix}^{S_{A\otimes E}}$ and lower estimate of $\widehat{C}^{cb}_{S_{A\otimes E}}\left(\left(E_{\rm fix}^A\right)^*\right)$.

\noindent One upper bound of $t_{\rm mix}^{S_{A\otimes E}}$ is given by the usual mixing time, as stated in the following proposition.
\end{itemize}
\begin{prop}\label{upper bound: mixing time}
   Under the assumption of Theorem \ref{lower bound: non-unital}, define the usual mixing time of $T_t$ by
   \begin{equation}
    t_{\rm mix}(\varepsilon):= \inf\{t \ge 0: \left\|T_t - \left(E_{\rm fix}^A\right)^*\right\|_{\diamond} \le \varepsilon\}
\end{equation} 
Then $t_{\rm mix}$ is an upper bound of $t_{\rm mix}^{S_{A\otimes E}}$, i.e., 
\begin{equation}
    t_{\rm mix}^{S_{A\otimes E}}(\varepsilon) \le t_{\rm mix}(\varepsilon),\quad \forall \varepsilon \in (0,1).
\end{equation}
\end{prop}
\begin{proof}
    Note that $T_t^* \circ E_{\rm fix}^A = E_{\rm fix}^A,\quad \forall t>0$, which means we can say that 
    \begin{align*}
        (T_t^* - E_{\rm fix}^A)\circ(id_A - E_{\rm fix}^A) = T_t^* - E_{\rm fix}^A,\quad \forall t>0.
    \end{align*}
    Then apply Lemma \ref{lem:continuity}, by choosing $\Phi = T_t, \Psi = (E_{\rm fix}^A)^*$
    \begin{align*}
        \left|\widehat{C}^{cb}_{S_{A\otimes E}}(T_t)-\widehat{C}^{cb}_{S_{A\otimes E}}((E_{\rm fix}^A)^*)\right|\le \widehat{C}^{cb}_{S_{A\otimes E}}\left((E_{\rm fix}^A)^* \right)\|T_t - (E_{\rm fix}^A)^*\|_{\diamond}.
    \end{align*}
    Therefore, \begin{align*}
        \widehat{C}^{cb}_{S_{A\otimes E}}(T_t) \ge (1-\varepsilon) \widehat{C}^{cb}_{S_{A\otimes E}}((E_{\rm fix}^A)^*)
    \end{align*}
    if $\|T_t - \left(E_{\rm fix}^A\right)^*\|_{\diamond} \le \varepsilon$, thus proving that $t_{\rm mix}^{S_{A\otimes E}}(\varepsilon) \le t_{\rm mix}(\varepsilon)$.
\end{proof}
Before we discuss illustrating examples, we remark that for fixed time and error, we have a similar lower bound as in Section \ref{subsec:lower bound fixed}:
\begin{theorem}\label{main: lower bound non-unital fixed time}
Under the assumption of Theorem \ref{lower bound: non-unital}, we have
    \begin{equation}
        \widehat{l}^{\mc U_{A\otimes E}}_{\alpha \tau^{\beta}}(T_{\tau}) \ge \tau \frac{\widehat{C}^{cb}_{S_{A\otimes E}}\left(\left(E_{\rm fix}^A\right)^*\right)}{8\widehat D t_{\rm mix}^{S_{A\otimes E}}}
    \end{equation}
    for any $\alpha>0,\beta>1, \tau>0$ such that $t_{\rm mix}^{S_{A\otimes E}} \alpha \tau^{\beta - 1} < \frac{1}{8}$, where $t_{\rm mix}^{S_{A\otimes E}}$ is defined as $t_{\rm mix}^{S_{A\otimes E}}(\frac{1}{2})$ in \eqref{def:induced mixing time non-unital}.    
\end{theorem}

\subsection{Noise generated by non-symmetric and non-unitary jump operators}\label{sec:non-sym-non-uni}
In Sections \ref{sec:symm_case} and \ref{sec:nonsymm_case}, we studied the simulation of symmetric or unitary jump operators using unitary gates. Now, we consider a more general case, involving non-symmetric and non-unitary jump operators. Suppose the non-unital quantum dynamics $T_t = \exp(tL)$ is purely dissipative, i.e., 
    \begin{equation}\label{model: non-symmetric and non-unitary}
        L(x):= \sum_{j=1}^m L_{V_j} = \sum_{j=1}^m \left[V_j x V_j^* - \frac{1}{2}(V_j^*V_jx + x V_j^*V_j)\right],
    \end{equation}
where the jump operators $V_j$ are non-selfadjoint and non-unitary. Furthermore, we assume that the fixed point set of $T_t^*$ for any $t>0$ is given by the von Neumann algebra:
\begin{equation} \label{fixed point set: commutator}
    \mc N_{\rm fix} = \{V_j,V_j^*: 1\le j \le m\}' = \{x \in \mb B(\mc H): xV_j = V_j x, x V_j^* = V_j^*x\}.
\end{equation}
Note that \eqref{fixed point set: commutator} can be ensured by the assumption of Proposition \ref{characterization: fixed point}.

\subsubsection{A general construction of accessible unitaries.}
Motivated by \cite{cleve_et_al:LIPIcs.ICALP.2017.17}, a natural choice of accessible unitary with environment $E = \mb C^2 = \text{span}\{\ket{0}_E, \ket{1}_E\}$ is given by
\begin{equation} \label{accessible unitaries: with environment}
        \mc U_{A\otimes E} = \left\{\exp(it H_j): |t|\le \tau,\; H_j = V_j^* \otimes \ketbra{0}{1}_E + V_j \otimes \ketbra{1}{0}_E\right\}.
    \end{equation}
Note that $H_j$ is Hermitian by construction. Similar to Theorems \ref{main: symmetric} and \ref{main: Poisson}, we have the following theorem for the upper bound on the maximum uniform length $\widehat{\mm}_{\alpha,\beta}$:
\begin{prop}
   Suppose we have the encoding scheme $\mc E: \rho \mapsto \rho \otimes \ketbra{0}_E$, and decoding scheme via the partial trace $\mc D: \wt \rho \mapsto \Tr_E(\wt \rho)$, the maximum uniform length of order 2 is finite:
    \begin{equation}
        \widehat{\mm}_{\alpha,2} < \infty,\quad \forall \alpha >0.
    \end{equation}
\end{prop}

\begin{proof}

Note that $H_j$ is hermitian and by Lemma \ref{symmetric: approximation}, for each $1\le j \le m$, there exists $\Phi_t^j = \int_{\mc U_{A\otimes E}^{m_j}} \adm_{u_{m_j}\cdots u_1}d\mu_j(u_{m_j}\cdots u_1) $ with length $m_j$ such that
\begin{equation}
    \left\|\exp(tL_{H_j}) - \Phi_t^j\right\|_{\diamond} = O\left(\left\|V_j\right\|^6 t^3\right).
\end{equation}
Furthermore, by direct calculation, we get
\begin{equation}
    \mc D \circ L_{H_j} \circ \mc E = L_{V_j}.
\end{equation}
Thus via first order expansion $\|\exp(tL)-id_{\mb B(\mc H)} - tL\|_{\diamond} = O(\|L\|^2_{\diamond} t^2)$ (see Appendix of \cite{cleve_et_al:LIPIcs.ICALP.2017.17}), we have
\begin{equation}\label{eqn:exp_L}
    \left\|\exp(tL_{V_j}) - \mc D \circ \exp(tL_{H_j}) \circ \mc E\right\|_{\diamond} = O(\|V_j\|^4 \,t^2).
\end{equation}
Combining the above two equalities, we have
\begin{equation}
    \left\|\exp(tL_{V_j}) - \mc D \circ \Phi_t^j \circ \mc E\right\|_{\diamond} = O(\|V_j\|^4 \, t^2).
\end{equation}
Applying first-order Suzuki-Trotter formula, we have
\begin{equation}
      \left\|\exp\left(t\sum_j L_{V_j}\right) - \prod_j \mc D \circ \Phi_t^j \circ \mc E\right\|_{\diamond} = O\left(\left(\sum_j \|V_j\|^4\right)t^2\right)\,.
\end{equation}
This implies that $\prod_j \mc D \circ \Phi_t^j \circ \mc E$ is a second-order approximation scheme for $\exp\left(t\sum_j L_{V_j}\right)$ with finite circuit depth. The remaining steps are same as the proof of Theorem \ref{main: symmetric} in the unital case and we omit it.
\end{proof}

We set $E_{\rm fix}^A$ as the conditional expectation onto the fixed point algebra of $T_t^*$. To derive a lower bound on the simulation length, following the framework introduced in the previous section, we need to check the following three conditions:
\begin{enumerate}
    \item $T_t^* \circ E_{\rm fix}^A = E_{\rm fix}^A$ for some conditional expectation.
    \item $(\mc D \circ \adm_u \circ \mc E)^* \circ E_{\rm fix}^A  = E_{\rm fix}^A$ for any $u \in \mc U_{A\otimes E}$.
    \item $\mc U_{A\otimes E}$ and $S_{A\otimes E}$ are compatible, i.e., there exists $\widehat{D}>0$ such that for any $u \in \mc U_{A\otimes E}$,
    \begin{equation}
        \widehat{C}_{S_{A\otimes E}}^{cb}(\mc D \circ \adm_u \circ \mc E) \le \widehat{D}.
    \end{equation}
\end{enumerate}
We check them one by one:
\begin{enumerate}
\item This condition 1 is straightforward from the definition of $E^A_{\rm fix}$.

\item For the second condition, for any $u\in \mc U_{A\otimes E}$, we denote it as $$u = \exp(is (V_j^* \otimes \ketbra{0}{1}_E + V_j \otimes \ketbra{1}{0}_E)), \quad |s| \le \tau.$$
Recall that the encoding and decoding schemes are defined by 
$\mc E: \rho \mapsto \rho \otimes \ketbra{0}_E$, and $\mc D: \wt \rho \mapsto \Tr_E(\wt \rho)$. Hence by definition, the dual map is given explicitly as
\begin{equation}\label{dual: encode and decode}
   \begin{aligned}
    & \mc E^* : \mb B(AE) \to \mb B(A): X_{AE} \mapsto \bra{0}_E X_{AE} \ket{0}_E, \\
    & \mc D^* : \mb B(A) \to \mb B(AE):  X_A \mapsto X_A \otimes I_E.  
\end{aligned} 
\end{equation}
Then for any $X_A \in \mb B(A)$, we have 
\begin{align*}
    & (\mc D \circ \adm_u \circ \mc E)^* \circ E_{\rm fix}^A (X_A) = \mc E^* \circ \adm_{u^*} \circ \mc D^* (E_{\rm fix}^A(X_A)) \\
    & = \mc E^* \left(\exp(-is (V_j^* \otimes \ketbra{0}{1}_E + V_j \otimes \ketbra{1}{0}_E))  (E_{\rm fix}^A(X_A) \otimes I_E) \exp(is (V_j^* \otimes \ketbra{0}{1}_E + V_j \otimes \ketbra{1}{0}_E))\right) \\
& = \mc E^*(E_{\rm fix}^A(X_A) \otimes I_E)=  E_{\rm fix}^A(X_A).
\end{align*}
Note that apart from the definition, the main property we used is 
\begin{align*}
& \exp(is (V_j^* \otimes \ketbra{0}{1}_E + V_j \otimes \ketbra{1}{0}_E))  (E_{\rm fix}^A(X_A) \otimes I_E)\\
& = (E_{\rm fix}^A(X_A) \otimes I_E) \exp(is (V_j^* \otimes \ketbra{0}{1}_E + V_j \otimes \ketbra{1}{0}_E)).
\end{align*}
To see why the above is true, note that $E_{\rm fix}^A(X_A)$ is in the fixed point algebra of $T_t^*$ and thus commutes with $V_j,V_j^*$. Then using Taylor expansion and the fact that 
\begin{align*}
(V_j^* \otimes \ketbra{0}{1}_E + V_j \otimes \ketbra{1}{0}_E) (E_{\rm fix}^A(X_A) \otimes I_E) = (E_{\rm fix}^A(X_A) \otimes I_E)(V_j^* \otimes \ketbra{0}{1}_E + V_j \otimes \ketbra{1}{0}_E).
\end{align*}
we get the desired condition $(\mc D \circ \adm_u \circ \mc E)^* \circ E_{\rm fix}^A  = E_{\rm fix}^A$. 
\item Finally, to show that the last condition holds, we choose our resource set $S_{A\otimes E}$ as 
\begin{equation} \label{environment assisted resource set}
    S_{A\otimes E} = \{X_E \otimes I_A, Y_E \otimes I_A, H_j = V_j^* \otimes \ketbra{0}{1}_E + V_j \otimes \ketbra{1}{0}_E: 1\le j \le m\}.
\end{equation}
Then our claim is as follows:
\end{enumerate}
\begin{prop}\label{compatibility: non-unital}
    Suppose $S_{A\otimes E}$ is given as \eqref{environment assisted resource set}, then for any $u \in \mc U_{A\otimes E}$ with $\mc U_{A\otimes E}$ given in \eqref{accessible unitaries: with environment}, we have 
    \begin{equation} \label{D hat estimate}
        \widehat{C}_{S_{A\otimes E}}^{cb}(\mc D \circ \adm_u \circ \mc E) \le \tau + 6=: \widehat D
    \end{equation}
\end{prop}
\begin{proof}
For any $u \in \mc U_{A\otimes E}$, suppose $u = \exp(is H_j)$, then by definition
    \begin{align*}
& \widehat{C}^{cb}_{S_{A\otimes E}}(\mc D \circ \adm_u \circ \mc E)= \left\|(\mc E (\mc D \circ \adm_u \circ \mc E - id_A) \mc D)^*: (\mb B(A\otimes E), \||\cdot|\|_{S_{A\otimes E}}) \to \mb B(A\otimes E)\right\|_{cb} \\
& =\left\|(\mc E \circ \mc D)^*\, \adm_u^* \, (\mc E \circ \mc D)^* - (\mc E \circ \mc D)^* : (\mb B(A\otimes E), \||\cdot|\|_{S_{A\otimes E}}) \to \mb B(A\otimes E) \right\|_{cb} \\
& \le \left\|(\mc E \circ \mc D)^* \,\adm_u^*\, (\mc E \circ \mc D)^* - id: (\mb B(A\otimes E), \||\cdot|\|_{S_{A\otimes E}}) \to \mb B(A\otimes E)\right\|_{cb} \\
& +\left\|(\mc E \circ \mc D)^* - id: (\mb B(A\otimes E), \||\cdot|\|_{S_{A\otimes E}}) \to \mb B(A\otimes E)\right\|_{cb} \\
& = C_{S_{A\otimes E}}^{cb}((\mc E \circ \mc D) \circ \adm_u \circ (\mc E \circ \mc D)) + C_{S_{A\otimes E}}^{cb}(\mc E \circ \mc D) \\
& \le C_{S_{A\otimes E}}^{cb}(\adm_u) + 3 C_{S_{A\otimes E}}^{cb}(\mc E \circ \mc D),
\end{align*}
where we used the standard definition of $C_{S_{A\otimes E}}(\Phi)$ for a quantum channel $\Phi$ acting on operators on $A\otimes E$. Note that 
\begin{equation}
    C_{S_{A\otimes E}}^{cb}(\mc E \circ \mc D) = C_{S_{A\otimes E}}^{cb}(id_A \otimes E_0^{E*}) \le C_{S_E}^{cb}(E_0^{E*}),\quad S_E = \{X_E, Y_E\}
\end{equation}
where $E_0^E: \mb B(E) \to \mb B(E), X_E \mapsto \bra{0}_E X_E \ket{0}_E I_E$ and $E_0^{E*}$ is the replacer channel defined by $E_0^{E*}(\rho_E) = \Tr(\rho_E) \ketbra{0}_E$. Then using Example \ref{example: replacer}, we have $C_{S_E}^{cb}(E_0^{E*}) \le 2$ and using a similar argument as in Proposition \ref{Assumption: Trotter example}, we have $C_{S_{A\otimes E}}^{cb}(\adm_u) \le \tau$ thus we have $ \widehat{C}_{S_{A\otimes E}}^{cb}(\mc D \circ \adm_u \circ \mc E) \le \tau + 6=: \widehat D$.
\end{proof}

\begin{rmk}
Given a Lindbladian as in \eqref{model: non-symmetric and non-unitary}, the choice of the resource set is not unique. Another general way to construct $S_{A\otimes E}$ is given as follows:
    \[ S_A\lel \{V_j|1\le j\le m\}\cup \{V_j^*|1\le j\le m\} \pl , \]
 $S_E=\{X_E,Y_E\}$ which leads to the combined resource set 
  \[ S_{A\otimes E}\lel S_A\ten I_E\cup I_A\ten  S_E \pl. \]
It is clear that the resource set $S_A$ on the system $A$ has the property that 
\[  S_A' = \mc N_{\rm fix} \implies E_{S_A'}   \lel  E_{\rm fix}^A\] 
which is consistent with the assumption in Section \ref{sec:lower_bound}. Moreover, we notice that,
\begin{align*}
     H_j =  V_j^* \otimes \ketbra{0}{1}_E + V_j \otimes \ketbra{1}{0}_E = V_j \ten \frac{X-iY}{2} + V_j \ten \frac{X+iY}{2}
\end{align*}
Similar to the proof of Example \ref{example: replacer}, we can show that if one resource set can be represented as linear combinations of products of the other resource set, then they are comparable. This implies $\widehat{C}_{S_{A\otimes E}}^{cb}$ are comparable for the above two choices of the resourse set. 
\end{rmk}

\noindent The summary of the lower bound for the model with Lindbladian as in \eqref{model: non-symmetric and non-unitary} is stated as follows:
\begin{theorem}\label{main: lower bound non-unital}
    Suppose $T_t = \exp(tL)$ with $L$ given by \eqref{model: non-symmetric and non-unitary}. Additionally, we assume $E_{\rm fix}^A$ is the conditional expectation onto the fixed point algebra of $T_t^*$. Then for $E = \mb C^2$, the set of accessible unitaries
    \begin{equation}
        \mc U_{A\otimes E} = \{\exp(it H_j): |t|\le \tau,\; H_j = V_j^* \otimes \ketbra{0}{1}_E + V_j \otimes \ketbra{1}{0}_E\}.
    \end{equation}
    and the encoding and decoding schemes given by $\mc E: \rho \mapsto \rho \otimes \ketbra{0}_E$, $\mc D: \wt \rho \mapsto \Tr_E(\wt \rho)$, we have the following lower bound: 
 
   \begin{equation}
      \widehat{\mm}_{\alpha,\beta} \ge C_{\alpha,\beta} \frac{\widehat{C}^{cb}_{S_{A\otimes E}}\left(\left(E_{\rm fix}^A\right)^*\right)}{\widehat D},\quad C_{\alpha,\beta}:= \min \left\{ \alpha^{-\frac{1}{\beta -1}} \left( \frac{1-\varepsilon}{3t_{\rm mix}^{S_{A\otimes E}}(\varepsilon)} \right)^{\frac{\beta}{\beta - 1}} , \frac{1-\varepsilon}{2} \right\}.
   \end{equation}
    where $\widehat{D}$ can be given by $\tau+6$ in \eqref{D hat estimate}.
\end{theorem}

\begin{rmk} \label{iterate} We have seen in \eqref{ff} that if  $\frac{C_{S}^{cb}(E_{\infty})}{t_{\rm mix}}$ is large, our estimates are meaningful. For the tensor product model, we are able to show that this can occur indeed. This trick works for unital and non-unital quantum dynamics and we consider the non-unital case for generality. For example, consider the new generator is given by $L^{(n)}=\sum_{j=1}^n \pi_j(L)$ where $\pi_j(L)$ is acting on the $j$-th register and $L$ is a fixed Lindblad generator. Then we have (see the details in Appendix \ref{appendix: tensor product})
 \[  \frac{\widehat{C}^{cb}_{S_{A^{\otimes n}\otimes E}}\left(\left(E_{\rm fix}^A\right)^{* \otimes n}\right)}{t_{\rm mix}(L^{(n)})} \gl \Omega \left(n \frac{\widehat{C}^{cb}_{S_{A\otimes E}}\left(\left(E_{\rm fix}^A\right)^*\right)}{t_{\rm mix}(L) }\right) \pl.   \]
Thus, we obtain a lower bound of order $n$ for the simulation cost of tensor product models. 
\end{rmk}

\subsubsection{A different construction: $n$-qubit amplitude damping noise}\label{sec:nqadn}
In the previous subsection, we present a general framework that works for a large class of non-unitary and non-Hermitian jump operators. In this subsection, we present a different construction of accessible unitaries which work for $n$-qubit amplitude damping noise. 
The model is defined on $n$-qubit system $A^{\otimes n} = (\mb C^2)^{\otimes n}$. Denote $a_j$ as the operator $a = \begin{pmatrix}
    0 & 1 \\
    0 & 0
\end{pmatrix}$ acting on the $j$-th system, i.e., 
\begin{align*}
    a_j = I_2 \otimes \cdots I_2 \otimes \underbrace{a}_\text{j-th system} \otimes I_2 \cdots \otimes I_2.
\end{align*}
Denote \begin{align*}
    L(x) = \sum_{j=1}^n L_{a_j}(x) = \sum_{j=1}^n \big(a_jxa_j^* - \frac{1}{2}(a_j^*a_j x + xa_j^*a_j)\big), \quad x \in \mb B(A),
\end{align*}
we have by the commuting property,
\begin{equation}
    T_t = \exp(tL) = \mc N_t^{\otimes n}, 
\end{equation}
where $\mc N_t = \exp(t \mc L_a)$ is acting on a single qubit system $A = \mb C^2$. The accessible unitary here is given by
\begin{equation}
    \mc U_{A^n\otimes E} = \{\exp(it (Y_j\otimes X_E)), \exp(it ( X_j\otimes Y_E)): |t |\le \tau, 1\le j \le n\}.
\end{equation}
and the standard encoding and decoding schemes given by $\mc E: \rho \mapsto \rho \otimes \ketbra{0}_E$, $\mc D: \wt \rho \mapsto \Tr_E(\wt \rho)$. Then, we have 
\begin{equation}\label{Simulation: amplitude damping}
    \mc N_t^{\otimes n} = \prod_{j=1}^n \mc D \circ \adm_{\exp(\frac{i \arccos (e^{-t})}{2}Y_A \otimes X_E )} \adm_{\exp(\frac{-i \arccos (e^{-t})}{2} X_A \otimes Y_{E} )}\circ \mc E,\quad \forall t \ge 0.
\end{equation}
Details on this elementary calculation can be found in  appendix~\ref{app:hamsim-for-ampdamp-noise}. 
Note that $0\le \arccos (e^{-t}) \le \frac{\pi}{2}$, which allows us to obtain the upper bound 
\begin{equation}
   \widehat{\mm}_{\alpha,\beta} \le \frac{\pi n}{\tau},\quad \forall \alpha>0,\beta >1.
\end{equation}
To show the lower bound, we choose the resource set  
\begin{equation}
    S_{A^nE} = \{I_{A^{\otimes n}} \otimes X_E, I_{A^{\otimes n}} \otimes Y_E,  X_j\otimes Y_E, Y_j\otimes X_E: 1\le j \le n\}
\end{equation}
and the conditional expectation 
\begin{equation}
    E_{\rm fix}^{A^n}: \mb B(A^{\otimes n}) \to \mb B(A^{\otimes n}), \rho_n \mapsto \bra{0^n} \rho_n \ket{0^n} I_{A^{\otimes n}}.
\end{equation}
Our goal is to show the following conditions. 
\begin{enumerate}
    \item $T_t^* \circ E_{\rm fix}^{A^n} = E_{\rm fix}^{A^n}$ for some conditional expectation.
    \item $(\mc D \circ \adm_u \circ \mc E)^* \circ E_{\rm fix}^{A^n}  = E_{\rm fix}^{A^n}$ for any $u \in \mc U_{{A^n}\otimes E}$.
    \item $\mc U_{{A^n}\otimes E}$ and $S_{{A^n}\otimes E}$ are compatible, i.e., there exists $\widehat{D}>0$ such that for any $u \in \mc U_{{A^n}\otimes E}$,
    \begin{equation}
        \widehat{C}_{S_{{A^n}\otimes E}}^{cb}(\mc D \circ \adm_u \circ \mc E) \le \widehat{D}.
    \end{equation}
\end{enumerate}
The first and second condition are obvious because $T_t^*$ and $(\mc D \circ \adm_u \circ \mc E)^*$ are unital. For the last condition, using the same argument as Proposition \ref{compatibility: non-unital}, we have 
\begin{equation}
    \widehat{C}_{S_{A\otimes E}}^{cb}(\mc D \circ \adm_u \circ \mc E) \le \tau + 6=: \widehat{D},\quad \forall u\in \mc U_{A\otimes E}.
\end{equation}
Applying Theorem \ref{lower bound: non-unital}, we can get a lower bound for $\widehat{\mm}_{\alpha,\beta}$: 
\begin{equation}
 \widehat{\mm}_{\alpha,\beta} \ge C_{\alpha,\beta} \frac{\widehat{C}^{cb}_{S_{A\otimes E}}\left(\left(E_{\rm fix}^{A^n}\right)^*\right)}{\widehat D},\quad C_{\alpha,\beta}:= \min \left\{ \alpha^{-\frac{1}{\beta -1}} \left( \frac{1-\varepsilon}{3t_{\rm mix}^{S_{A\otimes E}}(\varepsilon)} \right)^{\frac{\beta}{\beta - 1}} , \frac{1-\varepsilon}{2} \right\}.
\end{equation}
Finally, we provide estimates of $\widehat{C}^{cb}_{S_{A\otimes E}}\left( \left(E_{\rm fix}^{A^n}\right)^*\right)$ and $t_{\rm mix}^{S_{A\otimes E}}(\varepsilon)$. In fact, we choose the certificate operator $\sum_{j=1}^n X_j \otimes I_E$ and recall the definition of $\widehat{C}_{S_{A\otimes E}}$:
\begin{align*}
    \widehat{C}_{S_{A\otimes E}}\left( \left(E_{\rm fix}^{A^n}\right)^*\right) & = \sup_{X_{AE} \in \mb B(AE), \||X_{AE}|\|_{S_{A\otimes E}} \neq 0} \frac{\|(\mc E ((E_{\rm fix}^{A^n})^* - id_{A^n}) \mc D)^*(X_{AE})\|_{\infty}}{\||X_{AE}|\|_{S_{A\otimes E}}} \\
    & = \sup_{X_{AE} \in \mb B(AE), \||X_{AE}|\|_{S_{A\otimes E}} \neq 0} \frac{\|\mc D^* (E_{\rm fix}^{A^n} - id_{A^n}) \mc E^*(X_{AE})\|_{\infty}}{\||X_{AE}|\|_{S_{A\otimes E}}} \\
    & \ge \frac{\|\mc D^* (E_{\rm fix}^{A^n} - id_{A^n})(\sum_{j=1}^n X_j)\|}{\||\sum_{j=1}^n X_j \otimes I_E|\|_{S_{A\otimes E}}} = \frac{\|\sum_{j=1}^n X_j\|_{\infty}}{\sup_k \|[Y_k,\sum_{j=1}^n X_j]\otimes X_E\|_{\infty}} = \frac{n}{2}.
\end{align*}
This directly implies $\widehat{C}^{cb}_{S_{A\otimes E}}\left( \left(E_{\rm fix}^{A^n}\right)^*\right)\geq \widehat{C}_{S_{A\otimes E}}\left( \left(E_{\rm fix}^{A^n}\right)^*\right)\geq \frac{n}{2}$. Similarly, we can check that 
\begin{align*}
    \widehat{C}^{cb}_{S_{A\otimes E}}(T_t)&\geq \widehat{C}_{S_{A\otimes E}}(T_t) = \sup_{X_{AE} \in \mb B(AE), \||X_{AE}|\|_{S_{A\otimes E}} \neq 0} \frac{\|(\mc E (T_t - id_{A^n}) \mc D)^*(X_{AE})\|_{\infty}}{\||X_{AE}|\|_{S_{A\otimes E}}} \\
    & \ge \frac{\|\mc D^* (T_t^* - id_{A^n})(\sum_{j=1}^n X_j)\|}{\||\sum_{j=1}^n X_j \otimes I_E|\|_{S_{A\otimes E}}} = (1-e^{-t})\frac{\|\sum_{j=1}^n X_j\|}{\||\sum_{j=1}^n X_j \otimes I_E|\|_{S_{A\otimes E}}} = \frac{(1-e^{-t})n}{2}, 
\end{align*}
thus we can show that $t_{\rm mix}^{S_{A\otimes E}}$ is bounded by a universal constant. Using Theorem \ref{lower bound: non-unital}, we find a tight lower bound for $\widehat{\mm}_{\alpha,\beta}$ in terms of the number of qubits $n$:
\begin{equation}
    C_{\alpha,\beta} \frac{n}{\tau+6}\le \widehat{\mm}_{\alpha,\beta} \le \frac{\pi n}{\tau}.
\end{equation}
\begin{rmk}
    Using Theorem \ref{main: lower bound non-unital fixed time}, we can get the desirable behavior of linear dependency$$\widehat{l}_{\delta}^{\mc U_{A \ten E}}(T_t) = \Theta(nt).$$
\end{rmk}

\section{Conclusion}\label{sec:conclusion}
Our technique for the lower bound of simulation cost of open quantum systems is based on the convex cost function defined in Section \ref{sec:lower_bound}. The linear in time behavior of the lower bounds that already have been observed in \cite{berry2015hamiltonian, childs2016efficient} is captured in our Theorem \ref{main result: lower bound fixed time}. Moreover, our dimension factor $\frac{\kappa^{cb}(S)}{t_{\rm{mix}}^S}$ exhibits a new phenomenon, which works for a more general class of generators than presented in previous works. Explicit estimates of the dimension factor are given for $n$ qubit Pauli-type and amplitude damping noise and the tightness of our lower bound is illustrated by showing that the upper and lower bounds coincide in those two examples. We leave the concrete analysis for more complicated generators for future research.

\section{Acknowledgements}
The authors thank Tyler Kharazi for helpful discussions. Z.D., M.J. and P.S. thank the Institute for Pure and Applied Mathematics (IPAM) for its hospitality in hosting them as long-term visitors during the semester-long program “Mathematical and Computational Challenges in  Quantum Computing” in Fall 2023. 

\begin{appendices}
    
\section{Simulation scheme for tensor product model}\label{appendix: tensor product}
In the appendix, we discuss the limitation of a natural simulation scheme for $n$ copies of a fixed quantum Markov semigroup. To be more specific, if we have an accessible unitary $\mc U_{A\otimes E}$ together with an encoding $\mc E: \rho \mapsto \rho \otimes \ketbra{0^{k_E}}_E$ and a decoding $\mc D = \Tr_E$ which can efficiently simulate $\mc T_t = e^{tL}$, then we use $\mc U_{A^{\otimes n} \otimes E} = \{\pi_j(u_{AE}): u_{AE} \in \mc U_{A\otimes E}, 1\le j \le n\}$ where $\pi_j(u_{AE})$ is an operator acting on $A^{\otimes n} \otimes E$, only nontrivially on $j$-th register of $A$ and $E$ by $u_{AE}$, together with an encoding $\mc E_n: \rho_n \mapsto \rho_n \otimes \ketbra{0^{k_E}}_E$ and a decoding $\mc D_n = \Tr_E$, to simulate $\mc T_t^{\otimes n}$. Our framework in Section \ref{sec:non-unital} can confirm our naive intuition that we need at least $\Omega(n)$ gate length to simulate $\mc T_t^{\otimes n}$. 
\begin{figure}[h!]
    \centering
    \includegraphics[width=.5\textwidth]{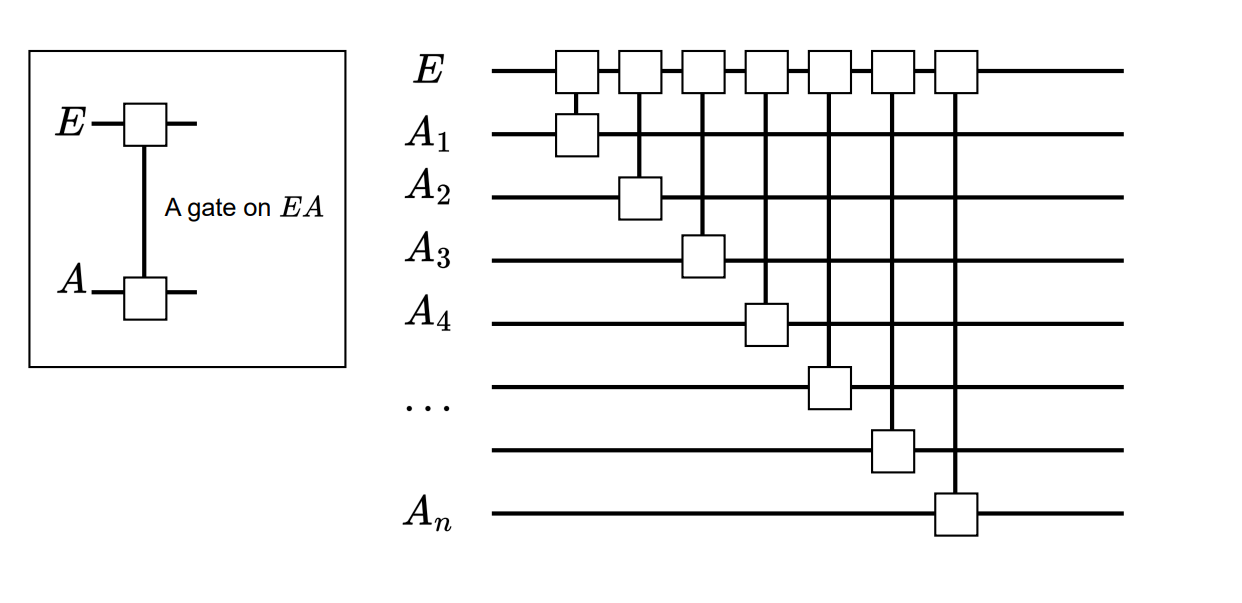}
    \caption{Using a natural amplified accessible gate set $\mc U_{A^{\otimes n} \otimes E} = \{\pi_j(u_{AE}): u_{AE} \in \mc U_{A\otimes E}, 1\le j \le n\}$ to approximate $\mc T_t^{\otimes n}$.}
\end{figure}

In fact, suppose we have a resource set $S_{A\otimes E}$ and conditional expectation $E_{\rm{fix}}^A$ onto the fixed point algebra of $\mc T_t^*$, such that the assumption of Theorem \ref{main: lower bound non-unital fixed time} holds. Then applying the same theorem with the resource set $S_{A^{\otimes n}\otimes E} = \{\pi_j(s_{AE}): s_{AE} \in S_{A\otimes E}, 1\le j \le n\}$, conditional expectation $(E_{\rm{fix}}^A)^{\otimes n}$ we can get a lower bound on the simulation cost of $\mc T_t^{\otimes n}$ as follows:
\begin{align*}
    \Omega \bigg(t \frac{\widehat{C}^{cb}_{S_{A^{\otimes n}\otimes E}}\left( \left( E_{\rm fix}^{A*}\right)^{\otimes n} \right)}{t_{\rm mix}^{S_{A^{\otimes n}\otimes E}}} \bigg).
\end{align*}
First, note that by additivity of $\widehat{C}^{cb}_{S_{A^{\otimes n}\otimes E}}$, we get $\widehat{C}^{cb}_{S_{A^{\otimes n}\otimes E}}\left( \left( E_{\rm fix}^{A*}\right)^{\otimes n} \right) = n \widehat{C}^{cb}_{S_{A\otimes E}}\left( E_{\rm fix}^{A*}\right)$. Moreover, we claim that for $\mc T_t^{\otimes n} = e^{t L^{(n)}}$,
\begin{align*}
    & t_{\rm mix}^{S_{A^{\otimes n}\otimes E}}(\varepsilon,L^{(n)}):= \inf\{t \ge 0: \widehat{C}^{cb}_{S_{A^{\otimes n}\otimes E}}\left( e^{t L^{(n)}} \right) \ge (1-\varepsilon)\widehat{C}^{cb}_{S_{A^{\otimes n}\otimes E}}\left( \left( E_{\rm fix}^{A*}\right)^{\otimes n} \right)\} \\
    & \le t_{\rm mix}^{S_{A\otimes E}}(\varepsilon,L) = \inf\{t\ge 0: \widehat{C}^{cb}_{S_{A\otimes E}}\left( \mc T_t\right) \ge (1-\varepsilon) \widehat{C}^{cb}_{S_{A\otimes E}}\left( E_{\rm fix}^{A*}\right)\}.
\end{align*}
In fact, suppose $t_0$ is such that $\widehat{C}^{cb}_{S_{A\otimes E}}\left( \mc T_{t_0}\right) \ge (1-\varepsilon) \widehat{C}^{cb}_{S_{A\otimes E}}\left( E_{\rm fix}^{A*}\right)$, then for $\delta >0$ arbitrary, we choose a certificate operator $W_{VAE}$ with $\||W_{VAE}|\|_{I_V \otimes S_{A\otimes E}} \le 1$ acting on some ancillary system $V$ and the composite system $A\otimes E$ such that 
\begin{align*}
    \left\|id_{\mb B(V)} \otimes \left(\mc E \circ (\mc T_{t_0} - id_{\mb B(A)}) \circ \mc D \right)^*(W_{VAE})\right \| \ge (1-\varepsilon) \widehat{C}^{cb}_{S_{A\otimes E}}\left( E_{\rm fix}^{A*}\right) - \delta.
\end{align*}
The certificate operator acting on the whole system is chosen as 
\begin{align*}
    W_{V^{\otimes n}A^{\otimes n}E} = \sum_{j=1}^n W_{V_j A_j E}.
\end{align*}
\begin{figure}[h!]
    \centering
    \includegraphics[width=.6\textwidth]{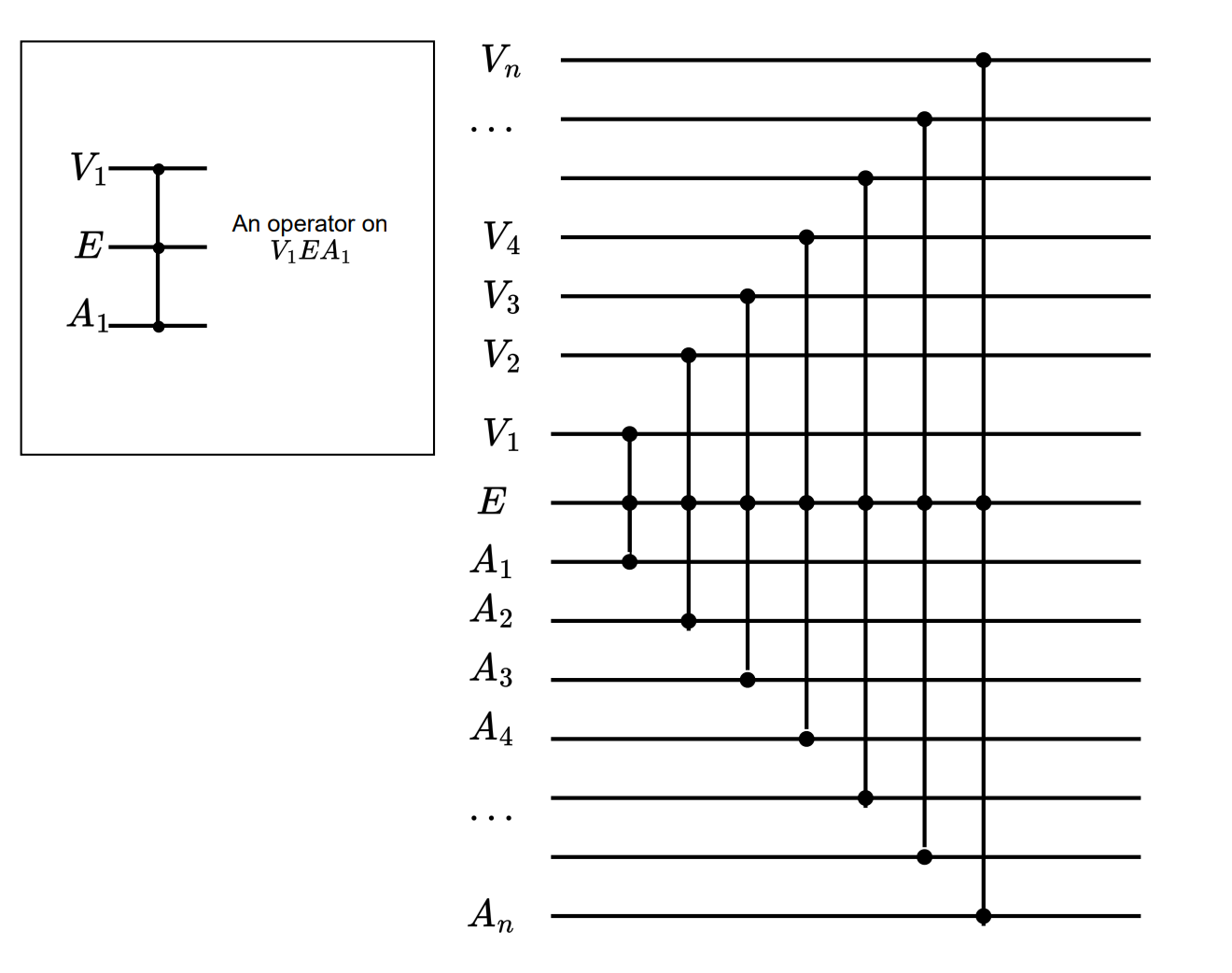}
    \caption{Certificate operator on the whole system is given by the sum of the certificate operator on $VAE$.}
\end{figure}

\noindent Then we have 
\begin{align*}
    & \left\|id_{\mb B(V^{\otimes n})} \otimes \left(\mc E_n \circ (\mc T_{t_0}^{\otimes n} - id_{\mb B(A^{\otimes})}) \circ \mc D_n \right)^*(W_{V^{\otimes n}A^{\otimes n}E})\right\| \\
    & = \left\|\sum_{j=1}^n \bigg(id_{\mb B(V)} \otimes \left(\mc E \circ (\mc T_{t_0} - id_{\mb B(A)}) \circ \mc D \right)^*(W_{VAE})\bigg)_j\right\| \\
    & = n \left\|id_{\mb B(V)} \otimes \left(\mc E \circ (\mc T_{t_0} - id_{\mb B(A)}) \circ \mc D \right)^*(W_{VAE})\right\| \ge (1-\varepsilon) n\widehat{C}^{cb}_{S_{A\otimes E}}\left( (E_{\rm fix}^{A})^*\right) - n\delta.
\end{align*}
Since $\delta >0$ is arbitrary, we have 
\begin{align*}
    \widehat{C}^{cb}_{S_{A^{\otimes n}\otimes E}}\left( \mc T_{t_0}^{\otimes n} \right) \ge (1-\varepsilon)\widehat{C}^{cb}_{S_{A^{\otimes n}\otimes E}}\left( \left( E_{\rm fix}^{A*}\right)^{\otimes n} \right),
\end{align*}
which shows $t_{\rm mix}^{S_{A^{\otimes n}\otimes E}}(\varepsilon,L^{(n)}) \le t_0 \to t_{\rm mix}^{S_{A\otimes E}}(\varepsilon,L)$. In summary, our lower bound is given by
\begin{align*}
    \Omega \bigg(nt\frac{\widehat{C}^{cb}_{S_{A\otimes E}}\left( (E_{\rm fix}^{A})^*\right)}{t_{\rm mix}^{S_{A\otimes E}}} \bigg),
\end{align*}
which is given by $n$ times the original lower bound for simulating $\mc T_t$. We remark here that $$t_{\rm mix}^{S_{A^{\otimes n}\otimes E}}(\varepsilon,L^{(n)}) \le t_{\rm mix}^{S_{A\otimes E}}(\varepsilon,L)$$
is a new feature that the classical mixing time induced by the diamond norm does not have. It is shown in \cite{gao2022complete} that in a lot of cases, we have $t_{\rm{mix}}(\varepsilon, L^{(n)}) \le \mc O(\log n) \cdot t_{\rm{mix}}(\varepsilon, L)$.

\subsection{Hamiltonian simulation for n-qubit amplitude damping noise}\label{app:hamsim-for-ampdamp-noise}
In this section, we present the missing calculation about the amplitude damping channel $\mc N_t$ in \eqref{Simulation: amplitude damping}. To be more specific, We need to find the explicit expression of $U_t$ such that $\mc N_t(\rho) = \Tr_E(U_t ( \rho \otimes \ketbra{0}_E) U_t^*)$, where the amplitude damping noise is given by
\begin{equation}\label{def:amplitude damping noise}
    \mc N_t\begin{pmatrix}
        \rho_{00} & \rho_{01} \\
        \rho_{10} & \rho_{11}
    \end{pmatrix} = \begin{pmatrix}
        \rho_{00} + (1-e^{-2t}) \rho_{11} & e^{-t}\rho_{01} \\
        e^{-t}\rho_{10} & e^{-2t}\rho_{11}
    \end{pmatrix}.
\end{equation}
In terms of Stinespring dilation, $\mc N_t$ can be viewed as a coupling with a zero-temperature external bath $\ket{0}_E$, i.e., $$\mc N_t(\rho) = \Tr_E(V_t ( \rho \otimes \ketbra{0}_E) V_t^*), $$
with the partial isometry $V_t$ so that
\begin{equation}
    \begin{aligned}
        V_t:\;\; & \ket{0}_A \otimes \ket{0}_E \mapsto \ket{0}_A \otimes \ket{0}_E, \\
        & \ket{1}_A \otimes \ket{0}_E \mapsto \sqrt{1-e^{-2t}}\ket{0}_A \otimes \ket{1}_E + e^{-t} \ket{1}_A \otimes \ket{0}_E.
    \end{aligned}
\end{equation}
To find a suitable Hamiltonian to simulate $V_t$, we properly extend $V_t$ to a full unitary $U_t$ acting on the joint system $A \otimes E$. A specific choice is given as 
\begin{equation}\label{dilation unitary: amplitude damping}
    \begin{aligned}
        U_t:\;\; & \ket{0}_A \otimes \ket{0}_E \mapsto \ket{0}_A \otimes \ket{0}_E, \\
        & \ket{1}_A \otimes \ket{0}_E \mapsto \sqrt{1-e^{-2t}}\ket{0}_A \otimes \ket{1}_E + e^{-t} \ket{1}_A \otimes \ket{0}_E, \\
        & \ket{0}_A \otimes \ket{1}_E \mapsto e^{-t}\ket{0}_A \otimes \ket{1}_E - \sqrt{1-e^{-2t}} \ket{1}_A \otimes \ket{0}_E, \\
        & \ket{1}_A \otimes \ket{1}_E \mapsto \ket{1}_A \otimes \ket{1}_E.
    \end{aligned}
\end{equation}
We claim that we have
\begin{equation}
    U_t = \exp(-i \arccos (e^{-t}) H), \quad  H = \frac{1}{2}(X_A \otimes Y_{E} - Y_A \otimes X_E),
\end{equation}
where $X_A,Y_A,X_E,Y_E$ are the Pauli-$X,Y$ operators on the corresponding system. One can check directly that for any real number $\theta$,
\begin{align*}
  i\theta H(\ket{1}_A \otimes \ket{0}_E)  = \theta \ket{0}_A \otimes \ket{1}_E, \quad i\theta H(\ket{0}_A \otimes \ket{1}_E)  = -\theta \ket{1}_A \otimes \ket{0}_E
\end{align*}
Then by induction, we get for any $k\ge 0$ and real number $\theta$, 
\begin{equation}
    \begin{aligned}
        (i\theta H)^{2k}: & \ket{1}_A \otimes \ket{0}_E \mapsto \theta^{2k}(-1)^k \ket{1}_A \otimes \ket{0}_E, \\
        & \ket{0}_A \otimes \ket{1}_E \mapsto \theta^{2k}(-1)^k \ket{0}_A \otimes \ket{1}_E.\\
        (i\theta H)^{2k+1}: & \ket{1}_A \otimes \ket{0}_E \mapsto \theta^{2k+1}(-1)^{k+1} \ket{0}_A \otimes \ket{1}_E,\\
        & \ket{0}_A \otimes \ket{1}_E \mapsto \theta^{2k+1}(-1)^k \ket{1}_A \otimes \ket{0}_E.
    \end{aligned}
\end{equation}
Also note that $(i\theta H)^{k}(\ket{0}_A \otimes \ket{0}_E ) = (i\theta H)^{k}(\ket{1}_A \otimes \ket{1}_E ) = \textbf{0}$(zero vector) for $k\ge 1$. By Taylor expansion, we have 
\begin{equation}
    \begin{aligned}
       \exp(i\theta H): & \ket{0}_A \otimes \ket{0}_E \mapsto \ket{0}_A \otimes \ket{0}_E, \\
        & \ket{1}_A \otimes \ket{0}_E \mapsto \sum_{k\ge 0} \frac{\theta^{2k+1} (-1)^{k+1}}{(2k+1)!}\ket{0}_A \otimes \ket{1}_E + \sum_{k\ge 0} \frac{\theta^{2k} (-1)^k}{(2k)!} \ket{1}_A \otimes \ket{0}_E, \\
        & \ket{0}_A \otimes \ket{1}_E \mapsto \sum_{k\ge 0} \frac{\theta^{2k} (-1)^k}{(2k)!}\ket{0}_A \otimes \ket{1}_E + \sum_{k\ge 0} \frac{\theta^{2k+1} (-1)^{k}}{(2k+1)!} \ket{1}_A \otimes \ket{0}_E, \\
        & \ket{1}_A \otimes \ket{1}_E \mapsto \ket{1}_A \otimes \ket{1}_E.
    \end{aligned}
\end{equation}
Recall that 
\begin{equation}
    \cos \theta = \sum_{k\ge 0} \frac{\theta^{2k} (-1)^k}{(2k)!},\quad \sin \theta = \sum_{k\ge 0} \frac{\theta^{2k+1} (-1)^{k}}{(2k+1)!}, 
\end{equation}
we have 
\begin{align*}
    & \exp(i\theta H) (\ket{1}_A \otimes \ket{0}_E) =- \sin \theta \ket{0}_A \otimes \ket{1}_E + \cos \theta \ket{1}_A \otimes \ket{0}_E , \\
    &  \exp(i\theta H) (\ket{0}_A \otimes \ket{1}_E) = \cos \theta \ket{0}_A \otimes \ket{1}_E + \sin \theta \ket{1}_A \otimes \ket{0}_E.
\end{align*}
Comparing the expression in \eqref{dilation unitary: amplitude damping}, we choose $\theta = -\arccos (e^{-t}) \in [-\frac{\pi}{2}, 0]$, and we get 
\begin{equation}
    U_t = \exp(-i \arccos (e^{-t}) H).
\end{equation}
Finally note that $[X_A \otimes Y_{E}, Y_A \otimes X_E]=0$, we have 
\begin{equation}
    \adm_{U_t} = \adm_{\exp(\frac{i \arccos (e^{-t})}{2}Y_A \otimes X_E )} \adm_{\exp(\frac{-i \arccos (e^{-t})}{2} X_A \otimes Y_{E} )}.
\end{equation}
\subsection{Hamiltonian simulation for n-qubit Pauli noise using one ancilla qubit}\label{app:hamsim-for-pauli-noise}
In this section, we show that using the calculation of the previous section, we can construction a perfect simulation of the Pauli-noise on the qubit system $A = \mc H_2$: 
\begin{equation}
    \mc T_t = \exp(t (L_X + L_Y)),\quad (L_X+L_Y)(\rho):= X\rho X + Y\rho Y - 2\rho.
\end{equation}
In fact, denote $\rho = \begin{pmatrix}
    \rho_{00} & \rho_{01} \\
    \rho_{10} & \rho_{11}
\end{pmatrix}$, a calculation based on Taylor expansion shows that 
\begin{equation}
\begin{aligned}
   \mc T_t\begin{pmatrix}
    \rho_{00} & \rho_{01} \\
    \rho_{10} & \rho_{11}
\end{pmatrix}& = \begin{pmatrix}
    \frac{1+ e^{-4t}}{2} \rho_{00} + \frac{1- e^{-4t}}{2} \rho_{11} & e^{-2t}\rho_{01} \\
    e^{-2t}\rho_{10} & \frac{1- e^{-4t}}{2} \rho_{00} + \frac{1+ e^{-4t}}{2} \rho_{11}
\end{pmatrix} \\
& = \frac{1}{2} \begin{pmatrix}
     \rho_{00} + (1- e^{-4t}) \rho_{11} & e^{-2t}\rho_{01} \\
    e^{-2t}\rho_{10} & e^{-4t} \rho_{11}
\end{pmatrix} + \frac{1}{2}\begin{pmatrix}
     e^{-4t}\rho_{00}  & e^{-2t}\rho_{01} \\
    e^{-2t}\rho_{10} & (1- e^{-4t}) \rho_{00} + \rho_{11}
\end{pmatrix} \\
& = \frac{1}{2}(\mc N_{2t}(\rho) + \widehat{\mc N}_{2t}(\rho)),
\end{aligned}
\end{equation}
where $\mc N_t$ is the amplitude damping noise defined in \eqref{def:amplitude damping noise} and $\widehat{\mc N}_t$ is defined by 
\begin{equation}
    \widehat{\mc N}_t\begin{pmatrix}
        \rho_{00} & \rho_{01} \\
        \rho_{10} & \rho_{11}
    \end{pmatrix} = \begin{pmatrix}
        e^{-2t}\rho_{00} & e^{-t}\rho_{01} \\
        e^{-t}\rho_{10} & (1-e^{-2t})\rho_{00} + \rho_{11}
    \end{pmatrix}.
\end{equation}
Using exactly the same argument as the previous section, we have 
\begin{equation}
    \widehat{\mc N}_t(\rho) = \Tr_E(\widehat{U}_t ( \rho \otimes \ketbra{0}_E) \widehat{U}_t^*),\quad \widehat{U}_t = \exp(-i \frac{\arccos(e^{-t})}{2} (X_A \otimes Y_E + Y_A \otimes X_E)).
\end{equation}
Therefore, for the encoding $\mc E: \rho \mapsto \rho \otimes \ketbra{0}_E$ and decoding $\mc D: \rho_{AE} \mapsto \rho_A$, we have 
\begin{equation}
    \mc T_t = \mc D \circ \frac{1}{2}\left(\adm_{U_1(2t)} \adm_{U_2^*(2t)} + \adm_{U_1^*(2t)} \adm_{U_2^*(2t)}\right) \circ \mc E,
\end{equation}
where \begin{equation}
    U_1(t) = \exp(\frac{i \arccos (e^{-t})}{2}Y_A \otimes X_E ), \quad U_2(t)= \exp(\frac{i \arccos (e^{-t})}{2}X_A \otimes Y_E ).
\end{equation}
\end{appendices}
\bibliography{sim}
\end{document}